\documentclass[sigconf, authorversion]{acmart}
\AtBeginDocument{%
  }

\copyrightyear{2026}
\acmYear{2026}
\setcopyright{cc}
\setcctype{by}
\acmConference[CCS '26] {Proceedings of the 2026 ACM SIGSAC Conference on Computer and Communications Security}{November 15--19, 2026}{The Hague, Netherlands.}
\acmBooktitle{Proceedings of the 2026 ACM SIGSAC Conference on Computer and Communications Security (CCS '26), November 15--19, 2026, The Hague, Netherlands}
\acmISBN{979-8-4007-2871-6/2026/11}
\acmDOI{10.1145/3830454.3846666}

\usepackage{amsmath,amsthm,amsfonts,bm,bbm}
\usepackage{multirow}
\usepackage{longtable}
\usepackage{mdframed}
\usepackage{algorithm}

\usepackage{booktabs}
\usepackage{tabularx}
\usepackage{makecell}
\usepackage{tablefootnote}

\usepackage{soul}
\soulregister\ref7
\soulregister\pageref7
\soulregister\eqref7
\soulregister\cite7
\soulregister\autoref7
\soulregister\cref7

\usepackage[noend]{algpseudocode}

\makeatletter
\newcommand*{\rom}[1]{\expandafter\@slowromancap\romannumeral #1@}

\newtheorem{theorem}{Theorem}

\newtheorem{remark}{Remark}

\newtheorem{proposition}{Proposition}
\newtheorem{definition}{Definition}
\renewcommand{\v}[1]{\ensuremath{\v{#1}}}

\def \v x{\bm x}

\def \v x{\bm X}

\renewcommand{\v}[1]{\ensuremath{\boldsymbol{#1}}}

\definecolor{matchgreen}{RGB}{34, 139, 34} 
\definecolor{diffred}{RGB}{220, 20, 60}    

\newcommand{\red}[1]{{\color{black}#1}}

\begin{document}

\title{When Topology Betrays Privacy: Lattice-Based Reconstruction\\ Attacks on Secure Aggregation in Decentralized Federated Learning}



\author{Wenrui Yu}
\affiliation{%
  \institution{Aalborg University}
  \city{Copenhagen}
  \country{Denmark}
}
\email{wenyu@es.aau.dk}

\author{Changlong Ji}
\affiliation{%
  \institution{Institut Polytechnique de Paris}
  \city{Paris}
  \country{France}}
\email{changlong.ji98@gmail.com}

\author{Johannes Bjerva}
\affiliation{%
  \institution{Aalborg University}
  \city{Copenhagen}
  \country{Denmark}
}
\email{jbjerva@cs.aau.dk}

\author{Qiongxiu Li}
\authornote{Corresponding author.}
\affiliation{%
  \institution{Aalborg University}
  \city{Copenhagen}
  \country{Denmark}
}
\email{qili@es.aau.dk}

\renewcommand{\shortauthors}{Wenrui Yu, Changlong Ji, Johannes Bjerva, and Qiongxiu Li}

\begin{abstract}

Secure Aggregation (SA) is widely regarded as a strong defense against model-update leakage in Federated Learning (FL), as it reveals only aggregate results while hiding individual updates. In Decentralized Federated Learning (DFL), SA is commonly instantiated as local neighborhood aggregation, where each node obtains a weighted aggregate over its neighbors. We show that this locality creates a structural leakage surface: sparse decentralized topologies provide colluding semi-honest nodes with asymmetric aggregate views, exposing multiple hidden linear combinations of honest participants’ private states. Reconstructing private states from these aggregate views is fundamentally challenging, as both the private states and the aggregation coefficients are hidden. We tackle this challenge by establishing a formal connection to the Hidden Subset Sum Problem, a long-studied problem in cryptography. Building on this formulation, we design a lattice-based reconstruction approach that combines lattice reduction with structural filtering to reconstruct protected model states. We evaluate our attack on image, tabular, and text tasks under sparse DFL topologies. Our results show that colluding semi-honest nodes can recover the original local updates of honest nodes, enabling downstream reconstruction of private training data. These findings demonstrate that SA alone does not guarantee privacy in DFL when local aggregation induces asymmetric observations.
\end{abstract}




\maketitle

\section{Introduction}

Federated Learning (FL) has gained significant attention as a paradigm for collaborative machine learning. Beyond centralized architectures, Decentralized Federated Learning (DFL) enables peer-to-peer collaboration without relying on a central coordinator, thereby avoiding vulnerabilities such as a single point of failure and communication bottlenecks at the server side.

Despite the privacy-centric design of FL, a growing body of literature, primarily focusing on centralized settings, has demonstrated that FL is not a panacea for data security and privacy. Various reconstruction attacks~\cite{zhu2019deep,geiping2020inverting,xu2022agic} have proven that private training data can be leaked through shared model updates. In the context of DFL, the privacy landscape is even more complex. While some studies suggest that DFL does not inherently offer superior security compared to centralized FL (CFL)~\cite{pasquini2023security}, others argue that the peer-to-peer nature of DFL provides a privacy advantage via an information-theoretical perspective~\cite{ji2025re,yu2024provable}.

Regardless of the underlying architecture (CFL or DFL), Secure Aggregation (SA) has emerged as the foundational defense mechanism to bolster privacy~\cite{bonawitz2017practical, bell2020secure, so2022lightsecagg}. The prevailing intuition is that as long as a sufficient number of nodes are aggregated, individual contributions remain obscured within the collective sum. However, even in the more mature field of CFL, research shows that SA can be compromised by malicious servers or through specific partial participation patterns~\cite{pasquini2022eluding, so2023securing}.
The security guarantees of SA become even more tenuous when moved to a decentralized setting, where they lack rigorous analytical support. Unlike the star topology of CFL, DFL relies on localized, neighborhood-based aggregation governed by complex graph structures. Indeed, recent literature has begun to expose the shortcomings of existing DFL security proofs; for instance, Dekker et al. \cite{dekker2025topology} point out the fallacious assumption in many works that "if a single summation is secure, the protocol remains secure across multiple training iterations."
However, beyond these temporal vulnerabilities, DFL suffers from a more immediate structural vulnerability rooted in the network topology itself. The heterogeneous connectivity in DFL inherently creates uneven observational perspectives for colluding adversaries. This mismatch between the idealized algebraic symmetry expected by SA and the structural reality of DFL raises a critical question: \emph{Can topology-induced observational asymmetry turn secure neighborhood aggregates into a solvable inverse problem?}

Although recent works have hinted at potential vulnerabilities in DFL-SA, current discussions remain limited to relatively restricted cases. An example is provided by \cite{pasquini2023security}, which demonstrates that model updates can be recovered if an adversary can intercept two specific aggregated values: one containing the honest node’s contribution and another that does not. We categorize this straightforward differencing as a \textit{trivial attack} (see \eqref{eq.trivial_attack}). As the authors themselves acknowledged, a comprehensive "topology-aware" analysis is still lacking. Indeed, the security boundaries of SA in DFL have remained blurry for a long time; while the research community suspects inherent risks, there has been no systematic, generalized analysis to explain the precise conditions under which the privacy of SA can be systematically breached.

In this paper, we provide a systematic study of this problem.
By reformulating the DFL-SA process into a unified mathematical framework, we demonstrate that the privacy of honest nodes can be compromised through structural and algebraic analysis. Our main contributions are as follows:
\begin{itemize}
    \item \textbf{Exploiting Structural and Observational Asymmetry:} We reveal that the intrinsic sparsity of DFL networks leads to the observational asymmetry for corrupted nodes. We demonstrate that this structural heterogeneity allows colluding attackers to triangulate and expose the private features of honest nodes.
    \item \textbf{Unified Formulation via Lattice Challenges:} We establish a formal connection between decentralized aggregation and the Hidden Subset Sum Problem (HSSP) - a well-established NP-hard challenge in lattice-based cryptography that involves recovering hidden coefficients from a weighted sum. By exploiting the inherent sparsity of DFL topologies, we reveal that neighborhood-based SA can be reduced to a multi-dimensional HSSP (or its generalization, the Hidden Linear Combination Problem). This unified formulation allows us to apply lattice basis reduction techniques to systematically deconstruct SA, enabling the reconstruction of private model updates in the first round, or for intermediate states in later rounds.
    \item \textbf{Extensive Empirical Validation:} We validate our attack framework across diverse data modalities, including image, tabular, and text datasets. Our results demonstrate that the proposed framework successfully disentangles individual honest updates from the aggregated sum. We show that once these private model updates are recovered, they can be used as inputs to downstream gradient inversion attacks to reconstruct raw training samples, suggesting that current SA implementations in DFL may not provide additional privacy.
\end{itemize}

The remainder of this paper is structured as follows. Section~\ref{sec.prelim} establishes the fundamental concepts required and Section~\ref{ssec.threat} presents the threat model. In Section~\ref{sec.formulation}, we present the problem formulation, in which the privacy inference task is cast into a unified matrix representation and mapped to the HSSP framework. Section~\ref{sec.hssp} introduces the core lattice-based attack framework. This is followed by the optimized methodology detailed in Section~\ref{sec.filter}. Section~\ref{sec.exp} provides extensive numerical results to validate the efficacy of our proposed framework. Section~\ref{sec.related} introduces the related work. Section~\ref{sec:scope_limitations} discusses the limitations of our analysis. Finally, Section~\ref{sec.conclusion} concludes the paper and discusses potential directions for future work.

\section{Preliminaries}\label{sec.prelim}

\subsection{Notation}

Throughout this paper, we adopt the following notation: scalars are denoted by lowercase letters ($x$), vectors by lowercase boldface letters ($\boldsymbol{x}$), and matrices by uppercase boldface letters ($\boldsymbol{X}$). The $i$-th entry of vector $\boldsymbol{x}$ is denoted by $x_i$, and the $(i,j)$-th entry of matrix $\boldsymbol{X}$ is denoted by $X_{ij}$. We use $(\cdot)^{\top}$ to represent the transpose and $|\cdot|$ for the cardinality of a set. The symbols $\boldsymbol{0}$, $\boldsymbol{1}$, and $\boldsymbol{I}$ refer to the all-zero vector, the all-one vector, and the identity matrix of appropriate dimensions, respectively. Additionally, the superscript $(\cdot)^{(t)}$ indicates the value of a variable at the $t$-th iteration. The Euclidean norm is denoted by $\|\cdot\|$ and $\langle \cdot, \cdot \rangle$ denotes the standard inner product.

\subsection{Decentralized Federated Learning}
\subsubsection{Formulation}

In a standard machine learning setting, the primary objective is to find optimal model parameters $\v x \in \mathbb{R}^d$ that minimize the expected loss over a data distribution $\mathcal{D}$. This can be formulated as an \textit{Expected Risk Minimization} problem
\begin{align*}    
      &\min_{\v x}f(\v x):=\mathbb{E}_{\xi\sim\mathcal{D}}[\ell (\v x,\xi)],
\end{align*} 
where $f(\v{x})$ denotes the expected loss over the data distribution $\mathcal{D}$ and $\ell(\cdot, \cdot)$ denotes the loss function that quantifies the discrepancy between the model's prediction and the ground truth for a data sample $\xi$.

In the context of FL, data is not centrally stored but distributed across a set of nodes/clients $\mathcal{V} = \{1, 2, \dots, n\}$. Each node $i\in\mathcal{V}$ has access only to its local data distribution $\mathcal{D}_i$. So the global objective can be expressed as minimizing the average of the local loss functions
\begin{align*}    
    \min_{\v{x}} F(\v{x}) := \frac{1}{|\mathcal{V}|} \sum_{i \in \mathcal{V}} f_i(\v{x}), \quad \text{where } f_i(\v{x}) := \mathbb{E}_{\xi \sim \mathcal{D}_i}[\ell (\v{x}, \xi)].
\end{align*} 

Further, for DFL, nodes communicate through a peer-to-peer network without the coordination of a central server. We represent the network topology as a graph $\mathcal{G} = (\mathcal{V}, \mathcal{E})$, which can be either undirected or strongly connected directed to ensure convergence, 
where $\mathcal{V}$ denotes the set of nodes and 
$\mathcal{E} \subseteq \mathcal{V} \times \mathcal{V}$ denotes the set of communication links. A link $(i, j) \in \mathcal{E}$ indicates that there exists a directed connection from node $i$ to node $j$. 
We also define the set of in-neighbors of \(i\) as 
\(\mathcal{N}_i^{in} = \{\, j \mid (j,i) \in \mathcal{E} \,\}\), 
and the set of out-neighbors of \(i\) as 
\(\mathcal{N}_i^{out} = \{\, j \mid (i,j) \in \mathcal{E} \,\}\). The number of in-neighbors and out-neighbors of node \(i\) is represented as 
\(d_i^{in} = |\mathcal{N}_i^{in}|\text{~and~} d_i^{out} = |\mathcal{N}_i^{out}|\) respectively. For the undirected graph, we have $\mathcal{N}_i =\mathcal{N}_i^{in} =\mathcal{N}_i^{out} $ and $d_i=d_i^{in}=d_i^{out}$.

To solve the global optimization problem in a decentralized manner, each node $i$ maintains its own local copy of the model parameters $\v{x}_i$. So the DFL problem can be equivalently formulated with consensus constraints as
\begin{align*}    
      &\min_{\{\v x_i:i\in\mathcal{V}\}}\frac{1}{|\mathcal{V}|}\sum_{i\in\mathcal{V}} f_i(\v x_i),\nonumber\\
    &{\text{s.t.}} ~ \v x_i =\v x_j, ~ \forall (i,j)\in \mathcal{E},
\end{align*} 
where the constraints $\v{x}_i = \v{x}_j$ ensure that all nodes in the network converge to a common global model.

\subsubsection{Common Paradigm}\label{sssec.commonparadigm}

The prevalent paradigm in DFL is primarily based on the average aggregation mechanism~\cite{pasquini2023security}. This process can be generally decomposed into two sequential steps: 1) local update and 2) neighborhood aggregation.

In the local update step, in each iteration $t$, every node $i \in \mathcal{V}$ first performs a local computation step (e.g., gradient descent) to update its local model parameters based on its private dataset $\mathcal{D}_i$:
\begin{align}
\label{eq.local_update}
    \forall i \in \mathcal{V}: \v{x}_i^{(t+\frac{1}{2})} = \v{x}_i^{(t)} + \Delta \v{x}_i^{(t)},
\end{align}
where $\Delta \v{x}_i^{(t)}$ represents the local update increment, typically involving the local gradient $-\mu\nabla f_i(\v{x}_i^{(t)})$, where $\mu$ is the learning rate, or multiple local SGD steps. Also, following the standard practice in federated learning \cite{mcmahan2017communication}, all nodes are initialized with identical weights to ensure a consistent starting point for the optimization process, such that $\v x_1^{(0)}=\v x_2^{(0)}=\cdots=\v x_n^{(0)}$.

Following the local update, nodes transmit their parameters to their (out-)neighbors to achieve consensus. 
This communication and mixing process is governed by a weight matrix 
$\v{W}^{(t)} \in \mathbb{Q}^{n \times n} \cap [0,1]^{n \times n}$.~\footnote{We assume rational weights instead of real-valued ones, which aligns with the practical setting adopted in most existing implementations.} To respect the network topology, $\v{W}^{(t)}$ must belong to the set of matrices $\mathcal{W}$ defined by the graph's connectivity
\begin{align}\label{eq.w_set}
    \mathcal{W} = \left\{ \v{W}^{(t)} \in \mathbb{Q}^{n \times n} \mid W_{ij}^{(t)} = 0 \text{ if } (j, i) \notin \mathcal{E} \text{ and } i \neq j \right\}.
\end{align}
Then the global aggregation step can be expressed in a compact matrix form
\begin{align}
\label{eq.aggr}
    \v{X}^{(t+1)} = \v{W}^{(t+1)} \v{X}^{(t+\frac{1}{2})},
\end{align}
where $\v{X}^{(t)} = [\v{x}_1^{(t)}, \dots, \v{x}_n^{(t)}]^\top \in \mathbb{R}^{n \times d}$ concatenates the local parameters of all nodes.

In fact, the convergence of DFL algorithms heavily relies on the spectral properties of $\v{W}^{(t)}$. For undirected graphs, most existing work~\cite{lian2017can, koloskova2019decentralized, xing2020decentralized, yuan2021decefl} assumes that 1) $\v{W}^{(t)}$ is doubly stochastic (and often symmetric), i.e., $\v{W}^{(t)} \v{1} = \v{1}$ and $\v{1}^\top \v{W}^{(t)} = \v{1}^\top$.
2) $\rho\left(\v{W}^{(t)} - \frac{\v{1}\v{1}^\top}{n}\right) < 1$, where $\rho(\cdot)$ denotes the spectral radius. This condition ensures that the local models contract toward the global average at a geometric rate \cite{olshevsky2009convergence}.

For scenarios involving directed graphs, where the communication links are not necessarily bi-directional, a doubly stochastic matrix may be difficult or impossible to construct. A few recent works, such as \cite{assran2019stochastic,li2025asymmetrically}, employ the Push-Sum algorithm~\cite{kempe2003gossip} to handle such topologies. However, regardless of whether a standard consensus or a Push-Sum-based mechanism is employed, the fundamental update paradigm remains structurally consistent with the linear mixing form in \eqref{eq.aggr}, as we explain in Appendix~\ref{app.pushsum}. For the sake of brevity and clarity, we will adopt this unified representation in the subsequent sections of this paper to describe the general decentralized aggregation process.

\subsubsection{Temporal Dynamics of the Weight Matrix}

In the majority of DFL literature, the weight matrix is typically assumed to be static, meaning $\v{W}^{(t)} = \v{W}$ for all iterations $t$. From a physical perspective, this assumption implies that the network topology is invariant and that all participating nodes engage in the aggregation process in every round with constant mixing weights.
This static assumption does not always hold in practice. For instance, asynchronous updates ~\cite{dekker2025topology, jin2016scale} or time-varying topologies~\cite{li2024accelerated} can result in a dynamic weight matrix $\v{W}^{(t)}$ that evolves over time.

\subsubsection{(Approximated) Gradient Inversion Attack}\label{sssec.gia}
Gradient Inversion Attack (GIA) is a potent privacy threat in FL, where an adversary aims to reverse-engineer private local datasets by intercepting the information exchanged during training.
\red{
 Recent systematizations provide comprehensive analyses of GIA threat models and attack mechanisms~\cite{carletti2025sok,guo2025exploring}. The underlying intuition of GIA is that the shared model updates (gradients or gradient-related increments) may encode rich information about the local training samples.
}

Formally, given the local update $\Delta \v{x}_i^{(t)}$ from node $i$, the attacker attempts to reconstruct the private dataset $\mathcal{D}_i$ by solving a constrained optimization problem
\begin{align*} 
    \arg\min_{\mathcal{D}'} \text{Dist} \left( \Delta \v{x}_i^{(t)}, \hat{\Delta} \v{x}_i^{(t)}\right),
\end{align*}
where $\mathcal{D}'$ is the dummy dataset, $\hat{\Delta} \v{x}_i^{(t)}$ is the update generated by the dummy data, $\text{Dist}$ is a distance metric, such as Euclidean distance~\cite{zhu2019deep} or cosine similarity~\cite{geiping2020inverting} between the observed update ${\Delta} \v{x}_i^{(t)}$ and the update generated by the dummy data $\hat{\Delta} \v{x}_i^{(t)}$.

Existing literature has explored numerous variants of GIA, ranging from early methods targeting single-batch gradients~\cite{zhu2019deep,geiping2020inverting,zhao2020idlg} to more complex approximated attacks~\cite{xu2022agic,song2023approximate,hatamizadeh2023gradient} capable of handling multiple local steps or aggregated updates. In fact, the efficacy of any GIA variant fundamentally relies on the attacker's ability to process the obtained ${\Delta} \v{x}_i^{(t)}$. 

In this work, our focus remains on the structural exploitability of the communication protocol. Specifically, we investigate the extent to which an adversary can intercept or reconstruct $\Delta \v{x}_i^{(t)}$ when it is shielded by SA primitives, rather than focusing on the refinement of the GIA algorithm itself. 
\red{GIA is therefore used as a downstream instantiation to demonstrate the privacy consequences of exposing individual updates. For clarity, we first consider a standard single-step SGD update~\cite{geiping2020inverting},}
\begin{align*}
   \Delta \v{x}_i^{(t)} = -\mu \v{g}_i^{(t)}, \quad \text{with } \v{g}_i^{(t)} = \frac{1}{|\mathcal{B}_i^{(t)}|} \sum_{\xi \in \mathcal{B}_i^{(t)}} \nabla \ell(\v{x}_i^{(t)}, \xi),
\end{align*}
where $\mathcal{B}_i^{(t)}\subseteq\mathcal{D}_i$ is the mini-batch of private samples. So the leakage of $\Delta \v{x}_i^{(t)}$ directly exposes the batch-averaged gradient.

\subsection{Secure Aggregation in Decentralized Network}

In a decentralized setting, for each node $i$, the objective of SA is to compute the weighted sum of updates from its (in)-neighbors $\mathcal{N}_i^{in}$ without revealing the weight $W_{ij}$ and the specific vector $\v{x}_j^{(t+\frac{1}{2})}$ of any individual neighbor $j \in \mathcal{N}_i^{in}$. Following \eqref{eq.w_set} and the global aggregation representation in \eqref{eq.aggr}, the protocol aims to output the weighted result at node $i$:
\begin{align}
\label{eq:sec_agg_dfl}
    \v{x}_i^{(t+1)} = \sum_{j \in \mathcal{N}_i^{in} \cup \{i\}} W_{ij} \v{x}_j^{(t+\frac{1}{2})}.
\end{align}

Under this formulation, the decentralized SA problem can be viewed as a series of local instances of the SA in centralized FL. Consequently, established protocols from centralized FL can be migrated and adapted to the decentralized topology~\cite{pasquini2023security}. Each node $i$ effectively executes a SA protocol with its direct neighborhood, treating its (in)-neighbors as "clients" and itself as the "local server". While this framework theoretically provides a layer of privacy by "shrouding" individual updates within the weighted sum \eqref{eq:sec_agg_dfl}, we will demonstrate in the following sections that this structural dependency on the topology introduces vulnerabilities that can be exploited for data reconstruction for exact local neighborhood SA protocols that
reveal aggregate states to corrupted participants.

\subsection{Fundamentals of Hidden Subset Sum Problem}

The formal definitions of the HSSP and HLCP problem are
\begin{definition}[Hidden Subset Sum Problem~\cite{gini2022hardness}] \label{def:hssp}
    Let $Q$ be a positive integer modulus and $\{{x}_i\}_{i=1}^N$ be a set of unknown secret integers in $\mathbb{Z}_Q$. Let $\boldsymbol{\alpha}_1, \dots, \boldsymbol{\alpha}_N \in \{0, 1\}^M$ be $N$ unknown binary vectors. Given $Q$ and an observed vector $\v{h} = (h_1, \dots, h_M) \in \mathbb{Z}_Q^M$ that satisfies:
    \begin{equation*}
        \v{h} \equiv \sum_{i=1}^N x_i \boldsymbol{\alpha}_i \pmod Q, 
    \end{equation*}
    the HSSP is to recover the secret $\{x_i\}_{i=1}^N$ and the binary vectors $\{\boldsymbol{\alpha}_i\}_{i=1}^N$, up to a permutation of the indices $i \in \{1, \dots, N\}$.
\end{definition}

\begin{definition}[Hidden Linear Combination Problem (HLCP)~\cite{gini2022hardness}]
Let $Q$ be a positive integer modulus and $\{{x}_i\}_{i=1}^N$ be a set of unknown secret integers in $\mathbb{Z}_Q$. Let $\boldsymbol{\alpha}_1, \dots, \boldsymbol{\alpha}_N \in \{0, 1,\cdots, c\}^M$ be $N$ unknown vectors for a given $c\in\mathbb{N}^+$. Given $Q$ and an observed vector $\v{h} = (h_1, \dots, h_M) \in \mathbb{Z}_Q^M$ that satisfies:
\begin{equation*}
    \v{h} \equiv \sum_{i=1}^N x_i \boldsymbol{\alpha}_i \pmod Q, 
\end{equation*}
the HLCP is to recover the secret $\{x_i\}_{i=1}^N$ and the vectors $\{\boldsymbol{\alpha}_i\}_{i=1}^N$, up to a permutation of the indices $i \in \{1, \dots, N\}$.
\end{definition}

We further extend the HSSP and HLCP to the case where the secrets are multidimensional.
\begin{definition}[Multidimensional Hidden Subset Sum Problem (mHSSP)~\cite{li2024perfect}] \label{def:mhssp}
    Let $Q$ be a positive integer modulus and $\{\v{x}_i\}_{i=1}^N$ be $N$ unknown secret vectors in $\mathbb{Z}_Q^u$. Let $\boldsymbol{\alpha}_1, \dots, \boldsymbol{\alpha}_N \in \{0, 1\}^M$ be $N$ unknown binary vectors. For each dimension $j \in \{1, \dots, u\}$, let $\v{h}_j \in \mathbb{Z}^M_Q$ be an observed sample vector satisfying:
    \begin{equation*}
        \v{h}_j \equiv \sum_{i=1}^N x_{i,j} \boldsymbol{\alpha}_i \pmod Q,
    \end{equation*}
    where $x_{i,j}$ is the $j$-th component of vector $\v{x}_i$. Given $Q$ and the sample matrix $\v{H} = [\v{h}_1, \dots, \v{h}_u] \in \mathbb{Z}^{M \times u}_Q$, the goal is to recover the secret vectors $\{\v{x}_i\}_{i=1}^N$ and the binary vectors $\{\boldsymbol{\alpha}_i\}_{i=1}^N$, up to a common permutation of the index $i$.
    
    The problem can be expressed more compactly in matrix form. Let $\v{X} = [\v{x}_1, \dots, \v{x}_N]^\top \in \mathbb{Z}^{N \times u}_Q$ and $\v{A} = [\boldsymbol{\alpha}_1, \dots, \boldsymbol{\alpha}_N] \in \{0, 1\}^{M \times N}$. The mHSSP is then characterized by the relation:
    \begin{equation*}
    \label{eq.mhssp}
        \v{H} \equiv \v{A}\v{X} \pmod Q.
    \end{equation*}
    
    Analogously, the \textbf{Multidimensional Hidden Linear Combination Problem (mHLCP)} follows the same definition, with the exception that the entries of $\v{A}$ are drawn from $\{0, \dots, c\}$ for some $c\in\mathbb{N}^+$.
\end{definition}

\section{Threat Model}\label{ssec.threat}

In this study, we consider local neighborhood SA protocols, where a corrupted receiver observes only its post-aggregation state together with its own and colluding nodes' states, while individual honest contributions and their exact coefficients remain hidden. The attack assumes access to aggregate outputs over a finite field.

We consider a set of honest-but-curious (semi-honest) nodes, denoted by $\mathcal{V}_c$, which follow the prescribed training and SA protocols but collude by sharing their locally observed information. The adversaries aim to infer the private states of the honest participants $\mathcal{V}_h$, with the goal of reconstructing their intermediate model parameters $\v{x}_h^{(t+\frac{1}{2})}$, which are intended to be protected by the SA mechanism. These local parameters $\v{x}_i^{(t+\frac{1}{2})}$ (and the corresponding updates $\Delta \v{x}_i^{(t)}$) are highly sensitive, as they encode rich information about the underlying training data. As demonstrated in Section~\ref{sssec.gia}, such information can be leveraged to reconstruct raw training samples with high fidelity.

Formally, for all iterations $t \in \mathcal{T}$, where $\mathcal{T} = \{0, 1, \dots, t_{\max}-1\}$, the minimum collective information set available to the adversaries is defined as
\begin{align}
\label{eq.adv_info}
    \mathcal{O} \coloneq \left\{ \v{x}_i^{(0)} \right\}_{i \in \mathcal{V}} 
    \cup \left\{ \v{x}_i^{(t)}, \v{x}_i^{(t+\frac{1}{2})} \right\}_{i \in \mathcal{V}_c,\, t \in \mathcal{T}} 
    \cup \left\{ W_{ij}^{(t+1)} \right\}_{i,j \in \mathcal{V}_c,\, t \in \mathcal{T}}.
\end{align}
This set captures the adversaries' local states before and after aggregation, their mutual weights, and the globally shared initialization (cf. Section~\ref{sssec.commonparadigm}).
 Beyond this baseline information, the adversary may possess varying degrees of topological metadata depending on the specific SA protocol in use. We categorize this prior knowledge into three hierarchical scenarios:
\begin{itemize}
    \item \textbf{Case 1: Full Neighbor Awareness (Identity-Aware).} In this scenario, each corrupted node $i \in \mathcal{V}_c$ knows the exact global identities (indices $j$) of its incoming neighbors. By colluding, adversaries can cross-reference these indices to reconstruct the local adjacency structure. The information set in this case is augmented by the support of the weight matrix
    \begin{equation*}
        \mathcal{O} \coloneq \mathcal{O} \cup \left\{ \mathbbm{1}_{\{W_{ij} > 0\}} \right\}_{i \in \mathcal{V}_c, j \in \mathcal{V}},
    \end{equation*}
    where $\mathbbm{1}_{\{\cdot\}}$ is the indicator function. This allows the adversary to pinpoint which honest nodes' parameters are involved in their specific aggregation.

    \item \textbf{Case 2: Partial Awareness (Degree-Only).} Each corrupted node $i \in \mathcal{V}_c$ is aware of its in-degree $d_i^{in}$ (i.e., the number of nodes it aggregates from) but cannot map these neighbors to specific global identities. Since identities are masked, colluding adversaries cannot link their local observations to form a global map. The additional knowledge is restricted to the cardinality of the neighbor sets
    \begin{equation*}
        \mathcal{O} \coloneq \mathcal{O} \cup \left\{ d_i^{in} \right\}_{i \in \mathcal{V}_c}.
    \end{equation*}

    \item \textbf{Case 3: Zero Awareness.} The SA protocol provides no information regarding the identities or the quantity of incoming neighbors.

\end{itemize}

The three cases correspond to different real-world implementations, as detailed in Appendix~\ref{app.cases}.

\section{Problem Formulation}\label{sec.formulation}

In this section, we first establish a unified matrix formulation for the inference problem. Building upon this, we propose a topology-aware reduction method to further simplify the model by stripping away redundant connections. Finally, we map the resulting matrix structure of the core sub-graph to the HSSP framework for joint reconstruction.

\subsection{Matrix Representation}

By analyzing the global model evolution from the adversarial perspective, we characterize the resulting privacy leakage as an algebraically decoupled linear inverse problem. 

We partition the weight matrix $\v{W}^{(t)}$ according to the indices of corrupted nodes $\mathcal{V}_c$ and honest nodes $\mathcal{V}_h$. Without loss of generality, assuming the first $|\mathcal{V}_c|$ rows/columns correspond to $\mathcal{V}_c$, $\v{W}^{(t)}$ is partitioned as:
\begin{align*}
\v W^{(t)} =\begin{pmatrix}
\v W^{(t)}_{cc} & \v W^{(t)}_{ch}  \\
\v W^{(t)}_{hc} & \v W^{(t)}_{hh}  \\
\end{pmatrix},
\end{align*}
where $\v{W}^{(t)}_{cc}$ and $\v{W}^{(t)}_{hh}$ represent internal weights, while $\v{W}^{(t)}_{ch}$ and $\v{W}^{(t)}_{hc}$ represent cross-group influences.
Let $\v X_c^{(t)}$ and $\v X_h^{(t)}$ denote the model parameters held by the respective node sets, we summarize this structural reduction in the following proposition.
\begin{proposition}
\label{prop:isolation}
Given the adversarial observation set $\mathcal{O}$, the global aggregation process over $t_{\max}$ iterations can be algebraically decoupled into independent round-wise linear systems. Specifically, for any iteration $t \in \mathcal{T}$, the privacy inference of honest states reduces to solving the core sub-problem:
\begin{align} \label{eq.sub_problem}
    \forall t \in \mathcal{T}: \v{W}_{ch}^{(t+1)} \v{X}_h^{(t+\frac{1}{2})} = \v{X}_c^{(t+1)}-\v W_{cc}^{(t+1)}\v X_c^{(t+\frac{1}{2})}, 
\end{align}
where the right-hand side consists entirely of information accessible to the adversaries. 
The adversary's capability differs across iterations. At $t=0$, the adversary can exactly recover the local updates $\Delta \v{x}_i^{(0)}$ for all $i \in \mathcal{V}_h$, due to the globally known initialization. For $t \geq 1$, the adversary can recover the intermediate states $\v{X}_h^{(t+\frac{1}{2})}$ but cannot determine the corresponding updates $\Delta \v{x}_i^{(t)}$, as reconstructing $\v{X}_h^{(t)}$ requires access to the internal weights $\v{W}_{hc}^{(t+1)}$ and $\v{W}_{hh}^{(t+1)}$, which are not observable. Therefore, exact gradient reconstruction is only possible at $t=0$, which is sufficient for the downstream attack considered in this work.
\end{proposition}
\begin{proof}
    See Appendix~\ref{app.proof_isolation}.
\end{proof}

\subsection{Topology-aware Simplification}\label{ssec.graph_sim}

By inspecting \eqref{eq.sub_problem}, it is evident that the reconstruction of honest states $\v{X}_h^{(t+\frac{1}{2})}$ depends heavily on the structure and properties of $\v{W}_{ch}^{(t+1)}$. In decentralized networks, the sparsity of the graph often dictates that the global inference task is not a single monolith, but rather a collection of loosely coupled or even isolated sub-problems. 

The motivation for simplifying the problem based on topology is two-fold. First, the computational complexity of the HSSP attack is highly sensitive to the number of unknown honest variables ($|\mathcal{V}_h|$). By decomposing the graph into smaller sub-components, we can significantly reduce the search space for the lattice-based solver. Second, specific topological configurations allow for trivial reconstructions that do not require complex solvers. 
To systematically reduce the problem size, we classify nodes into three categories based on their neighborhood connectivity. We begin our discussion with undirected graphs, noting that the logic naturally extends to directed graphs with minor adjustments to the weight constraints.

\subsubsection{Undirected Graph.}

As also noted in \cite{dekker2025topology}, secure summation is vulnerable if a corrupted node has only one honest neighbor. In a weighted aggregation setting, this vulnerability persists as such trivial attack. If a corrupted node $i$ has exactly one honest neighbor $j$, the adversary can directly compute $\v{x}_j^{(t+\frac{1}{2})}$ by subtracting the contributions of its corrupted neighbors and itself from the aggregated result, i.e.,
\begin{align}\label{eq.trivial_attack}
    \v{x}_j^{(t+\frac{1}{2})} = \frac{\v{x}_i^{(t+1)} - \sum_{k \in (\mathcal{N}_i \cap \mathcal{V}_c)\cup\{i\}} W_{ik}^{(t+1)} \v{x}_k^{(t+\frac{1}{2})}}{W_{ij}^{(t+1)}},
\end{align}
where $W_{ij}^{(t+1)} = 1 - \sum_{k \in (\mathcal{N}_i \cap \mathcal{V}_c)\cup\{i\}} W_{ik}^{(t+1)}$ due to the row-stochastic property of $\v{W}^{(t+1)}$.

Thus, corrupted nodes $i \in \mathcal{V}_c$ can be categorized into three types based on their honest neighborhood $\mathcal{N}_i \cap \mathcal{V}_h$: 
\begin{itemize}
    \item If $|\mathcal{N}_i \cap \mathcal{V}_h| = 0$, the node provides no extra information about honest participants, as its observations are entirely determined by other corrupted nodes.
    \item If $|\mathcal{N}_i \cap \mathcal{V}_h| = 1$, the state of its sole honest neighbor can be uniquely and directly recovered via \eqref{eq.trivial_attack}.
    \item If $|\mathcal{N}_i \cap \mathcal{V}_h| \geq 2$, the local observation at node $i$ is individually underdetermined.
\end{itemize}
This categorization is applicable to both Case 1 and 2, since the knowledge of $d_i^{in}$ allows it to determine the exact number of honest contributors.

Correspondingly, from the perspective of each independent honest node $i \in \mathcal{V}_h$, its exposure risk within a single iteration is determined as
\begin{itemize}
    \item If there exists at least one corrupted neighbor $j \in \mathcal{V}_c$ such that $i$ is its only honest neighbor ($|\mathcal{N}_j \cap \mathcal{V}_h| = 1$), node $i$ is compromised by a trivial attack.
    \item If the node has no corrupted neighbors ($\mathcal{N}_i \cap \mathcal{V}_c = \emptyset$), its information cannot be leaked to the adversary within a single iteration, as DFL updates only propagate to immediate neighbors in one hop.
    \item Otherwise, its privacy depends on the resolvability of the resulting HSSP attack.
\end{itemize}

Based on these classifications, the network can be simplified into a reduced bipartite graph, as illustrated in Figure~\ref{fig:bipartite}(a), by stripping away redundant or trivial connections. In this refined structure, red dashed lines represent communication between corrupted nodes, which is fully transparent to the adversaries. Conversely, black dashed lines denote internal interactions among honest nodes that remain inaccessible to the adversaries within a single iteration. The red solid lines highlight "trivial" links where an honest node serves as the sole neighbor of an adversary; the privacy states of these nodes are directly reconstructed and subsequently removed from the inference set. Finally, the yellow solid lines constitute the core sub-graph, where each corrupted node is coupled with multiple honest neighbors. This reduced structure represents the minimal and most challenging form of the inference problem.~\footnote{As any arbitrary graph can be systematically decomposed through this methodology, we adopt the convention that \textbf{all graphs discussed hereafter are the simplified core sub-graphs}. This simplification allows us to focus the analysis on the fundamental solvability of the HSSP-based attack within the most difficult components of the network.}

\begin{figure}[h]
  \centering
  \includegraphics[width=0.85\linewidth]{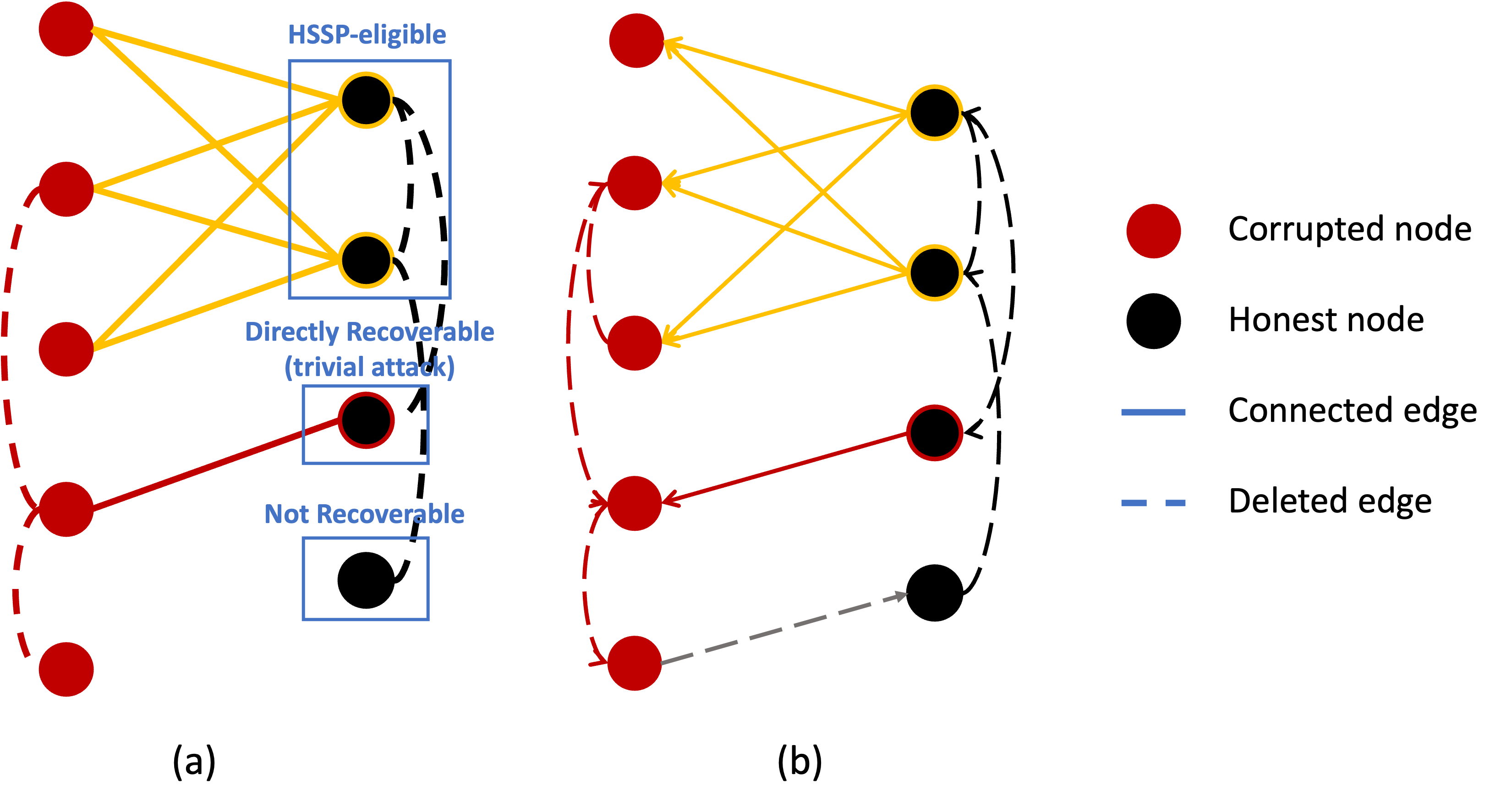}
  \caption{Example of (a) undirected graph and (b) directed simplification. For an undirected graph, red dashed lines denote transparent intra-adversary communication; black dashed lines represent inaccessible intra-honest interactions. Red solid lines highlight trivial links for direct reconstruction, while yellow solid lines identify the core sub-graph targeted for HSSP-based inference. For a directed graph, arrows indicate the direction of information flow. Grey solid arrows represent information flowing from corrupted nodes to honest nodes, providing no utility for inference, and other symbols are consistent with the undirected graph.}
  \Description{}
  \label{fig:bipartite}
\end{figure}

\subsubsection{Directed Graph}

We assume the directed graph is strongly connected, as this is a necessary condition for achieving consensus. Strong connectivity ensures that information flow covers the entire network, implying that every node possesses at least one in-neighbor and one out-neighbor. While the simplification logic for directed graphs parallels that of undirected graphs, several distinctions need to be addressed.
First, the classification of corrupted nodes must focus exclusively on their in-neighbors, as these represent the direction of information flow from honest participants to the adversaries. Second, unlike undirected protocols, directed graph algorithms typically do not satisfy the row-stochastic constraint. Consequently, even if a corrupted node has exactly one honest in-neighbor, its privacy states cannot always be uniquely recovered because the precise weight is unknown.
However, since weights in DFL are rational numbers constrained within $[0,1)$, the adversary can still narrow down the possible values. By identifying the least common denominator of the potential weights (see Section~\ref{ssec.connection}), the adversary can restrict the honest node’s privacy state to a finite set of candidate solutions. This "directed trivial attack" yields a discrete solution space rather than a unique value.

Figure~\ref{fig:bipartite}(b) illustrates the simplification of a directed graph. We use the arrow to signify the direction of information flow. Edges from corrupted nodes to honest nodes do not provide useful observations for
inferring honest states at corrupted receivers (represented by the grey line). The red solid line indicates scenarios where a directed trivial attack can be executed to obtain a finite solution set.

\subsection{Connection to HSSP}\label{ssec.connection}
Following the topological simplification process, the task reduces to a matrix decomposition inverse problem. In this subsection, we establish a formal reduction showing that localized aggregation in DFL can be formulated as solving a lattice-based cryptographic problem.
\begin{theorem}
\label{thm:reduction}
Let the continuous linear system $\v{Y} = \v{W}_{ch} \v{X}_h$ denote the core sub-problem after topological pruning, where $\v{Y} \in \mathbb{R}^{M \times u}$ is the known adversarial observation matrix, $\v{W}_{ch} \in \mathbb{Q}^{M \times N}$ is the hidden weight matrix, and $\v{X}_h \in \mathbb{R}^{N \times u}$ encapsulates the unknown private honest states. 

Under the conditions of:
\begin{enumerate}
    \item There exists a known scaling integer $\beta \in \mathbb{Z}^+$ such that $\tilde{\v{W}} = \beta \v{W}_{ch} \in \mathbb{Z}^{M \times N}$.
    \item \red{The honest states can be mapped exactly into the integer domain using a
precision factor $\gamma$, yielding
$\tilde{\v{X}}_h=10^{\gamma}{\v{X}}_h\in\mathbb{Z}^{N\times u}$.\footnote{We consider the idealized setting in which the integer-domain
conversion is exact here. The truncation errors introduced by
finite-precision conversion are addressed in
Section~\ref{sec:finite_precision_reduction}.}}
    \item The scaled integer weights $\tilde{W}_{ij}$ are bounded within a known discrete set $\mathcal{C} \subset \mathbb{Z}$.
\end{enumerate}
\red{
Then the recovery of the hidden scaled weight matrix $\tilde{\v W}$ and state $\tilde{\v X_{h}}$ from the aggregated observations can be formulated as an instance of mHSSP if $\mathcal{C} \subseteq \{0, 1\}$, or mHLCP if $\mathcal{C} \subseteq \{0, 1, \dots, c\}$ for some integer $c \geq 2$.
}
\end{theorem}

\begin{proof}
    See Appendix~\ref{app.reduction}.
\end{proof}

Theorem~\ref{thm:reduction} establishes the algebraic framework of the attack. However, whether the adversary can uniquely recover the ground-truth solution from the resulting system critically depends on the information diversity induced by the network topology.

\begin{remark}[Topological Asymmetry]

The unique identifiability of the honest nodes' private states is governed by the rank of the joint observation matrix. Within the minimal core sub-graph, all "trivial" instances have already been resolved and removed; the remaining system consists of honest nodes whose states are coupled. A prerequisite for identifying a unique solution is that the resulting system must not be underdetermined. Mathematically, this imposes a structural necessity: the number of independent observations (corrupted nodes) must be at least equal to the number of hidden variables (honest nodes) within the core sub-graph, i.e., $|\mathcal{V}_c| \geq |\mathcal{V}_h|$. 

While this node cardinality is a necessary condition, topological asymmetry is the essential catalyst that ensures the system is not rank-deficient. In a perfectly symmetric network, such as a fully connected graph with uniform weights (e.g., \cite{zhang2026information}), all adversaries share redundant observational vantage points. This causes the rows of $\v{W}_{ch}$ to be linearly dependent, resulting in a rank-deficient system ($\text{rank}=1$) where a unique solution to the HSSP cannot exist regardless of the number of adversaries. Conversely, the heterogeneous neighborhood structures and non-uniform weights inherent in practical DFL break this symmetry. This observational diversity ensures that colluding adversaries capture linearly independent projections of the honest nodes' states, allowing $\v{W}_{ch}$ to achieve the full column rank necessary for a unique and deterministic reconstruction.

\end{remark}

\red{
\subsection{Finite-Precision Reduction to Noisy
mHSSP/mHLCP}
\label{sec:finite_precision_reduction}

Theorem~\ref{thm:reduction} focuses on settings in which the
aggregation process admits an exact integer-domain representation. This
condition is satisfied when the protocol operates directly over an integer or finite-field domain, or when all represented states lie
exactly on a prescribed fixed-point grid. In practical finite-precision
implementations, rounding and truncation may introduce a bounded numerical mismatch.

Define the scalar integerization operator
\begin{equation*}
q_{\gamma}(z)=\operatorname{trunc}(10^{\gamma}z),
\end{equation*}
where $\operatorname{trunc}(\cdot)$ denotes truncation operation. For a matrix $\boldsymbol{Z}$, $q_{\gamma}(\boldsymbol{Z})$
denotes the element-wise application of $q_{\gamma}$.

Let
\[
\tilde{\boldsymbol{Y}}
=
q_{\gamma}
\left(
\beta\boldsymbol{Y}^{}
\right)
\qquad\text{and}\qquad
\tilde{\boldsymbol{X}}_h
=
q_{\gamma}(\boldsymbol{X}_h).
\]
We ues $\boldsymbol{E}$ to represent the aggregate numerical residual.
Then the finite-precision observation therefore satisfies
\begin{equation}
\tilde{\boldsymbol{Y}}
\equiv
\tilde{\v{W}}\tilde{\boldsymbol{X}}_h
+
\boldsymbol{E}
\pmod Q.
\label{eq:noisy_mhssp_relation}
\end{equation}
The residual $\boldsymbol{E}$ jointly captures the numerical errors
introduced by finite-precision aggregation and by the integerization of
the aggregate and local states.

This reformulation recasts the reconstruction problem as a noisy mHSSP/mHLCP instance.

}

\section{HSSP and attacks}\label{sec.hssp}
In this section, we provide an introduction to HSSP attack and the fundamental concepts essential to the proposed attack framework.

\subsection{Lattice Fundamentals}

A lattice $\mathcal{L}$ in $\mathbb{R}^M$ is a discrete additive subgroup formed by the integer linear combinations of a set of basis vectors. It is defined as~\cite{cassels1971introduction,micciancio2002complexity}:
\begin{definition}[Lattice]
Given a set of $N$ linearly independent vectors $\{\v{b}_1, \dots, \v{b}_N\}$ in $\mathbb{R}^M$ (where $M \geq N$), the lattice $\mathcal{L}$ generated by these vectors is the set:
\begin{equation*}
    \mathcal{L}(\v{b}_1, \dots, \v{b}_n) = \left\{ \sum_{i=1}^N x_i \v{b}_i : x_i \in \mathbb{Z}, i = 1, \dots, N \right\}.
\end{equation*}
The parameter $N$ denotes the rank of the lattice, which is also represented as $\text{dim}(\mathcal{L})$. When $N = M$, the lattice is referred to as full-rank. 
\end{definition}

The lattice can be expressed compactly as $\mathcal{L}(\v{B}) = \{ \v{B} \cdot \v{x} : \v{x} \in \mathbb{Z}^N \}$, where matrix $\v{B} \in \mathbb{R}^{M \times N}$ 
comprises the basis vectors $\v{b}_1, \dots, \v{b}_N$ as its columns.
In this work, we focus on integer lattices, i.e., sublattices of $\mathbb{Z}^M$. 
Given a lattice $\mathcal{L} \subseteq \mathbb{Z}^M$, it is useful to consider the following concept.
\begin{definition}[Orthogonal Lattice]
For an integer lattice $\mathcal{L} \subseteq \mathbb{Z}^M$, its orthogonal lattice $\mathcal{L}^\perp$ is defined by
\begin{equation*}
    \mathcal{L}^\perp := \{ \v{y} \in \mathbb{Z}^M : \forall \v{x} \in \mathcal{L}, \langle \v{x}, \v{y} \rangle = 0 \}.
\end{equation*}
 Furthermore, the orthogonal lattice modulo $Q$ is defined as
\begin{equation*}
    \mathcal{L}_Q^\perp := \{ \v{y} \in \mathbb{Z}^M : \forall \v{x} \in \mathcal{L}, \langle \v{x}, \v{y} \rangle \equiv 0 \pmod Q \}.
\end{equation*}
\end{definition}

The lattice completion $\bar{\mathcal{L}}$ is defined as $\bar{\mathcal{L}} = \text{Span}_{\mathbb{R}}(\mathcal{L}) \cap \mathbb{Z}^M = (\mathcal{L}^\perp)^\perp$.

\begin{definition}[Shortest Vector and Successive Minima]
 The length of the shortest non-zero vector $\v{v}$ in $\mathcal{L}$ is the first minimum, denoted as $\lambda_1(\mathcal{L}) = \|\v{v}\|$. More generally, for a lattice of rank $N$, the $i$-th successive minimum $\lambda_i(\mathcal{L})$ is the smallest radius $r$ such that the closed ball of radius $r$ centered at the origin contains at least $i$ linearly independent lattice vectors, for $i\in\{1,2,\cdots,N\}$.
\end{definition}

Following these definitions, lattice basis reduction algorithms are important tools for addressing cryptographic challenges such as HSSP, where they are used to recover short vectors corresponding to hidden structures. The prominent techniques are briefly summarized as follows:
\begin{itemize}
    \item \textbf{LLL Algorithm:} Proposed by Lenstra, Lenstra, and Lovász \cite{lenstra1982factoring}, this polynomial-time algorithm transforms an arbitrary lattice basis into an LLL-reduced basis. Such a basis consists of relatively short, nearly orthogonal vectors, which significantly simplifies tasks such as finding the shortest lattice vector.
    
    \item \textbf{BKZ Algorithm:} The Block Korkine-Zolotarev (BKZ) algorithm \cite{chen2011bkz} generalizes LLL by performing basis reduction on blocks of size $k$. It offers a tunable trade-off, where larger block sizes provide higher reduction quality (i.e., closer to the shortest vector) at the cost of increased computational effort. While BKZ with the lowest block size $k=2$ runs in polynomial time, finding the absolute shortest vector generally requires exponential time with full block size.
\end{itemize}

\subsection{The Nguyen-Stern Attack for HSSP and mHSSP}

Based on LLL and BKZ algorithm, the Nguyen-Stern (NS) attack framework \cite{nguyen1999hardness} is a cornerstone technique for solving HSSP. The attack is executed in two primary phases.

\paragraph{Step 1: }
This phase aims to identify the orthogonal lattice associated with the hidden coefficient matrix $\v{A}=[\v{\alpha}_1, \dots, \v{\alpha}_N]$ in Definition~\ref{def:hssp}. Given the observed vector $\v{h} \equiv \sum_{i=1}^{N} x_i \v{\alpha}_i \pmod Q$, the attack begins by constructing the orthogonal lattice modulo $Q$ of $\v{h}$, denoted as $\mathcal{L}_Q^\perp(\v{h})$. 

Since $\v{h}$ is a linear combination of the hidden binary vectors $\{\v{\alpha}_i\}_{i=1, \dots, N}$, the orthogonal lattice $\mathcal{L}^\perp(\v{A})$ is inherently contained within $\mathcal{L}_Q^\perp(\v{h})$. Under the assumption that the $\v{\alpha}_i$ vectors are binary and sufficiently short, the first $M-N$ short vectors obtained from an LLL-reduced basis of $\mathcal{L}_Q^\perp(\v{h})$ are highly likely to span $\mathcal{L}^\perp(\v{A})$. Subsequently, the completion lattice $\bar{\mathcal{L}}(\v{A}) = (\mathcal{L}^\perp(\v{A}))^\perp$ is computed by applying the LLL algorithm a second time. This resulting lattice contains the target lattice $\mathcal{L}(\v{A})$, effectively isolating the span of the hidden coefficients.

This methodology can be naturally extended to the mHSSP setting by replacing the observation vector $\v{h}$ with the matrix $\v{H} \in \mathbb{Z}_Q^{M \times u}$, as formally established in \cite{li2024perfect}.

\paragraph{Step 2: }
The second phase focuses on extracting the specific binary vectors $\v{\alpha}_i$ from the reduced basis of $\bar{\mathcal{L}}(\v A)$. To achieve higher precision than LLL, the BKZ algorithm is typically employed to discover shorter vectors within this space. Let $\{\v{v}_i\}_{i=1}^N$ be the set of short vectors obtained.

Due to the binary nature of the target vectors, the $\v{\alpha}_i$ coefficients are expected to be either the vectors $\v{v}_i$ themselves or simple linear combinations, such as $\v{v}_i \pm \v{v}_j$. By searching the set $\{\v{v}_i\} \cup \{\v{v}_i - \v{v}_j\} \cup \{\v{v}_i + \v{v}_j\}$, the adversary identifies $N$ candidate binary vectors to reconstruct the matrix $\v{A}$. Once $\v{A}$ is determined, a non-singular $N \times N$ sub-matrix $\v{A}'$ is selected alongside its corresponding observation $\v{h}'$. The hidden private data $\v{x}$ is then recovered by solving the linear system:
\begin{equation}\label{eq.reconstruct}
    \v{x} \equiv (\v{A}')^{-1} \v{h}' \pmod Q.
\end{equation}

The effectiveness of the NS attack is contingent upon the parameters $M$, $N$, and $Q$ satisfying specific density conditions. The fundamental intuition is that any vector $\v{y}$ orthogonal to $\v{h}$ modulo $Q$ must also be orthogonal to each hidden vector $\v{\alpha}_i$. Specifically, we consider
\begin{equation*}
    \langle \v{y}, \v{h} \rangle = x_1 \langle \v{y}, \v{\alpha}_1 \rangle + \dots + x_N \langle \v{y}, \v{\alpha}_N \rangle \equiv 0 \pmod Q.
\end{equation*}
Let $\v{p}_{\v{y}} = (\langle \v{y}, \v{\alpha}_1 \rangle, \dots, \langle \v{y}, \v{\alpha}_n \rangle)$. If the norm of $\v{p}_{\v{y}}$ is smaller than the first minimum of the lattice $\mathcal{L}_Q^\perp(\v{x})$, then $\v{p}_{\v{y}}$ must be the zero vector, which confirms that $\v{y} \in \mathcal{L}^\perp(\v{A})$. To guarantee this with high probability, the parameters should satisfy the following bound~\cite{gini2022hardness}
\begin{equation*}
    \log Q > \iota MN + \frac{MN}{2(M-N)} \log M + \frac{N}{2} \log(N\frac{\gamma_{M-N}}{\gamma_N}),
\end{equation*}
where $0 < \iota < 1$ represents the LLL Hermite factor and $\gamma_{(\cdot)}$ represents the Hermite constant with corresponding dimensions, which characterizes the quality of the basis reduction.

  \subsection{Extending the Nguyen-Stern Attack for HLCP}

When transitioning from HSSP to the more general HLCP, where hidden coefficients are drawn from the integer range $\{0, \dots, c\}$, the NS attack framework remains structurally applicable but requires significant parameter adjustments and algorithmic enhancements.

 In the first phase, identifying the orthogonal lattice $\mathcal{L}^\perp(\v{A})$, follows the same procedural logic as the binary case. However, the presence of larger coefficients in $\{0, \dots, c\}$ necessitates a substantially larger modulus $Q$ to ensure that the target vectors remain among the shortest in the lattice. To guarantee the success of the orthogonal attack with non-negligible probability, the bitsize of $Q$ must satisfy a higher theoretical lower bound, specifically
\begin{equation*}
    \log Q > \iota MN + \frac{MN}{2(M-N)} \log M + \frac{N}{2} \log(N\frac{\gamma_{M-N}}{\gamma_N}) + \frac{MN}{M-N}\log c.
\end{equation*}
As $c$ increases, $Q$ must scale exponentially to maintain a sufficiently low density, ensuring that the short vectors found in $\mathcal{L}_Q^\perp(\v{H})$ correctly span the hidden subspace.

In the second phase, recovering the coefficient vectors from the basis of the completed lattice $\bar{\mathcal{L}}(\v{A})$, becomes significantly more challenging. In the HLCP setting, the recovered LLL basis vectors can be much larger than the original vectors. Since the hidden $\v{\alpha}_i$ vectors are still expected to be among the shortest non-zero elements of the lattice, the more powerful BKZ algorithm is required. By providing a superior approximation factor, BKZ enables the adversary to search the lattice more effectively and pinpoint the true hidden coefficients within the expanded search space.

While the standard NS attack is effective, several optimizations~\cite{coron2020polynomial,coron2021provably} have been proposed to improve its computational efficiency, reducing the complexity of the second step from exponential to polynomial time. We provide a brief overview in Appendix~\ref{app.cg_attack}. 

\subsection{Complexity Analysis}

The computational bottleneck of the NS attack resides in the second step, as lattice reduction problems are inherently difficult as the dimension $n$ grows. For the standard HSSP, the heuristic running time is dominated by the exponential term $2^{\Omega(N)}$. 

In the context of the HLCP where coefficients are bounded by $c$, the complexity scales as~\cite[Section 6.4.4]{gini2022hardness}
\begin{equation*}
    2^{\Omega(N)} \cdot \log^{\mathcal{O}(1)} c.
\end{equation*}
This indicates that while the complexity is exponential with respect to the number of honest nodes $n$, it only grows polynomially with the bitsize of the coefficient range $c$.

For the first step, the heuristic polynomial complexities for constructing the basis and performing reduction are approximately $O(N^9)$ and $O(N^7(N + \log c)^2)$, respectively.
\red{

\subsection{Noise-tolerant orthogonal-lattice recovery}\label{sec:noisy_step1}

The reformulation in Section~\ref{sec:finite_precision_reduction} recasts the reconstruction problem as a noisy mHSSP/mHLCP instance. In this section, we briefly show how such instances can be addressed heuristically by lifting the exact orthogonality constraints into an augmented lattice that explicitly carries their residuals~\cite{coron2019cryptanalysis,notarnicola2021hidden}. This extension only modifies the lattice construction in Step~1 of the NS attack.

Let
$\tilde{\v{Y}}_r\in\mathbb{Z}_Q^{M\times r}$
denote the submatrix formed by selecting $r$ of the $u$ columns of
$\tilde{\v{Y}}$. 
We can construct the augmented lattice with row basis
\begin{equation*}
\begin{pmatrix}
\lambda\v{I}_M
&
-\tilde{\v{Y}}_r\\
\v{0}
&
Q\v{I}_r
\end{pmatrix}\in \mathbb Z^{(M+r)\times(M+r)},
\end{equation*}
where $\lambda\in\mathbb{Z}^+$ balances the original and residual
coordinates. 

If we multiply this basis by any
integer row vector
$(\v{v}^{\top}\mid\v{w}^{\top})$, where
$\v{v}\in\mathbb{Z}^{M}$ and
$\v{w}\in\mathbb{Z}^{r}$, we will get the lattice vector
\begin{equation}
\left(
\lambda\v{v}^{\top}
\,\middle|\,
-\v{v}^{\top}\tilde{\v{Y}}_r
+
Q\v{w}^{\top}
\right).
\label{eq:noisy_lattice_vector}
\end{equation}
The first $M$ coordinates represent the candidate left-kernel vector $\v v$, while the last $r$ coordinates measure its modular residual with respect to the noisy observations.

Now consider a genuine integer left-kernel vector $\v v$ of $\tilde{\v W}$, i.e.,
$\v v^\top\tilde{\v W}=\v 0^\top$. By
\eqref{eq:noisy_mhssp_relation}, the selected observations satisfy
\begin{equation*}
\v{v}^{\top}\tilde{\v{Y}}_r
\equiv
\v{v}^{\top}\v{E}_r
\pmod Q,
\end{equation*}
Hence, $\v w$ can be chosen so that the last $r$ coordinates of \eqref{eq:noisy_lattice_vector} are the centered representatives of $-\v v^\top\v E_r$. Provided that $Q$ is sufficiently large to avoid modular wraparound, the corresponding lattice vector becomes
\begin{equation}
\left(
\lambda\v v^\top
\,\middle|\,
-\v v^\top\v E_r
\right).
\end{equation}
Therefore, if the numerical residual is small, a short left-kernel vector of $\tilde{\v W}$ gives rise to a short vector in the augmented lattice.

In contrast, for a vector $\v v$ outside the left kernel of $\tilde{\v W}$, the residual generally contains the additional structural term
$\v v^\top\tilde{\v W}\tilde{\v X}_{h,r}$ and is therefore not expected to remain comparably small. The augmented coordinates thus provide a heuristic separation between genuine left-kernel vectors and unrelated lattice vectors. 

Assuming $\tilde{\v W}$ has full column rank, its left kernel has dimension $M-N$. We retain reduced vectors whose residual coordinates are sufficiently small and whose first-block projections provide $M-N$ linearly independent candidates. We then discard the residual coordinates and divide the first $M$ coordinates by $\lambda$. The resulting vectors are passed to the original completion-lattice and BKZ-based Step~2, which remains unchanged.

}

\section{Optimized HSSP attack for Secure Aggregation}\label{sec.filter}

As established in previous sections, reconstructing private updates in DFL can be formulated as solving an mHSSP or mHLCP. But gaps still exist between theoretical solvability and practical execution.

Firstly, theoretical solvers for the HSSP require more observations than the naive condition $M \geq N$ to guarantee unique recovery. The NS attack typically requires $M = 2N$, while the methods of Caron and Gini require $M \approx N^2$. 
In a DFL core sub-graph, $M$ corresponds to the number of colluding corrupted nodes, while $N$ is the number of honest targets. 
Such a condition is not always satisfied in realistic decentralized networks. When $M$ is small, lattice reduction algorithms produce a candidate vector pool that contains the true weight vectors but also many spurious vectors that satisfy the lattice short-vector criteria but not the physical DFL constraints. Therefore, to reduce the spurious solutions, the adversary can apply filters to extract the correct $N$ weight vectors from the candidate set with different protocols.

Secondly, the filtering is primarily derived from the prior knowledge established in Cases 1-3. By verifying whether each candidate combination is consistent with these topological constraints, we can systematically eliminate part of spurious solutions.
Additionally, the inherent mathematical properties of the DFL mixing matrix $\v{W}$ provide rigorous constraints to prune the HSSP result set.

For undirected graphs, $\v{W}$ is typically doubly stochastic. After scaling by $\beta$, the $N$ chosen candidate vectors $\tilde{\v{w}}_j$ must satisfy the exact row-sum identity as $\sum_{j \in \mathcal{V}_h} \tilde{\v{w}}_j = \beta(\v{1} - \v{W}_{cc}\v{1})$. Any combination of $N$ vectors that does not sum to this known constant is eliminated.
In directed graphs, row-stochasticity is not guaranteed. The sum of $N$ candidate vectors is instead relaxed to a bound $\leq c N$. However, if the protocol utilizes auxiliary mass-balance scalars $a_i^{(t)}$ that are transparent (e.g., initialized to 1), their evolution serves as a powerful secondary filter to validate the weights by checking if the reconstructed scalars match the observed evolution.

Note that the recovered solution set $\{\v{a}_1, \dots, \v{a}_N\}$ is inherently unordered. While the topological prior knowledge in Case 1 and 2 can help establish a partial ordering, a unique sequence is not always guaranteed. This lack of ordering leads to a permutation ambiguity in the columns of the reconstructed matrix $\v{W}_{ch}$. 
Thus while the quality of the reconstructed private data $\v{X}_h$ remains high, the identities of the data, i.e., the specific mapping between a reconstructed model update and its corresponding honest node may be shuffled. This means the adversary may not always be certain who contributed which specific update, unless further identity-linking information is available.

Therefore, we build an end-to-end procedure for the proposed attack to circumvent SA, which is formalized in Algorithm~\ref{alg:math_hssp_attack}.

\begin{algorithm}[ht]
  \caption{Optimized HSSP attack for SA}
  \label{alg:math_hssp_attack}
  \begin{algorithmic}[1]
    \State \textbf{Input:} Adversarial observations $\mathcal{O}$
    \State \textbf{Output:} Recovered honest updates $\hat{\v{X}}_h \in \mathbb{R}^{|\mathcal{V}_h| \times u}$ and corresponding training data batch $\mathcal{B}_i$

    
    \State \textbf{Phase 1: Topological Reduction}\Comment{Case 1 and 2}
    \ForAll{$i \in \mathcal{V}_c$}
    \If{$|\mathcal{N}_i \cap \mathcal{V}_h| = 0$}
    \State $\mathcal{V}_c \leftarrow \mathcal{V}_c \setminus \{i\}$ \Comment{Pruning}
    \EndIf
    \If{$|\mathcal{N}_i \cap \mathcal{V}_h| = 1$ with honest neighbor $j$}
    \State Resolve trivial attack for node $j$ via \eqref{eq.trivial_attack}
        \State $\mathcal{V}_c \leftarrow \mathcal{V}_c \setminus \{i\}, \mathcal{V}_h \leftarrow \mathcal{V}_h \setminus \{j\}$ \Comment{Pruning}
    \EndIf
    \EndFor
    \State Build up core sub-graphs and calculate corresponding $\v{Y}$
    
    \vspace{3pt}
    \State \textbf{Phase 2: Orthogonal Lattice Construction}
    \State Map to mHSSP or mHLCP: $\v{H} \equiv 10^{\gamma}\beta\v{Y} \pmod Q$, with unknown coefficients $\v{A} \equiv \beta\v{W}_{ch} \pmod Q$
    \State \textbf{Execute NS attack:}
    \State Extract $\mathcal{L}^{\perp}(\v{A})$ via LLL reduction of $\mathcal{L}_Q^{\perp}(\v{H})$
    \State Compute completion lattice $\bar{\mathcal{L}}(\v{A}) = (\mathcal{L}^{\perp}(\v{A}))^{\perp} \cap \mathbb{Z}^M$
    \State Execute BKZ on $\bar{\mathcal{L}}(\v{A})$ to extract short independent vectors $\mathcal{S}_{cand} \subset \mathbb{Z}^M$
    \vspace{3pt}
    \State \textbf{Phase 3: Filtering}
    \State $\Omega_{valid} \leftarrow \emptyset$
    \ForAll{$\v{P} \in \binom{\mathcal{S}_{cand}}{|\mathcal{V}_h|}$}
        \If{$\v{P}$ violates topology prior and constraints}
            \State \textbf{continue} 
        \EndIf
        \State $\Omega_{valid} \leftarrow \Omega_{valid} \cup \{\beta^{-1}\v{P}\}$
    \EndFor
    \vspace{3pt}
    \State \textbf{Phase 4: State Inversion}
    \For{$\hat{\v{W}}_{ch} \in \Omega_{valid}$}
        \State $\hat{\v{X}}_h^{(t+1/2)} \leftarrow (\hat{\v{W}}_{ch})^\dagger \v{Y}$ 
        \State Use downstream GIA to recover training data
    \EndFor
  \end{algorithmic}
\end{algorithm}

\section{Numerical Results}\label{sec.exp}
In this section, we present numerical evaluations to demonstrate the effectiveness of the proposed attack and its practical implications in DFL.~\footnote{The code is available at: \url{https://github.com/Wenrui-Yu/hssp_dfl_sa}} The experimental setup is given in Appendix~\ref{ssec.setup}.

\subsection{Evaluation Metrics}\label{ssec.metrics}

Let $\mathcal{S}_{gt}$ be the set of ground-truth (GT) column vectors constituting the hidden weight matrix $\tilde{\v{W}}_{ch} \in \mathbb{Z}^{M \times N}$ and $\mathcal{S}_{cand}$ be the set of candidate vectors recovered via lattice reduction.

We define the success of the SA breach through the following metrics:

\noindent\textbf{Recall: }
 An attack achieves a primary success if the GT set is a subset of the candidate set, i.e., $\mathcal{S}_{gt} \subseteq \mathcal{S}_{cand}$.
This indicates that the lattice reduction has successfully captured all necessary basis vectors of the hidden weights.

\noindent\textbf{Solution Set Cardinality (Ambiguity): }
 We define the set of feasible solutions (filtered with different cases) as
    \begin{equation*}
        \mathcal{F} = \{ \mathcal{P} \subset \mathcal{S}_{cand} : |\mathcal{P}| = N \text{ and } \mathcal{P} \text{ satisfies } \mathcal{O} \},
    \end{equation*}
    where $\mathcal{P}$ is a combination of the vectors in $\mathcal{S}_{cand}$. The attack potency is measured by the cardinality $|\mathcal{F}|$. If the attack succeeds and $|\mathcal{F}| = 1$, the adversary has uniquely identified the true matrix. 

Since the recovered solution set is inherently unordered, so identifying the correct combination of vectors, regardless of their sequence, is regarded as a successful reconstruction. This permutation leads to an identity ambiguity among honest nodes but does not compromise the fidelity of the reconstructed data itself.

Once a candidate matrix $\hat{\v{W}}_{ch}$ is selected, the adversary reconstructs the honest updates $\hat{\v{X}}_h$ via \eqref{eq.sub_problem} and subsequently the raw data $\hat{\v{D}}_h$. We utilize common evaluation metrics to evaluate the reconstruction quality, e.g., \textit{Peak Signal-to-Noise Ratio (PSNR)} and \textit{SSIM}~\cite{ssim} to evaluate the structural and visual similarity of reconstructed images, \textit{Mean Square Error (MSE)} for the feature accuracy of tabular data, and \textit{Cosine Similarity} over sentence embeddings alongside \textit{ROUGE-L}~\cite{lin2004rouge} and \textit{BERTScore}~\cite{bert-score} to quantify the semantic and structural alignment of recovered text sequences.

We first present the attack's performance in recovering the weight matrix $\v{W}_{ch}$ under various network topologies, which focuses on the fundamental solvability of the mHSSP/mHLCP formulation. We then apply the recovered gradients to reconstruct private training samples. 
The primary objective of the HSSP-based attacks is the accurate inference of the weights $\v{W}_{ch}$ and then correspondingly disentangle the mixed gradients. Thus, the quality of the final data reconstruction depends predominantly on the performance of the downstream gradient inversion attacks rather than the HSSP process itself. The data reconstruction results serve to validate the downstream utility of our recovered gradients.

\subsection{Topological Simplification}\label{exp.graph}

We first empirically explore the graph simplification logic discussed in Section~\ref{ssec.graph_sim}. Recall that honest nodes are categorized into three classes: those directly recoverable via trivial attacks, those constituting the core sub-graph susceptible to HSSP-based inference, and those remaining unrecoverable within a single iteration. We conducted a statistical analysis across a large ensemble of randomly generated undirected graphs and strongly connected directed graphs with fixed node counts $n$ and edge counts $e$. The proportions of each node category, averaged over $100$ independent trials, are illustrated in Figure~\ref{fig:undirected_recovery} and Figure~\ref{fig:directed_recovery} in Appendix~\ref{app.table}.

\vspace{2pt}
\begin{mdframed}
Key Takeaway: Topological sparsity and the corruption ratio $\eta$ are the primary catalysts for privacy leakage. Specifically, sparse graphs lack the path redundancy necessary to mask honest updates at low $\eta$.
\end{mdframed}
\vspace{2pt}

For undirected graphs, the results indicate that as the corruption ratio $\eta$ increases, the proportion of honest nodes recoverable via trivial attacks rises consistently. This trend occurs because a higher density of corrupted nodes more easily creates "isolation" for honest nodes. Also, within the $\eta = 0.4 \sim 0.6$ interval, the proportion of unrecoverable nodes undergoes a precipitous decline. This sharp transition signifies the point where the adversary set $\mathcal{V}_c$ begins to outnumber the honest set $\mathcal{V}_h$ in the core sub-graphs, frequently satisfying the necessity $|\mathcal{V}_c| \geq |\mathcal{V}_h|$ for reconstruction.

For directed graphs, the vulnerability is heavily influenced by the sparsity of the topology. In sparse directed configurations, as shown in Figure~\ref{fig:directed_recovery}(a), honest nodes are highly susceptible to isolation even at low corruption ratios. Since directional communication requires specific in-flow paths that are less redundant in sparse configurations, a significant portion of honest nodes becomes either trivial-attack targets or HSSP-eligible even when $\eta$ is small. Conversely, as the graph becomes denser, the trend aligns more closely with that of undirected graphs, as seen in Figure~\ref{fig:directed_recovery}(d), which the unrecoverable node proportion remains high until the $\eta = 0.4 \sim 0.6$ threshold is reached.

\begin{figure}[ht]
    \centering
    \begin{minipage}{0.23\textwidth}
        \centering
        \includegraphics[width=\linewidth]{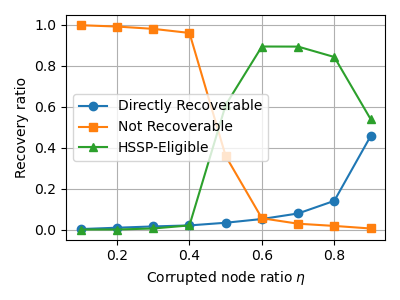}
        \\[-1ex]{\small (a) $n=20, e=40$}
    \end{minipage}
    \begin{minipage}{0.23\textwidth}
        \centering
        \includegraphics[width=\linewidth]{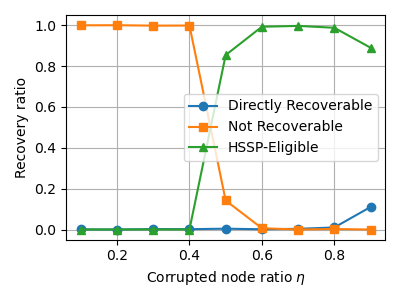}
        \\[-1ex]{\small (b) $n=20, e=60$}
    \end{minipage}
    \caption{
    Topological vulnerability analysis in undirected graphs as a function of the corrupted node ratio $\eta$. The total proportion of honest nodes is partitioned into three mutually exclusive categories: (i) nodes directly recoverable via trivial attacks (blue line), (ii) nodes constituting the core sub-graph targeted for HSSP-based inference (green line), and (iii) inaccessible nodes that remain hidden within a single iteration (orange line). Note that these three categories are exhaustive, with their sum normalized to 1.}
    \label{fig:undirected_recovery}
\end{figure}

\red{Appendix~\ref{exp.large_graph} examines how the sizes of the core subgraphs produced by the simplification strategy scale with the size of the underlying network. These results inform the choice of subgraph sizes in the subsequent experiments. In addition, to assess the dependence of attack feasibility on network topology, we also additionally evaluate several other representative graphs.}

 \subsection{Attack Performance on Undirected Graph}\label{exp.undirected}

Following the topological simplification discussed in the previous subsection, we evaluate the effectiveness of the mHSSP/mHLCP attacks on the resulting core sub-graphs. To ensure a fair comparison, we randomly generate core sub-graphs with a fixed number of honest nodes and corrupted nodes, specifically targeting those identified as potentially resolvable through lattice-based inference. 

\vspace{2pt}
\begin{mdframed}
Key Takeaway: Ambiguity in the recovered solution set is significantly mitigated by topological prior knowledge (Cases 1–3) and adversarial density.
\end{mdframed}
\vspace{2pt}

To evaluate the robustness and efficiency of the proposed attack, we conducted $100$ independent trials for different parameter configurations, as shown in Table~\ref{tab.large_hlcp_summary}. In each trial, a new core subgraph was randomly generated, and we recorded the successful inclusion of all GT vectors. The results demonstrate a high vector inclusion rate, confirming that the recovered candidate pool consistently contains the complete set of true weights. The cardinality of the candidate vector pool grows significantly as the ratio of honest nodes increases. In sparse-observation settings, lattice reduction produces a larger number of spurious short vectors.
\begin{table}[ht]
\centering
\begin{tabular}{l c c }
\hline
Config & Recall (\%) & Avg. Found Vec. \\
\hline
$n=10, e=20, \eta=0.6$ & 98.0\% & 1497.1 \\
$n=10, e=20, \eta=0.7$ & 92.0\% & 88.0  \\
$n=10, e=20, \eta=0.8$ &100.0\% & 17.7   \\
$n=20, e=40, \eta=0.7$ & 72.0\% & 4052.7 \\
$n=20, e=40, \eta=0.8$ & 97.0\% & 328.4 \\
\hline
\end{tabular}
\caption{mHLCP attack performance with different configurations in undirected graphs.}
\label{tab.large_hlcp_summary}
\end{table}

Table~\ref{tab.hlcp_6} and Table~\ref{tab.hlcp_7} of Appendix~\ref{app.table} provide a granular look at the results for the first 10 subgraphs under the configuration $n=10, e=20, \eta=0.6$ and a higher corruption ratio ($\eta=0.7$), showcasing the impact of the filtering process. We observe that the precision of the final reconstruction is heavily contingent on the level of topological awareness (Cases 1-3). Specifically, the availability of neighbor identities (Case 1) or node degrees (Case 2) provides powerful structural constraints that effectively prune spurious vector combinations. This reduction in the candidate set size significantly mitigates solution ambiguity.
Also, a higher density of adversarial nodes grants the attacker more heterogeneous and asymmetric observational viewpoints within the network. This increase in independent constraints significantly narrows the number of valid combinations, often leading to the identification of a unique, correct weight matrix in Table~\ref{tab.hlcp_7}.

For the mHSSP scenario involving uniform weights in undirected graphs, we set the edge weights to $W_{ij} = \frac{1}{d_{\max}+1}$ for $(i,j)\in\mathcal{E}$ and the self-weights to $W_{ii} = 1 - \sum_{j\in\mathcal{N}_i} w_{ij}$ for $i\in\mathcal{V}$, where $d_{\max}$ is the maximum degree of the whole graph. The corresponding results are detailed in Tables~\ref{tab.large_hssp_summary}, \ref{tab.hssp_6} and \ref{tab.hssp_7} in the Appendix~\ref{app.table}. Case 1 is not applicable to the mHSSP formulation. If the adversary possesses exact knowledge of neighbor identities and the protocol employs uniform weights, the subset-sum problem becomes trivial, as the weights for each neighbor are fixed and known beforehand, leaving no hidden parameters to resolve.

The results on directed graphs are presented in Appendix~\ref{exp.directed}.

\subsection{Attack on Real Dataset}\label{exp.true_dataset}

\vspace{2pt}
\begin{mdframed}
Key Takeaway: Information leakage occurs even in the absence of a unique solution, as some spurious candidates still leak recognizable semantic information.
\end{mdframed}
\vspace{2pt}

Figure~\ref{fig:img_re} illustrates representative reconstruction results for the image dataset, showcasing both the GT solution (specifically, the 24th entry in Case 1) and instances of spurious solutions recovered under Cases 1-3. Detailed results for tabular and text modalities, including the set of 30 random candidate reconstructions for each case, are provided in Appendix~\ref{app.complete}. All experiments were conducted on an undirected graph with non-uniform mixing weights (mHLCP) under the configuration $n=10, e=20, \eta=0.6$.
For each candidate solution identified by the HSSP solver, we first resolve the corresponding latent model updates via~\eqref{eq.sub_problem} and subsequently perform gradient inversion to recover the raw data. As observed in Figure~\ref{fig:img_re}, while the GT solution yields a better reconstruction fidelity, several spurious candidates also produce outputs that closely resemble the original training samples. This highlights that, even in the absence of a unique solution, an adversary can still obtain approximations of private data that retain significant semantic information.

\begin{figure}[h]
  \centering
  \includegraphics[width=0.94\linewidth]{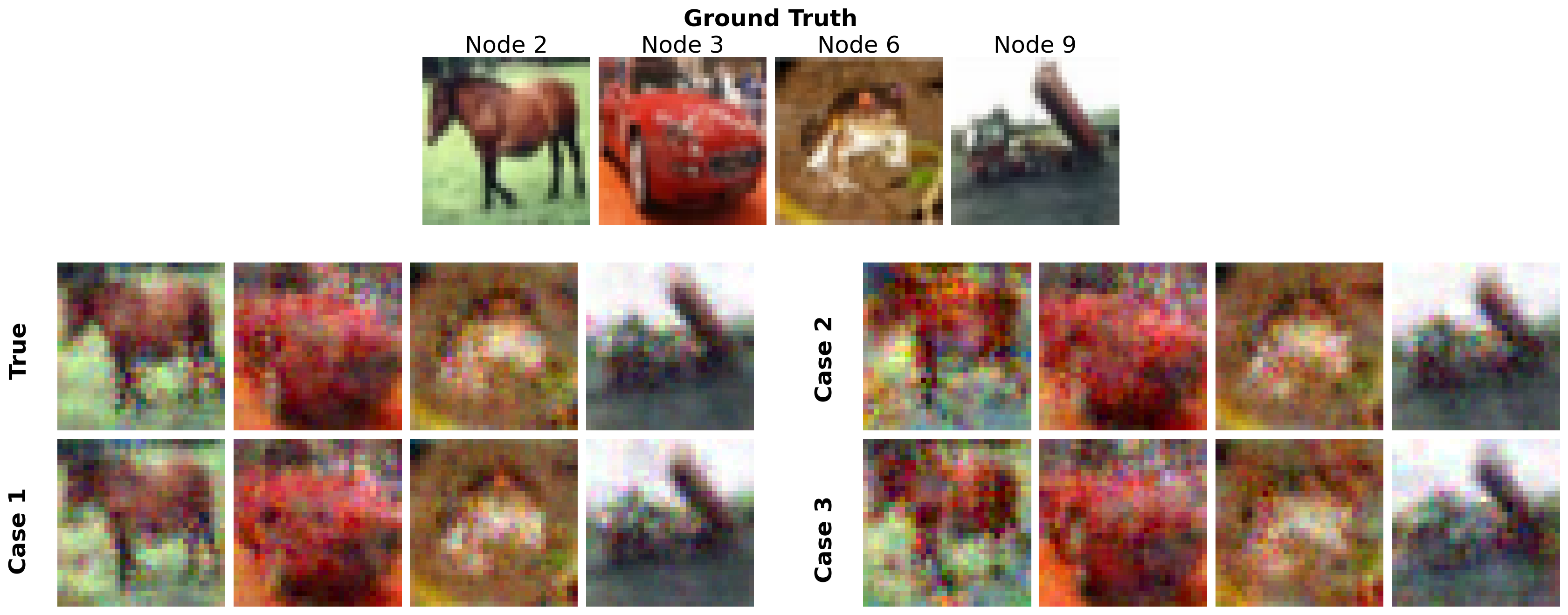}
  \caption{Visual comparison of reconstructed images derived from the GT solution and representative spurious candidates across Cases 1-3.}
  \Description{}
  \label{fig:img_re}
\end{figure}

The final reconstruction quality also depends on the effectiveness of the downstream inversion algorithm. In settings where gradient inversion is theoretically lossless, such as in logistic regression models for text data, the private input vectors can be perfectly recovered. For example, a representative reconstruction in Table~\ref{tab:sent_recon} achieves a cosine similarity of 1.0. However, inherent limitations in the subsequent conversion of embeddings back to natural language via \textit{vec2text} still introduce semantic distortions.
\begin{table*}[t]
\small
\begin{tabular}{cccp{8cm}cc}
\hline
& Node & Cosine Similarity & Recovered Text & Rouge-L & BERTScore \\
\hline
\multirow{4}{*}{Ground Truth} 
& 2 & - & Crunch berries! im tired. Who wants to do something tomorrow? & -& -\\
& 3 & - & @LAmale really? They're a band... good music! & - & -\\
& 6 & - & Where are all the banana slugs? & - & -\\
& 9 & - & @Jack\_thm aww why not?! Heck they do the job! I can't find those anywhere nemore either!! ur ipod ones hurt!!! & -& -\\
\hline
\multirow{4}{*}{Spurious solution} 
& 2 & 0.714 & \textcolor{diffred}{@Lammahey they're a good band!} & 0.000&-0.346 \\
& 3 & 0.956 & \textcolor{diffred}{@Lamem} \textcolor{matchgreen}{They're a} \textcolor{diffred}{good} \textcolor{matchgreen}{band}\textcolor{diffred}{!} & 0.571&0.255 \\
& 6 & \textbf{1.000} & \textcolor{matchgreen}{Where are all the banana slugs?} &\textbf{1.000} &\textbf{1.000}\\
& 9 & \textbf{1.000} & \textcolor{diffred}{mmmm Why} \textcolor{matchgreen}{can't} \textcolor{diffred}{anyone} \textcolor{matchgreen}{find} \textcolor{diffred}{them They} \textcolor{matchgreen}{hurt} & 0.071&-0.031\\
\hline

\multirow{4}{*}{True solution} 
& 2 & \textbf{1.000} & \textcolor{matchgreen}{Crunch berries!} \textcolor{diffred}{i am} \textcolor{matchgreen}{tired. Who wants to do something tomorrow?} & 0.947 &0.811\\
& 3 & \textbf{1.000} & \textcolor{diffred}{@Lameah...} \textcolor{matchgreen}{They're a} \textcolor{diffred}{really} \textcolor{matchgreen}{good} \textcolor{diffred}{band!} \textcolor{matchgreen}{music} & 0.625 &-0.474\\
& 6 & \textbf{1.000} & \textcolor{matchgreen}{Where are all the banana slugs?} &\textbf{1.000} & \textbf{1.000} \\
& 9 & \textbf{1.000} & \textcolor{diffred}{mmmm Why} \textcolor{matchgreen}{can't} \textcolor{diffred}{you} \textcolor{matchgreen}{find} \textcolor{diffred}{them} & 0.276 &-0.499\\
\hline
\end{tabular}
\caption{Example of text reconstruction results from candidate weight matrices. Exact matches and matching substrings are highlighted in \textcolor{matchgreen}{green}, while mismatched errors are highlighted in \textcolor{diffred}{red}.}
\label{tab:sent_recon}
\end{table*}

The statistical results are illustrated in Figure~\ref{fig:stat_cifar}. For the true solution, the reconstructed gradients are identical to the ground truth (with a negligible machine epsilon of $10^{-32}$).
In contrast, the average MSE for the sampled spurious candidates is significantly higher, leading to a corresponding decrease in image reconstruction quality. These statistical metrics exhibit high variance, reflecting the significant heterogeneity in reconstruction fidelity across the candidate set. While some lead to evident reconstruction failures, a substantial portion still yields high-quality approximations where semantic information remains clearly discernible.

\begin{figure}[h]
  \centering
  \includegraphics[width=0.86\linewidth]{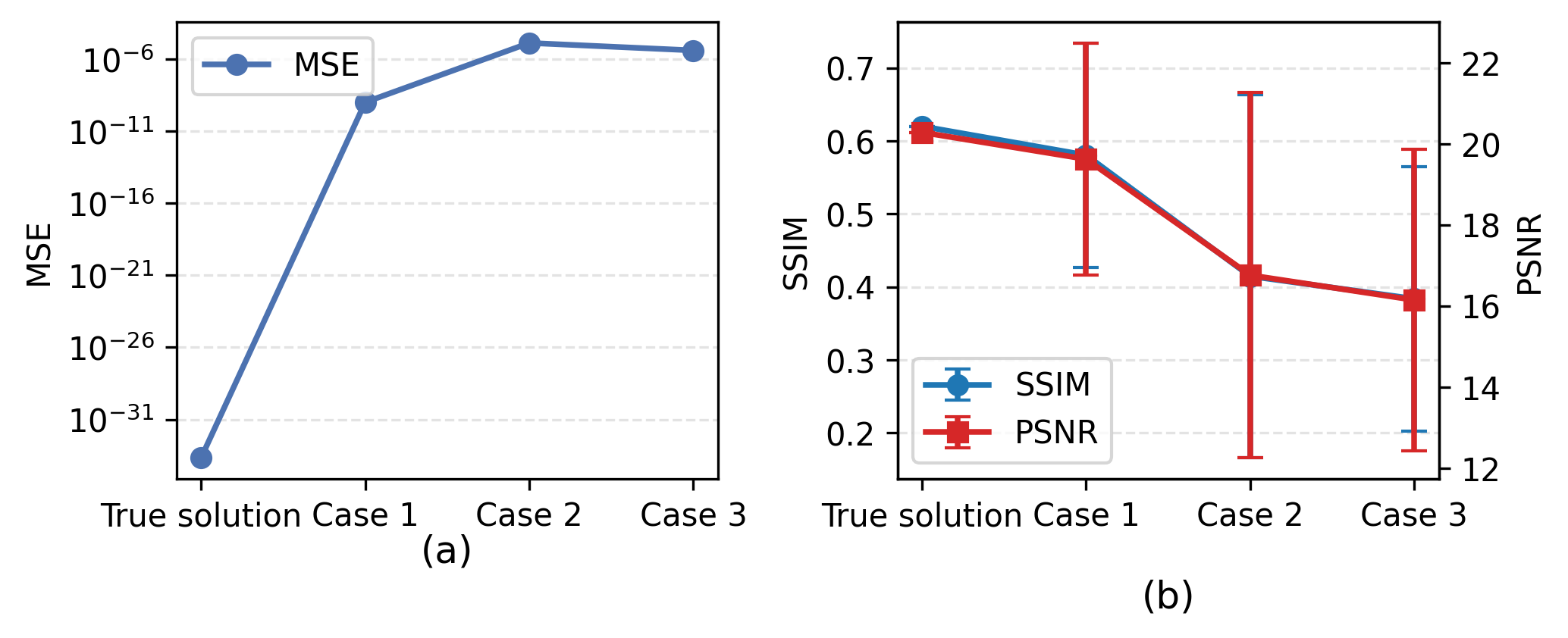}
  \caption{Statistical reconstruction performance on the CIFAR-10 dataset. (a) MSE between the HSSP-reconstructed gradients and the GT gradients. (b) Distribution of SSIM and PSNR scores for training data reconstructed from the sampled candidate gradients.}
  \Description{}
  \label{fig:stat_cifar}
\end{figure}

\red{
\subsection{Robustness to Finite-Precision Errors}\label{sec:finite_precision_experiments}

We also evaluate whether numerical errors introduced by integerization
degrade the exact-arithmetic attack. All other experimental settings are
kept the same as in Section~\ref{exp.true_dataset}; only the observation and
the Step~1 construction are varied. 

For clarity and to facilitate comparison with the idealized setting,
the original values are stored in float64 format and quantized to ten decimal places before integerization. Based on that, we conduct the
finite-precision evaluation by applying truncation with $\gamma\in\{4,6,8,10\}$ and vary the
number of selected observation columns as $r\in\{1,2,4,8\}$ in noisy mHLCP. For baseline, we apply the exact Step~1 to the exact ideal observation using $r=4$. We heuristically set
$\lambda=\lceil 1.25(\beta+1)\sqrt{r}\rceil$.
The factor $\sqrt{r}$ accounts for the growth of the Euclidean norm
across the $r$ residual coordinates.

\vspace{2pt}
\begin{mdframed}
Key Takeaway: When sufficient numerical precision is retained, the noisy mHSSP/HLCP attack is robust to finite-precision truncation.
\end{mdframed}
\vspace{2pt}

Table~\ref{tab:truncation-single-trial} reports the recovery results
under finite-precision truncation. The integerization precision
$\gamma$ has a clear impact on recovery: a smaller $\gamma$ retains
fewer decimal places, producing a coarser fixed-point representation
and a larger truncation error in the original numerical domain.
In our experimental setting, when $\gamma=8$, selecting $r=4$, as in
the exact mHLCP control, is sufficient to tolerate the resulting
numerical error for all tested datasets. Both the column recall and the
number of retained candidate vectors match those of the exact control,
indicating that the attack effectiveness is preserved at this level of
finite-precision truncation.

\begin{table*}[ht]
\centering
\begin{tabular}{llccccc}
\toprule
Dataset & & & $r=1$ & $r=2$ & $r=4$ & $r=8$ \\
\midrule
\multirow{5}{*}{Sentiment140} & Baseline & \textbf{1.00 / 1152} & \text{--} & \text{--} & \text{--} & \text{--} \\
& $S=10^{4}$ & \text{--} & 0.50 / 4096 & 0.50 / 4096 & 0.50 / 4096 & 0.50 / 4096 \\
& $S=10^{6}$ & \text{--} & 0.00 / 1068 & \textbf{1.00 / 1152} & \textbf{1.00 / 1152} & \textbf{1.00 / 1152} \\
& $S=10^{8}$ & \text{--} & \textbf{1.00 / 1152} & \textbf{1.00 / 1152} & \textbf{1.00 / 1152} & \textbf{1.00 / 1152} \\
& $S=10^{10}$ & \text{--} & \textbf{1.00 / 1152} & \textbf{1.00 / 1152} & \textbf{1.00 / 1152} & \textbf{1.00 / 1152} \\
\addlinespace
\multirow{5}{*}{Purchase} & Baseline & \textbf{1.00 / 1152} & \text{--} & \text{--} & \text{--} & \text{--} \\
& $S=10^{4}$ & \text{--} & 0.50 / 4096 & 0.50 / 4096 & 0.50 / 4096 & 0.50 / 4096 \\
& $S=10^{6}$ & \text{--} & 0.50 / 4096 & 0.50 / 4096 & \textbf{1.00 / 1152} & 0.25 / 1776 \\
& $S=10^{8}$ & \text{--} & 0.50 / 4096 & 0.50 / 4096 & \textbf{1.00 / 1152} & \textbf{1.00 / 1152} \\
& $S=10^{10}$ & \text{--} & 0.50 / 4096 & 0.50 / 4096 & \textbf{1.00 / 1152} & \textbf{1.00 / 1152} \\
\addlinespace
\multirow{5}{*}{CIFAR} & Baseline & \textbf{1.00 / 1152} & \text{--} & \text{--} & \text{--} & \text{--} \\
& $S=10^{4}$ & \text{--} & 0.50 / 4096 & 0.50 / 4096 & 0.50 / 4096 & 0.50 / 4096 \\
& $S=10^{6}$ & \text{--} & 0.50 / 4096 & 0.25 / 2752 & 0.00 / 1776 & 0.50 / 4096 \\
& $S=10^{8}$ & \text{--} & 0.50 / 4096 & \textbf{1.00 / 1152} & \textbf{1.00 / 1152} & \textbf{1.00 / 1152} \\
& $S=10^{10}$ & \text{--} & 0.50 / 4096 & \textbf{1.00 / 1152} & \textbf{1.00 / 1152} & \textbf{1.00 / 1152} \\
\bottomrule
\end{tabular}
\caption{\red{Robustness to finite-precision errors via noisy mHSSP/mHLCP. Each entry is reported as
$\text{Recall}/\text{Found Vec.}$}}
\label{tab:truncation-single-trial}
\end{table*}

}

\subsection{Defense Mechanisms}

\vspace{2pt}
\begin{mdframed}
Key Takeaway: Lattice-based attacks exhibit algebraic sensitivity to aggregation-level DP while local DP does not perturb the global structure.
\end{mdframed}
\vspace{2pt}

To evaluate the robustness of our HSSP-based attack, we further investigate the impact of DP mechanisms, as DP is commonly adopted to provide privacy guarantees and can be combined with SA~\cite{stevens2022efficient}. DP can typically be applied in two paradigms: Local DP (LDP), where noise is injected into individual gradients before aggregation, and Aggregation-level DP (analogous to Centralized DP, or CDP), where noise is added to the final aggregated result.

\red{
When noise is applied to the aggregated sum, the observation matrix $\hat{\v{H}}$ available to the adversary is perturbed by noise $\v{E}$. The mHSSP relation in \eqref{eq.mhssp} becomes $\hat{\v{H}} \equiv \v{A}\v{X} + \v{E} \pmod Q$, which takes the similar noisy modular form as in Section~\ref{sec:finite_precision_reduction}. 
Conversely, when nodes apply LDP, the noise matrix $\v{E}$ is added directly to the local parameters $\v{X}$ prior to the SA protocol. The observed aggregated system becomes $\hat{\v{H}} \equiv \v{A}(\v{X} + \v{E}) \pmod Q$.
Let $\hat{\v{X}} = \v{X} + \v{E}$ represent the perturbed local updates. The system can be rewritten as $\hat{\v{H}} \equiv \v{A}\hat{\v{X}} \pmod Q$. This preserves the exact structural invariant of the mHSSP/mHLCP formulation. 
Therefore, LDP does not directly disrupt the algebraic structure exploited by the lattice reduction algorithm, which can still recover $\v{A}$ under the considered settings. In this scenario, the protection of the raw data primarily depends on the magnitude of the LDP noise $\v{E}$, while the structural protection provided by SA against recovering the mixing coefficients remains limited.
}

To empirically validate this algebraic brittleness, we distribute $5000$ training images to each of the $10$ nodes. We evaluate the model's utility degradation and the attack's success rate when applying LDP and CDP during training. The DP parameters are configured with a clipping norm of $1.0$, $\delta = 10^{-5}$, and varying privacy budgets $\epsilon$s.
As illustrated in Figure~\ref{fig:dp}, in the CDP scenario, aggregate-level unknown perturbation disrupts the exact algebraic relation required by our lattice attack. However, in the LDP scenario, the injected noise presents no structural obstacle to the solver; while model utility predictably declines as the noise scale increases, the underlying weights and gradients remain resolvable. These results demonstrate that aggregation-level perturbations are effective under our exact-arithmetic attack model.
\begin{figure}[h]
  \centering
  \includegraphics[width=0.93\linewidth]{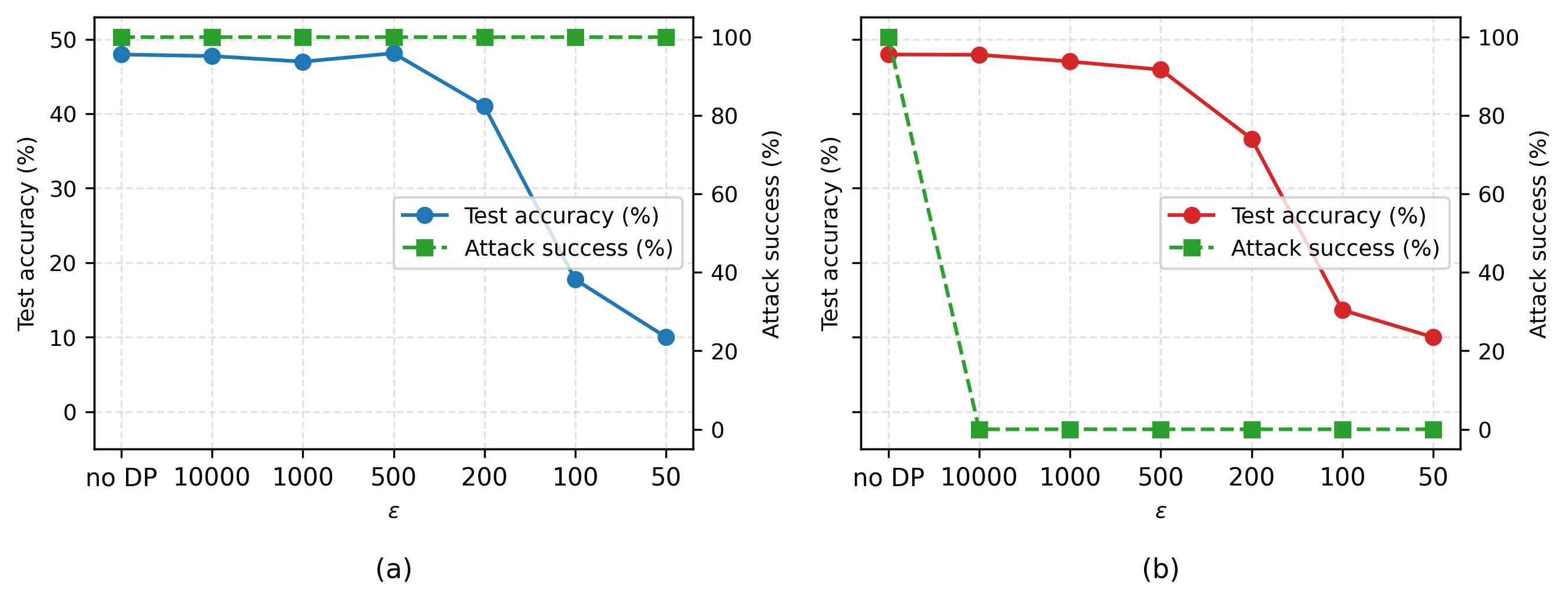}
  \caption{Test accuracy and HSSP attack success rate under different DP paradigms. (a) LDP and (b) CDP.}
  \Description{}
  \label{fig:dp}
\end{figure}

Moreover, although the noisy mHSSP/mHLCP construction is robust to bounded truncation errors as shown in Section~\ref{sec:finite_precision_experiments}, CDP remains effective against full reconstruction under our current experimental setting. We replace the original exact Step~1 with its noise-tolerant variant while leaving the remaining attack pipeline unchanged. As shown in Figure~\ref{fig:aggregation_dp_noisy} in Appendix~\ref{app.table}, the exact Step~1 recovers no ground-truth vector, whereas the noisy attack finds $25-50\%$ true vectors under CDP. This partial recovery is insufficient to complete the attack, but it shows that the noisy construction can retain partial structural information that is entirely lost in the exact formulation. We regard this as a preliminary result, suggesting that the interaction between DP perturbation and more advanced noise-tolerant lattice attacks warrants further investigation.

\section{Related Works}\label{sec.related}

\begin{table*}[ht]
\centering
\small
\begin{tabularx}{\textwidth}{l c c c l l  X}
\toprule
\textbf{Work} & \textbf{Perspective}& \textbf{Setting} & \textbf{SA} & \textbf{Weight} & \textbf{Rounds}  & \textbf{Key Mechanism \& Insight} \\ 
\midrule
Pasquini \cite{pasquini2023security} & attack & DFL & partial & known & multi  & Exploits inherent leakage in DFL caused by local generalization. \\ 
\hline
Mrini  \cite{mrini2024privacy} & attack & DFL & no  & known & multi  & Leverages the quasi-static nature of gradients across iterations. \\ 
\hline
Dekker \cite{dekker2025topology} & attack & \makecell[c]{summation} & yes & known\footnotemark & multi  & Resolves secure summation via repeated asynchronous observations. \\ 
\hline
So \cite{so2023securing} & defense & CFL & yes & known & multi  & Constrains user participation combinations to prevent leakage. \\ 
\hline
Pasquini \cite{pasquini2022eluding} & attack& CFL & yes & known & one  & Exploits active server control to inject non-identical models. \\ 
\hline
Zhang \cite{zhang2026information} & defense & \makecell[c]{DFL} & yes & known & one  & Establishes perfect security bounds for simple sums in fully connected DFL. \\
\hline
Lam \cite{lam2021gradient} & attack & CFL & yes & known & multi & Reconstructs binary participant matrix via side-channel analytics to disaggregate updates. \\
\hline
Boenisch \cite{boenisch2023reconstructing} & attack & CFL & yes & known & one & Injects malicious devices to nullify SA masking. \\
\hline
Marchand \cite{marchand2023sratta} & attack & CFL & yes & known & multi & Exploits neuron activation sparsity and discrete data priors to disentangle aggregated updates. \\
\hline
Ours & attack & DFL & yes & unknown & one  & Exploits structural asymmetry to resolve unknown weights via mHLCP. \\
\bottomrule
\end{tabularx}
\caption{Comparison with related works.}
\label{tab:related_works}
\end{table*}

SA has been extensively studied in centralized FL to ensure that the central server only observes the aggregate of client updates, thereby mitigating individual privacy leakage \cite{bonawitz2017practical, bell2020secure, so2022lightsecagg} to some extent. This paradigm has naturally been extended to decentralized networks. For instance, several studies adapt distributed protocols to achieve efficient secure summation for DFL \cite{zec2024efficient, tjell2020private}, often leveraging Secure Multi-Party Computation to safeguard single-node contributions \cite{danner2018robust, kanagavelu2020two, tran2021efficient}.

However, a distinction must be made regarding the communication topology of these protocols. Certain existing works \cite{pereira2024secure, jeon2021privacy, brunetta2021non} implement SA in peer-to-peer environments but remain fundamentally predicated on the centralized FedAvg logic.
These designs require each node to broadcast information to all other participants to achieve global synchronization. This deviates from the mainstream DFL paradigm, which prioritizes localized communication and single-hop aggregation exclusively with direct neighbors. Since our study focuses on the privacy vulnerabilities inherent in sparse, neighborhood-based aggregation models, these global-broadcasting approaches fall outside our immediate scope.

Meanwhile, several related works aim to recover individual private information from CFL or DFL aggregation processes. A summary is provided in Table~\ref{tab:related_works}. In addition, we discuss the intrinsic connections between our approach and prior works such as \cite{mrini2024privacy} and \cite{dekker2025topology} in Appendix~\ref{ssec.similar_works}.

\footnotetext{For works employing uniform weights, the weights are treated as a priori known.}

\section{Limitations}
\label{sec:scope_limitations}

The current attack remains subject to several limitations. Its strongest exact-update data-reconstruction result applies to the first training round, where the common initialization is known. Moreover, the basic identifiability condition requires that the number of independent corrupted observations in the core be at least as large as the number of unknown honest states. 
The attack also relies on sufficient prior topology information to filter spurious lattice candidates. 
The purpose of this work is not to suggest that every DFL deployment satisfies the demonstrated operating conditions, but to identify a topology-induced privacy risk that can arise in neighborhood-based SA. As the corresponding theory and attack techniques develop, operating regimes that are currently difficult may become more accessible, motivating early consideration of these structural risks in protocol design.

\section{Conclusion}\label{sec.conclusion}

In this paper, we present a systematic investigation of privacy vulnerabilities in SA within DFL environments. By formulating the privacy reconstruction problem under SA as instances of mHSSP and mHLCP, we show that structural asymmetries in DFL topologies can render these otherwise challenging problems tractable. Empirical results across multiple data modalities further demonstrate that colluding adversaries can recover individual updates with high fidelity, and consequently, reconstruct the underlying training data. These results show that, in asymmetric topologies, SA alone may not provide sufficient structural protection against inference attacks, as individual contributions can be disentangled under favorable conditions. This highlights a gap between the intended privacy guarantees of SA and its behavior in decentralized settings.

\begin{acks}
This paper was edited for grammar using Gemini and ChatGPT.
We thank the Aalborg University AI:X initiative for enabling this work via the AI:SECURITY lab.
\end{acks}

\bibliographystyle{ACM-Reference-Format}
\bibliography{ref}

\clearpage
\appendix
\onecolumn

\section{Push-Sum Algorithm in Directed Graphs}\label{app.pushsum}
For scenarios involving directed graphs, $\v{W}^{(t)}$ is typically required to be column-stochastic (i.e., $\v{1}^\top \v{W}^{(t)} = \v{1}^\top$).
Specifically, for Push-Sum algorithm, each node $i$ maintains two variables, $\v{x}_i^{(t)}$ and an auxiliary scalar $a_i^{(t)}\in\mathbb{R}^+$. The actual consensus model at node $i$ is then computed by the ratio $\tilde{\v{x}}_i^{(t)} = \frac{\v{x}_i^{(t)}}{a_i^{(t)}}$ to correct the bias induced by the asymmetry of the graph.
This pair can be interpreted as an augmented state $\v{z}_i^{(t)} = [\v{x}_i^{(t)\top}, a_i^{(t)}]^\top \in \mathbb{R}^{d+1}$, where both the model parameters and the auxiliary scalar evolve under the same mixing dynamics, i.e., $\v{Z}^{(t+1)} = \v{W}^{(t+1)} \v{Z}^{(t+\frac{1}{2})}$.~\footnote{The convergence target of Push-Sum is determined by the initialization of the auxiliary scalar $a_i^{(0)}$. In existing literature \cite{assran2019stochastic, li2025asymmetrically}, the auxiliary scalar is typically initialized as $a_i^{(0)} = 1$ for all $i \in \mathcal{V}$, leading to a standard arithmetic average. Adopting a non-uniform initialization for $a_i^{(0)}$ will enable the network to reach a weighted consensus.}
Consequently, regardless of whether a standard consensus or a Push-Sum-based mechanism is employed, the fundamental update paradigm remains structurally consistent with the linear mixing form in \eqref{eq.aggr}. For the sake of brevity and clarity, we will adopt this unified representation in the subsequent sections of this paper to describe the general decentralized aggregation process.~\footnote{if $a_i^{(0)}$ is fixed to a non-encrypted value (e.g., $1$), the evolution of $a_i^{(t)}$ becomes a deterministic trajectory governed solely by the graph topology. This provides an additional algebraic constraint that can be utilized to prune the solution space or filter out legitimate candidates in the HSSP-based reconstruction attack, which is detailed in Section~\ref{sec.filter} and Appendix~\ref{exp.directed}.}

\section{Prior Knowledge}\label{app.cases}
\textbf{Case 1} is common in decentralized peer-to-peer systems. 
To implement SA protocols based on pairwise masking~\cite{bonawitz2017practical}, nodes often perform a Diffie-Hellman (DH) key agreement~\cite{diffie2022new} or similar handshakes. In a DFL graph, this requires nodes to obtain network-level identifiers (e.g., IP addresses, node IDs, or public keys) of their neighbors to establish shared secrets. While DH does not require explicit identity disclosure, practical implementations involving identity-bound certificates typically expose the identities of communication partners. 
\textbf{Case 2} corresponds to protocols utilizing anonymizing relays or secure hardware like Trusted Execution Environments as intermediate aggregators. While these primitives can obscure the source of individual updates, the recipient often observes metadata regarding the incoming encrypted contributions. For instance, threshold-based Shamir secret-sharing schemes~\cite{shamir1979share} require the system to verify that a sufficient number of shares have been received, thereby inherently leaking the in-degree $d_i^{in}$.
Finally, \textbf{Case 3} represents an idealized setting where both identities and structural metadata are strictly hidden, e.g., through fully anonymous communication or blind homomorphic aggregation. In this case, the adversary only observes the final aggregated state $\v{x}_i^{(t+1)}$ without reliable information about the provenance or the size of the participating set.

\section{Proof of Proposition~\ref{prop:isolation}}\label{app.proof_isolation}
\begin{proof}
By stacking the aggregation process across $t_{\max}$ iterations, the global model evolution can be initially characterized as a block-diagonal linear system:
\begin{align}\label{eq.wd_vec}
\begin{pmatrix}
\v X^{(1)}\\
\v X^{(2)}\\
\vdots\\
\v X^{(t_{\max})}\\
\end{pmatrix} =\begin{pmatrix}
\v W^{(1)} & \v 0 & \dots & \v 0 \\
\v 0 & \v W^{(2)} & \dots & \v 0 \\
\vdots &  \vdots & \ddots & \vdots \\
\v 0 & \v 0 & \dots & \v W^{(t_{\max})} \\
\end{pmatrix} \begin{pmatrix}
\v X^{(\frac{1}{2})}\\
\v X^{(1\frac{1}{2})}\\
\vdots\\
\v X^{(t_{\max} -\frac{1}{2})}\\
\end{pmatrix}.
\end{align}

By extracting the rows corresponding to $\mathcal{V}_c$ from \eqref{eq.wd_vec}, we isolate the system from the adversaries' perspective:
\begin{align*}
\begin{pmatrix}
\v X_c^{(1)}\\
\v X_c^{(2)}\\
\vdots\\
\v X_c^{(t_{\max})}\\
\end{pmatrix} =\begin{pmatrix}
(\v W_{cc}^{(1)}~~\v W_{ch}^{(1)}) &\v 0&  \dots & \v 0 \\
\v 0& (\v W_{cc}^{(2)}~~\v W_{ch}^{(2)}) &  \dots & \v 0 \\
\vdots  &\vdots  & \ddots & \vdots \\
\v 0 &\v 0 & \dots & (\v W_{cc}^{(t_{\max})}~~\v W_{ch}^{(t_{\max})}) \\
\end{pmatrix} \begin{pmatrix}
\v X^{(\frac{1}{2})}\\
\v X^{(1+\frac{1}{2})}\\
\vdots\\
\v X^{(t_{\max}-\frac{1}{2})}\\
\end{pmatrix}.
\end{align*}

Given the prior knowledge of mutual weights $\v W_{cc}^{(t)}$ among corrupted nodes, we move all known terms to the left-hand side, yielding the compact matrix representation $\v Y = \v A \v B$:
\begin{align*}
\underbrace{
\begin{pmatrix}
\v X_c^{(1)}-\v W_{cc}^{(1)}\v X_c^{(\frac{1}{2})}\\
\v X_c^{(2)}-\v W_{cc}^{(2)}\v X_c^{(1+\frac{1}{2})}\\
\vdots\\
\v X_c^{(t_{\max})}-\v W_{cc}^{(t_{\max})}\v X_c^{(t_{\max}-\frac{1}{2})}\\
\end{pmatrix}
}_{\v Y}
=
\underbrace{
\begin{pmatrix}
\v W_{ch}^{(1)} &\v 0 & \dots & \v 0 \\
\v 0& \v W_{ch}^{(2)}  & \dots & \v 0 \\
\vdots  &\vdots  & \ddots & \vdots \\
\v 0  &\v 0  & \dots & \v W_{ch}^{(t_{\max})} \\
\end{pmatrix}
}_{\v A}
\underbrace{
\begin{pmatrix}
\v X_h^{(\frac{1}{2})}\\
\v X_h^{(1+\frac{1}{2})}\\
\vdots\\
\v X_h^{(t_{\max}-\frac{1}{2})}\\
\end{pmatrix},
}_{\v B}
\end{align*}
where $\v{Y}$ acts as the observation matrix composed entirely of known adversarial states and weights, $\v{A}$ is the block-diagonal mixing matrix, and $\v{B}$ represents the latent honest states. Because $\v{A}$ possesses a strictly block-diagonal structure, the global system $\v Y = \v A \v B$ decoupled completely into a series of independent round-wise sub-problems \eqref{eq.sub_problem}. As discussed in the proposition, while the hidden internal weights $\v{W}_{hc}^{(t+1)}$ and $\v{W}_{hh}^{(t+1)}$ prevent the explicit derivation of individual updates $\Delta \v{x}_i^{(t)}$ for $t \geq 1$, recovering the intermediate states $\v{X}_h^{(t+\frac{1}{2})}$ still constitutes a severe breach, as they encapsulate the cumulative private data of the honest nodes. 
\end{proof}

\section{Proof of Theorem~\ref{thm:reduction}}\label{app.reduction}
\begin{proof}

\red{We establish the reduction by constructing an integer-domain representation of the DFL observation system. 

Starting from the isolated sub-problem derived in Proposition~\ref{prop:isolation}, we have $\v{Y} = \v{W}_{ch} \v{X}_h$. 
First, we address the rational nature of the mixing weights. By multiplying both sides by the common denominator $\beta$, we obtain:
\begin{align*} 
    \beta \v{Y} = (\beta \v{W}_{ch}) \v{X}_h = \tilde{\v{W}} \v{X}_h,
\end{align*}
where $\tilde{\v{W}} \in \mathbb{Z}^{M \times N}$ is now an integer matrix. 


Second, under the exact integer-domain representation, define
\begin{equation*}
    \tilde{\v{X}}_h
    =
    10^\gamma \v{X}_h
    \in
    \mathbb{Z}^{N\times u}.
\end{equation*}
Multiplying both sides of
$\beta\v{Y}=\tilde{\v{W}}\v{X}_h$
by $10^\gamma$ yields
\begin{equation}
\label{eq.proof_quantize}
    \tilde{\v{Y}}
    =
    10^\gamma\beta\v{Y}
    =
    \tilde{\v{W}}\tilde{\v{X}}_h
    \in
    \mathbb{Z}^{M\times u}.
\end{equation}

Third, we reduce this integer relation entry-wise modulo a prime $Q$, yielding
\begin{align} 
\label{eq.proof_modulo}
    \tilde{\v{Y}} \equiv \tilde{\v{W}} \tilde{\v{X}}_h \pmod Q.
\end{align}
%
Here, $Q$ is a large prime which satisfy the requirements of the corresponding lattice attack.
\eqref{eq.proof_modulo} has exactly the matrix structure of the mHSSP/mHLCP, which the columns of $\tilde{\v{W}}$ constitute the hidden coefficient vectors, while the rows of $\tilde{\v{X}}_h$ correspond to the multidimensional secrets.

Finally, \eqref{eq.proof_modulo} aligns identically with the formal definition of multi-dimensional lattice challenges:
\begin{itemize}
    \item If the DFL protocol employs uniform mixing weights (e.g., Secure Summation where $W_{ij} \in \{0, 1\}$), the scaled coefficients lie within $\mathcal{C} \subseteq \{0, 1\}$. \eqref{eq.proof_modulo} is thus the exact formulation of an \textbf{mHSSP} instance, where the adversary seeks an unknown binary selection matrix.
    \item If the protocol employs heterogeneous mixing weights (e.g., weighted averaging), the scaled integer weights are constrained within $\mathcal{C} \subseteq \{0, \dots, c\}$. In this case, \eqref{eq.proof_modulo} corresponds perfectly to an \textbf{mHLCP} instance, generalizing the subset-sum structure to arbitrary bounded integer coefficients.
\end{itemize}
Since recovering $\v{X}_h$ requires first identifying the hidden integer matrix $\tilde{\v{W}}$, the coefficient-recovery stage of the privacy reconstruction problem can be reduced to solving the corresponding mHSSP/mHLCP instance. Once $\tilde{\v{W}}$ is recovered, the original weight matrix is obtained as $\v{W}_{ch}=\beta^{-1}\tilde{\v{W}}$, after which the honest states can be recovered from the original linear system $\v{Y}=\v{W}_{ch}\v{X}_h$ whenever $\v{W}_{ch}$ admits unique inversion. This completes the proof.

}
\end{proof}

\section{Other HSSP Attacks}\label{app.cg_attack}
Coron and Gini~\cite{coron2020polynomial,coron2021provably} introduced two alternative methods that reduce the complexity of the second phase from exponential to polynomial time, effectively shifting the overall computational bottleneck to the $O(N^9)$ and $O(N^7(N + \log c)^2)$ complexity of the first phase. 

The first method is a \textit{multivariate approach}, which reformulates the recovery of binary vectors into a system of multivariate quadratic equations. The second method is a \textit{statistical approach}, which leverages the geometric structure of parallelepipeds and the Nguyen-Regev algorithm~\cite{nguyen2006learning} to infer the hidden binary matrix through distributional analysis.

Since the specific choice of the attack algorithm does not alter the fundamental principles of our proposed privacy analysis, we refer the reader to \cite{coron2020polynomial,coron2021provably} for further technical details.

\section{Case Study}\label{ssec.similar_works}
To the best of our knowledge, relatively few works study privacy in DFL from a reconstruction perspective. Interestingly, several recent studies can be viewed as special cases of our framework under more restrictive assumptions. This connection situates our work within the existing literature and highlights the generality of our approach.

\subsection{Case Study: The Reconstruction Attack by Mrini et al. \cite{mrini2024privacy}}

The reconstruction attack proposed in \cite{mrini2024privacy} represents a specific instance of our framework where no SA is employed, and $\v{W}$ is assumed to be time-invariant and fully known. The core strategy of this attack is to exploit temporal correlations to jointly utilize observations across multiple iterations. 

By recursively expanding the aggregation and update rules in \eqref{eq.local_update} and \eqref{eq.aggr}, the model state of corrupted nodes at iteration $t+1$ can be expressed as
\begin{align*}
    \v X_c^{(t+1)}
    &=\v W_{cc}\v X_c^{(t+\frac{1}{2})}+\v W_{ch}\v X_h^{(t+\frac{1}{2})}\\
    &=\v W_{cc}\v X_c^{(t+\frac{1}{2})}+\v W_{ch}(\v X_h^{(t)}+\Delta \v X_h^{(t)})\\
    &=\v W_{cc}\v X_c^{(t+\frac{1}{2})}+\v W_{ch}(\v W_{hc}\v X_c^{(t-\frac{1}{2})}+\v W_{hh}\v X_h^{(t-\frac{1}{2})})+\v W_{ch}\Delta \v X_h^{(t)}\\
    &\cdots\\
    &=\v W_{cc} \v X_c^{(t+\frac{1}{2})}+\v W_{ch}\sum_{i=0}^{t-1}\v W_{hh}^i\v W_{hc}\v X_c^{(t-1-i+\frac{1}{2})}\\
    &+\v W_{ch}\sum_{i=0}^{t}(\v W_{hh}^i\Delta\v X_h^{(t-i)})+\v W_{ch}\v W_{hh}^t \v X_h^{(0)}.
\end{align*}

To simplify the solution space, the authors introduce a quasi-static update assumption, where the local update (gradient) at each node $i$ is treated as a constant vector corrupted by a zero-mean random noise
\begin{align*}
    \Delta \v{x}_i^{(t)} = \Delta \v{x}_i + \v{N}^{(t)}.
\end{align*}

Under this assumption, the system can be approximated by the linear form $\v{Y} \approx \v{A}\v{B}$, where $\v{Y} = [\v{Y}^{(1)\top}, \dots, \v{Y}^{(t_{\max})\top}]^{\top}$, which $\v Y^{(t)}=\v X_c^{(t+1)}-\v W_{cc} \v X_c^{(t+\frac{1}{2})}-\v W_{ch}\sum_{i=0}^{t-1}\v W_{hh}^i\v W_{hc}\v X_c^{(t-1-i+\frac{1}{2})}-\v W_{ch}\v W_{hh}^t \v X_h^{(0)}$, $\v{A} = [\v{W}_{ch}^{\top}, \dots, (\v{W}_{ch}\sum_{i=0}^{t_{\max}-1}\v{W}_{hh}^i)^{\top}]^{\top}$ is the cumulative weight matrix, and $\v{B} = \Delta\v{X}_h$ is the target constant update. The adversaries then solve for the updates using the Moore-Penrose pseudo-inverse
\begin{align*}
    \Delta\v{X}_h \approx \v{A}^{\dagger}\v{Y}.
\end{align*}

While this approach is empirically effective in certain settings, it possesses a limitation: the assumption that local updates remain approximately constant across iterations is often violated in realistic scenarios. This mismatch inevitably introduces reconstruction noise and limits the attack's precision.

\subsection{Case Study: Vulnerability of Secure Summation in Asynchronous DFL \cite{dekker2025topology}}
Dekker et al. \cite{dekker2025topology} demonstrate that the secure summation in DFL is not inherently robust against reconstruction attacks, specifically in asynchronous settings. The vulnerability arises from the dynamic participation of nodes where, in each iteration, only a subset of active nodes participates in the summation, while inactive nodes effectively have zero updates ($\Delta \v{x}_i^{(t)} = 0$). This creates opportunities for adversaries to observe stationary $\v{x}_i^{(t)}$ from the perspectives of different corrupted nodes across multiple rounds, thereby providing fixed reference points to isolate honest nodes' private states.

Furthermore, because the protocol performs a simple summation rather than a weighted average, the elements of the mixing matrix $\v{W}^{(t)}$ are binary, i.e., $W_{ij}^{(t)} \in \{0, 1\}$.  Under the assumption of full neighbor awareness (Case 1), by colluding observations from multiple corrupted nodes and exploiting $\v{x}_i^{(t)}$ that remain stationary over multiple rounds, the adversaries can formulate a system of linear equations. This allows for the direct recovery of part local updates by solving the resulting linear system.

\section{Experimental Setup}\label{ssec.setup}

\paragraph{Graph Generation.}
 \red{The main experiments use fixed-edge Erdős-Rényi graphs generated with NetworkX~\cite{networkx}. We construct connected undirected graphs and strongly connected directed graphs, each parameterized by the number of nodes $n$ and edges $e$.} A fraction $\eta$ of the total nodes are randomly designated as corrupted nodes. For undirected topologies, the weight matrix $\v{W}$ is designed to be doubly stochastic, while for directed topologies, it is column-stochastic. To map the weights into the integer domain, we scale the matrix by a factor $\beta$, yielding $\tilde{\v{W}} = \beta \v{W}$.

\paragraph{Synthetic Data}
For synthetic data evaluations, each node $i$ is assigned a model state $\v{x}_i^{(\frac{1}{2})} \in \mathbb{Z}^{u}$ with dimensionality $u=100$, where the integer elements are sampled uniformly from the range $[0,100]$. To simulate the finite field operations required for lattice attacks, we quantize all floating-point model parameters into integers and perform operations modulo $Q$, where $Q$ is a 50-bit pseudo-prime.

\paragraph{Real Dataset and Model Architectures.}
We evaluate our framework across image, tabular, and text modalities.
\begin{itemize}
    \item \textbf{Image Data:} We use the CIFAR-10 dataset~\cite{cifar}. The classification model is a Convolutional Neural Network consisting of three convolutional layers ($32$, $64$, and $128$ channels with $3\times3$ kernels) followed by two fully connected layers with $512$ and $10$ hidden units, respectively.
    \item \textbf{Tabular Data:} We use the Purchase dataset~\cite{shokri2017membership}, which consists of $600$ binary features per sample. The model is a multi-layer perceptron with two fully connected layers (256 hidden units). For visualization, the 600-dimensional tabular features are reshaped into a $24\times 25$ grid.
    \item \textbf{Text Data:} We use the Sentiment140 dataset~\cite{go2009twitter} for sentiment analysis, employing a logistic regression model for sentiment classification.
\end{itemize}


\red{
Unless otherwise stated, each node performs SGD with a batch size of $1$ in the main experiments. This setting allows us to isolate the privacy leakage introduced by SA, which constitutes the primary stage of our attack: by neutralizing the aggregation barrier, the proposed method recovers the individual gradients of honest nodes and makes them available to existing downstream reconstruction techniques. For single-layer architectures, such as logistic regression, the corresponding input can be reconstructed exactly by taking the ratio between the weight and bias gradients associated with a given label. For deep neural networks, we apply the GIA framework of~\cite{geiping2020inverting} to reconstruct raw images and tabular features from the recovered gradients. Although downstream gradient inversion is not the primary focus of this work, we further evaluate the transferability of our attack to more practical reconstruction settings. In Appendix~\ref{app.gia_batch}, we consider larger batchsizes and invoke several stronger GIA implementations through the \textit{Breaching}\cite{geiping2022breaching} framework. These experiments examine whether the isolated client updates recovered by our method remain exploitable under different downstream inversion algorithms. For text-based tasks, the recovered embeddings can additionally be mapped back to natural-language sequences using \textit{vec2text}~\cite{morris2023text}, providing another downstream instantiation of data reconstruction. Overall, our methodology circumvents the protection provided by SA against direct access to individual client updates, thereby enabling existing gradient- and embedding-based reconstruction techniques to be applied in decentralized learning settings.

}

\section{The Scale of the Core Sub-graphs}\label{exp.large_graph}

As illustrated in Figure~\ref{fig:large_graph}, even when the original graph is large (e.g., consisting of 100 or 200 nodes), its inherent sparsity ensures that the topological pruning process effectively decomposes the global system into multiple independent and manageable core subgraphs. The sizes of these resulting subgraphs are predominantly concentrated within the range of 0 to 10. At this scale, HSSP-based attacks remain exceptionally efficient. Consequently, we focus our experimental evaluations on core subgraphs with a representative size of $n=10$.

\begin{figure}[ht]
    \centering
    \begin{minipage}{0.24\textwidth}
        \centering
        \includegraphics[width=\linewidth]{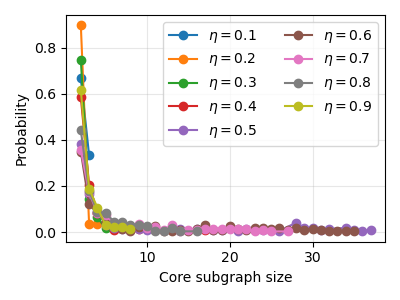}
        \\[-1ex]{\small (a) $n=100, e=150$}
    \end{minipage}
    \begin{minipage}{0.24\textwidth}
        \centering
        \includegraphics[width=\linewidth]{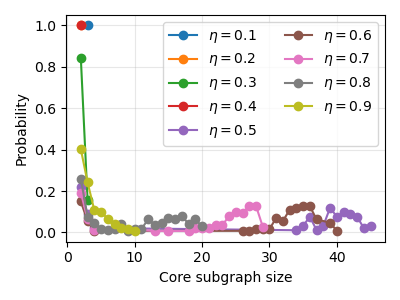}
        \\[-1ex]{\small (b) $n=100, e=200$}
    \end{minipage}
    \begin{minipage}{0.24\textwidth}
        \centering
        \includegraphics[width=\linewidth]{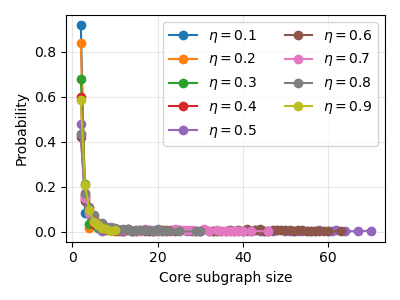}
        \\[-1ex]{\small (c) $n=200, e=300$}
    \end{minipage}
    \begin{minipage}{0.24\textwidth}
        \centering
        \includegraphics[width=\linewidth]{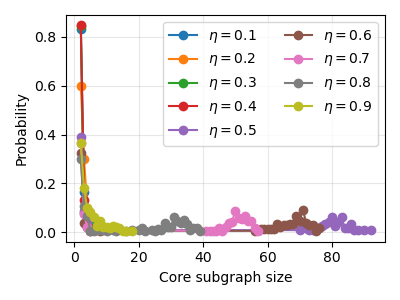}
        \\[-1ex]{\small (d) $n=200, e=400$}
    \end{minipage}
    \caption{Proportional distribution of varying core subgraph sizes following simplification in large-scale undirected networks.}
    \label{fig:large_graph}
\end{figure}

\red{
We additionally considered different connected graph families to examine the dependence on network topology: ER graphs, regular ring lattices, small-world graphs with rewiring probability \(p=0.1\), and scale-free graphs. As shown in Figure~\ref{fig:core_size_distribution}, the core-subgraph sizes remain predominantly concentrated between $0$ and $10$ across all considered graph families. This indicates that the tendency of the simplification procedure to reduce a large network to a relatively small core is not specific to the ER topology used in our main experiments. Nevertheless, the detailed distributions vary across graph families. The ring and small-world curves are smoother and place more probability mass on smaller core sizes than the ER and scale-free curves. This suggests that their more regular local connectivity allows the simplification procedure to eliminate nodes more consistently.

\begin{figure}
    \centering
    \includegraphics[width=0.95\linewidth]{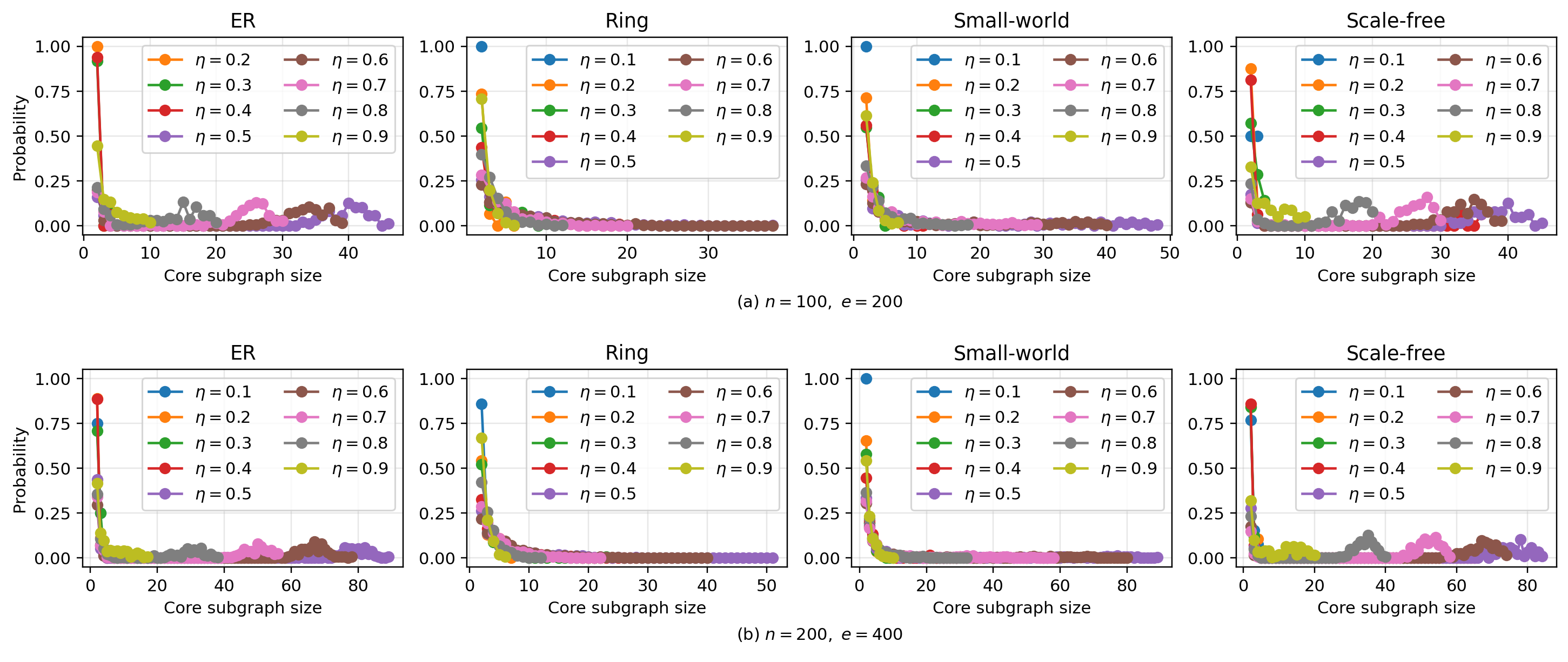}
    \caption{\red{Proportional distribution of varying core subgraph sizes following simplification in different graphs.}}
    \label{fig:core_size_distribution}
\end{figure}
}

\section{Attack Performance on Directed Graph}\label{exp.directed}

Table~\ref{tab.directed_hlcp_7} presents the experimental results for directed topologies. An observation is that the absence of the row-stochastic property makes directed graphs significantly more challenging for the adversary. Without this global summation information, it becomes difficult to distinguish between valid and invalid vector combinations, leading to a much larger set of feasible solutions compared to the undirected case.
However, this ambiguity can be mitigated by utilizing the auxiliary scalars $a_i^{(0)}$ as a secondary filter. Once the candidate weight matrices are reconstructed and the corresponding model updates are derived, the adversary can validate these results against the known evolution of such scalars. This step effectively prunes a vast majority of combinations, thereby substantially increasing the probability of identifying the unique GT solution.

\begin{table*}[ht]
\centering
\begin{tabular}{c c c c c c c c}
\hline
  & All Weights Found & Found Vec. & $|\mathcal{F}|$ (Case 1)  & $|\mathcal{F}|$ (Case 1 after validated by $a_i^{(0)}$) & Step1 Time (s) & Step2 Time (s)\\
\hline
1 & TRUE & 1728 & 161051 & 77  & 0.168152 & 0.299058 \\
2 & TRUE & 650 & 1331 & 1 & 0.001548 & 0.054637 \\
3 & TRUE & 936 &  13924 & 8  & 0.001124 & 0.094580 \\
4 & TRUE & 91  & 125 & 1  & 0.001270 & 0.004141 \\
5 & TRUE & 75  & 256 & 4  & 0.000873 & 0.002860 \\
\hline
\end{tabular}
\caption{Performance of the NS attack in resolving the mHLCP ($n=10, e=30, \eta=0.7$) in directed graph.}
\label{tab.directed_hlcp_7}
\end{table*}

\red{
\section{Transferability to different GIAs}\label{app.gia_batch}

We further examine whether the individual client updates recovered by our attack remain exploitable under different GIA methods and larger batchsizes. To this end, we use the unified API provided by the \textit{Breaching} framework~\cite{geiping2022breaching} to evaluate several representative GIA implementations, including \cite{zhu2019deep,geiping2020inverting,yin2021see}. This evaluation is intended to assess the transferability of our structural recovery attack to different downstream reconstruction methods, rather than to provide a comprehensive comparison among GIAs. We consider batchsizes $|\mathcal{B}_i|\in\{1,2,4,8\}$. For each setting, we compare downstream reconstruction from two inputs: the unprotected ground-truth client update $\v g$ and the corresponding update $\hat{\v g}$ recovered by circumventing SA. We restrict the evaluation to trials in which the HSSP-based attack successfully recovers the correct client update, thereby isolating the effect of replacing $\v g$ with its reconstructed counterpart $\hat{\v g}$. In addition to PSNR and SSIM, we report the learned perceptual image patch similarity (LPIPS) metric~\cite{zhang2018unreasonable} with the AlexNet backbone, for which a lower value indicates greater perceptual similarity between the reconstructed and ground-truth images.

\begin{table*}[t]
\centering
\resizebox{\textwidth}{!}{%
\begin{tabular}{rcccccc}
\toprule
& \multicolumn{2}{c}{Zhu et al.~\cite{zhu2019deep}} &
\multicolumn{2}{c}{Geiping et al.~\cite{geiping2020inverting}} &
\multicolumn{2}{c}{Yin et al.~\cite{yin2021see}} \\
\cmidrule(lr){2-3}\cmidrule(lr){4-5}\cmidrule(lr){6-7}
$|\mathcal{B}_i|$ & True & Recovered & True & Recovered & True & Recovered \\
\midrule
1 & 0.3826 / 0.0852 / 16.76 & 0.3746 / 0.0910 / 16.62 & 0.4712 / 0.1361 / 18.09 & 0.4680 / 0.1326 / 18.11 & 0.3699 / 0.0983 / 16.64 & 0.3750 / 0.1013 / 16.63 \\
2 & 0.3983 / 0.0721 / 17.67 & 0.4149 / 0.0694 / 17.89 & 0.4961 / 0.1165 / 19.53 & 0.5000 / 0.1170 / 19.57 & 0.4721 / 0.0589 / 18.91 & 0.4998 / 0.0545 / 19.27 \\
4 & 0.4165 / 0.0636 / 17.50 & 0.3832 / 0.0745 / 17.07 & 0.4640 / 0.1284 / 18.25 & 0.4651 / 0.1272 / 18.27 & 0.5146 / 0.0625 / 18.87 & 0.5058 / 0.0612 / 18.74 \\
8 & 0.3206 / 0.1001 / 15.99 & 0.3064 / 0.0978 / 15.64 & 0.3639 / 0.1634 / 16.35 & 0.3740 / 0.1633 / 16.52 & 0.4366 / 0.0946 / 17.49 & 0.4350 / 0.0928 / 17.69 \\
\bottomrule
\end{tabular}}
\caption{\red{Paired downstream GIA performance using the ground-truth gradient $\v g$ and the corresponding gradient $\hat{\v g}$ recovered by the HSSP-based attack. Each entry is mean SSIM / LPIPS / PSNR (dB) over the training images from honest nodes.}}
\label{tab:gia_transferability}
\end{table*}

Table~\ref{tab:gia_transferability} shows that the reconstructions obtained from the recovered updates closely match those obtained directly from the ground-truth updates across the evaluated GIA methods and batch sizes. In particular, the HSSP-based recovery stage introduces little additional degradation into the downstream reconstruction process. These results indicate that downstream reconstruction performance is primarily determined by the capability of the selected GIA and the batchsize, rather than by inaccuracies introduced during update recovery. Stronger GIAs may improve the absolute reconstruction quality, but do not alter the central conclusion of this work: once the protection provided by SA is circumvented and individual client updates are exposed, existing gradient-inversion techniques can be directly applied in decentralized learning settings.
}

\section{Extended Numerical Results}\label{app.table}

\begin{table*}[t]
\centering
\begin{tabular}{c c c c c c c c}
\hline
 &  All Weights Found & Found Vec. & $|\mathcal{F}|$ (Case 1) & $|\mathcal{F}|$ (Case 2) & $|\mathcal{F}|$ (Case 3) & Step1 Time (s)\footnotemark & Step2 Time (s)\\
\hline
1  & TRUE & 1152    & 30 & 571 & 1455876 & 0.103096 & 0.172203 \\
2  & TRUE & 544    & 20 & 40 & 956736  & 0.000942 & 0.044706 \\
3  & TRUE & 776     & 100 & 750 & 11588400 & 0.001284 & 0.079257 \\
4  & TRUE & 104     & 3  & 28 & 22392   & 0.048423 & 0.009377 \\
5  & TRUE & 360     & 3  & 12 & 977076  & 0.001345 & 0.023219 \\
6  & TRUE & 608     & 5  & 10 & 260304  & 0.000922 & 0.055605 \\
7  & TRUE & 640     & 15 & 60 & 2774016 & 0.000874 & 0.058244 \\
8  & TRUE & 506     & 1 & 4 & 153876  & 0.000877 & 0.039568 \\
9  & FALSE & 0      & -   & -       & -        & - \\
10 & TRUE & 352     & 28 &  168 & 475440  & 0.000893 & 0.022448 \\
\hline
\end{tabular}
\caption{Performance of the NS attack in resolving the mHLCP ($n=10, e=20, \eta=0.6$) in undirected graphs. The metrics include success status, number of recovered candidate vectors, remaining vector combinations after topological filtering (Cases 1 - 3), and execution times for the two attack phases in NS attack.}
\label{tab.hlcp_6}
\end{table*}
\footnotetext{All timings are measured on a personal computer.}

\begin{figure}[ht]
    \centering
    \begin{minipage}{0.23\textwidth}
        \centering
        \includegraphics[width=\linewidth]{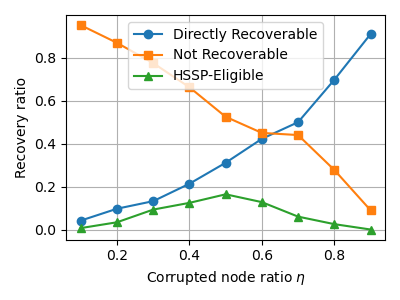}
        \\[-1ex]{\small (a) $n=20, e=40$}
    \end{minipage}
    \begin{minipage}{0.23\textwidth}
        \centering
        \includegraphics[width=\linewidth]{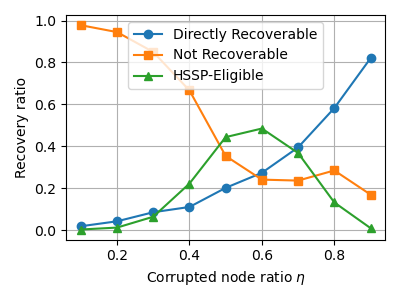}
        \\[-1ex]{\small (b) $n=20, e=60$}
    \end{minipage}
    \begin{minipage}{0.23\textwidth}
        \centering
        \includegraphics[width=\linewidth]{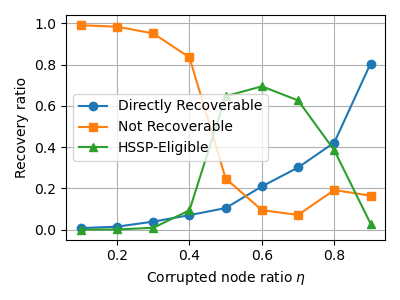}
        \\[-1ex]{\small (c) $n=20, e=80$}
    \end{minipage}
    \begin{minipage}{0.23\textwidth}
        \centering
        \includegraphics[width=\linewidth]{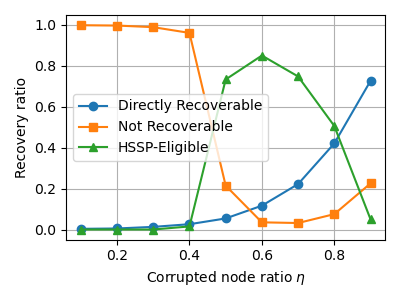}
        \\[-1ex]{\small (d) $n=20, e=100$}
    \end{minipage}
    \caption{Relationship between the ratio of corrupted nodes and the ratio of honest nodes categorized by recoverability in directed graphs.
    }
    \label{fig:directed_recovery}
\end{figure}


\begin{table*}[!ht]
\centering
\begin{tabular}{c c c c c c c c}
\hline
 &  All Weights Found & Found Vec. & $|\mathcal{F}|$ (Case 1) & $|\mathcal{F}|$ (Case 2) & $|\mathcal{F}|$ (Case 3) & Step1 Time (s) & Step2 Time (s)\\
\hline
1  & TRUE  & 47  & 2 & 4  & 204  & 0.133022 & 0.002395 \\
2  & TRUE  & 15  & 1 & 1  & 33   & 0.009878 & 0.000770 \\
3  & TRUE  & 14  & 1 & 3  & 27   & 0.000837 & 0.000625 \\
4  & TRUE  & 15  & 1 & 3  & 36   & 0.000772 & 0.000670 \\
5  & TRUE  & 61  & 1 & 2  & 456  & 0.003051 & 0.002427 \\
6  & TRUE  & 88  & 1 & 2  & 270  & 0.000936 & 0.003373 \\
7  & TRUE  & 196 & 9 & 26 & 3480 & 0.001093 & 0.009402 \\
8  & TRUE  & 35  & 1 & 2  & 129  & 0.000945 & 0.001764 \\
9  & FALSE & 0   & - & -  & -    & -        & -        \\
10 & TRUE  & 44  & 1 & 3  & 174  & 0.007832 & 0.001907 \\
\hline
\end{tabular}
\caption{Performance of the NS attack in resolving the mHLCP ($n=10, e=20, \eta=0.7$) in undirected graphs.}
\label{tab.hlcp_7}
\end{table*}

\begin{table}[!ht]
\centering
\begin{tabular}{l c c }
\hline
Config & Recall (\%) & Avg. Found Vec. \\
\hline
$n=10, e=20, \eta=0.6$& 90.0\% & 11.6  \\
$n=10, e=20, \eta=0.7$ & 96.0\% & 5.3 \\
$n=10, e=20, \eta=0.8$  & 100.0\% & 3.2  \\
$n=20, e=40, \eta=0.6$ & 79.0\% & 111.8 \\
$n=20, e=40, \eta=0.7$ & 98.0\% & 23.6  \\
$n=20, e=40, \eta=0.8$ & 100.0\% & 8.1 \\
$n=30, e=60, \eta=0.6$ & 79.0\% & 1295.9 \\
$n=30, e=60, \eta=0.7$ & 93.0\% & 117.0  \\
$n=30, e=60, \eta=0.8$ & 100.0\% & 20.2  \\
$n=40, e=80, \eta=0.6$& 44.0\% & 9183.0  \\
$n=40, e=80, \eta=0.7$& 75.0\% & 489.1  \\
$n=40, e=80, \eta=0.8$& 97.0\% & 56.0  \\
\hline
\end{tabular}
\caption{mHSSP attack performance with different configurations in undirected graphs.}
\label{tab.large_hssp_summary}
\end{table}

\begin{table*}[!ht]
\centering
\begin{tabular}{c c c c c c c c}
\hline
 &  All Weights Found & Found Vec. & $|\mathcal{F}|$ (Case 1) & $|\mathcal{F}|$ (Case 2) & $|\mathcal{F}|$ (Case 3) & Step1 Time (s) & Step2 Time (s)\\
\hline
1  & TRUE & 16 & - & 2  & 120 & 0.126755 & 0.001178 \\
2  & TRUE & 6  & - & 2  & 6   & 0.000871 & 0.000360 \\
3  & TRUE & 12 & -& 3  & 60  & 0.000684 & 0.000648 \\
4  & TRUE & 12 & - & 30 & 72  & 0.000608 & 0.000525 \\
5  & TRUE & 9  & - & 16 & 30  & 0.000549 & 0.000447 \\
6  & TRUE & 12 & - & 2  & 24  & 0.000623 & 0.000504 \\
7  & TRUE & 12 & - & 2  & 12  & 0.000537 & 0.000503 \\
8  & TRUE & 10 & - & 8  & 42  & 0.000642 & 0.000498 \\
9  & TRUE & 12 & - & 6  & 96  & 0.001185 & 0.000557 \\
10 & TRUE & 5  & - & 2  & 6   & 0.000644 & 0.000302 \\
\hline
\end{tabular}
\caption{Performance of the NS attack in resolving the mHSSP ($n=10, e=20, \eta=0.6$) in undirected graphs.}
\label{tab.hssp_6}
\end{table*}

\begin{table*}[!ht]
\centering
\begin{tabular}{c c c c c c c c}
\hline
 &  All Weights Found & Found Vec. & $|\mathcal{F}|$ (Case 1) & $|\mathcal{F}|$ (Case 2) & $|\mathcal{F}|$ (Case 3) & Step1 Time (s) & Step2 Time (s)\\
\hline
1  & TRUE  & 6 & - & 2 & 6 & 0.111992 & 0.000508 \\
2  & TRUE  & 4 & - & 2 & 3 & 0.005444 & 0.000416 \\
3  & TRUE  & 4 & - & 2 & 3 & 0.000890 & 0.000266 \\
4  & FALSE & 0 & - & - & - & -        & -        \\
5  & TRUE  & 5 & - & 2 & 3 & 0.000849 & 0.000312 \\
6  & TRUE  & 6 & - & 2 & 3 & 0.000995 & 0.000272 \\
7  & TRUE  & 4 & - & 1 & 3 & 0.000704 & 0.000307 \\
8  & TRUE  & 6 & - & 2 & 3 & 0.000780 & 0.000443 \\
9  & TRUE  & 4 & - & 1 & 3 & 0.000918 & 0.000276 \\
10 & TRUE  & 4 & - & 1 & 3 & 0.000622 & 0.000219 \\
\hline
\end{tabular}
\caption{Performance of the NS attack in resolving the mHSSP ($n=10, e=20, \eta=0.7$) in undirected graph.}
\label{tab.hssp_7}
\end{table*}

\begin{figure}[h]
  \centering
  \includegraphics[width=0.4\linewidth]{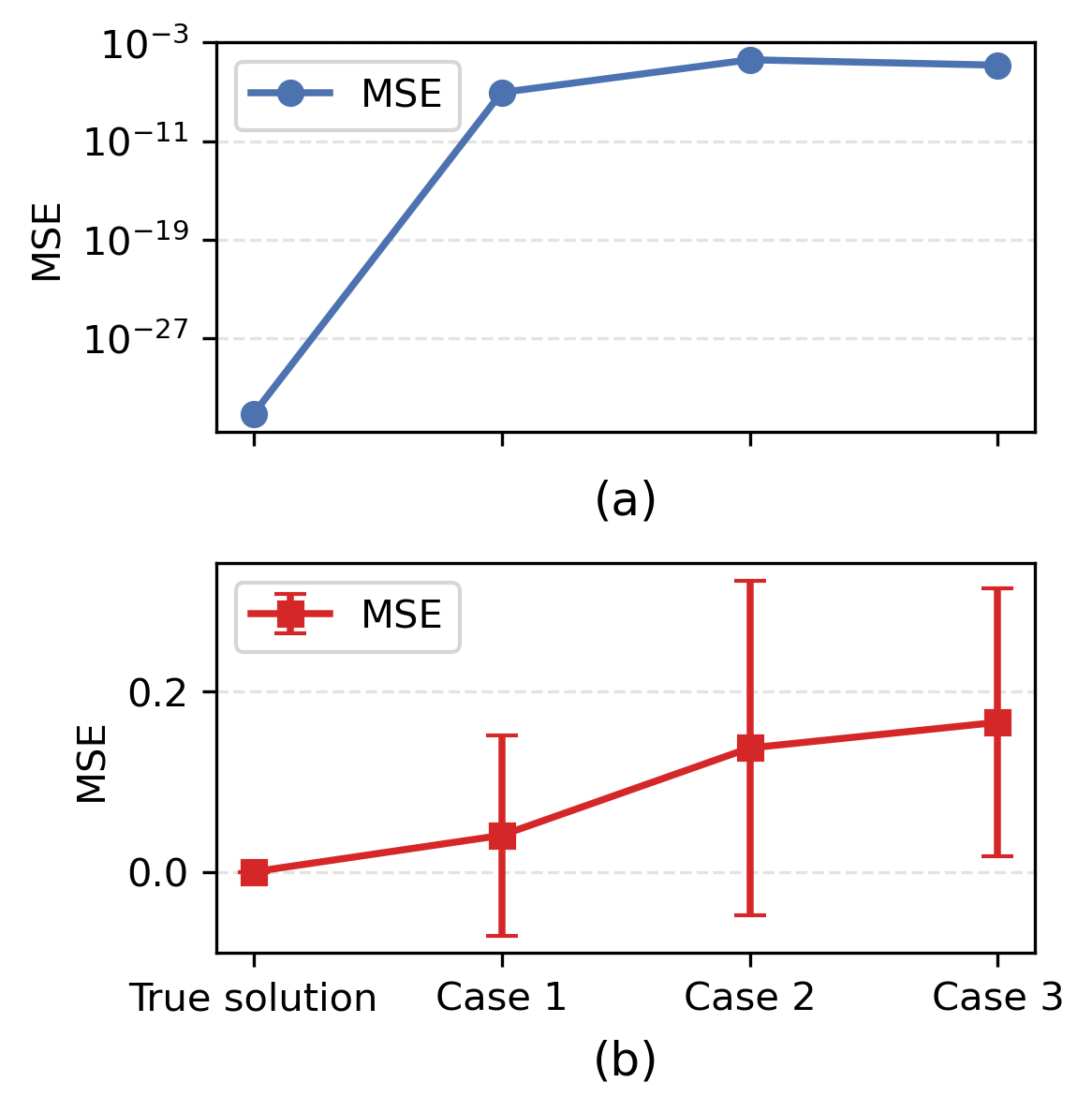}
  \caption{Statistical reconstruction performance on the Purchase dataset. (a) MSE between the HSSP-reconstructed gradients and the GT gradients. (b) Distribution of MSE scores for training data (binary) reconstructed from the sampled candidate gradients.}
  \Description{}
  \label{fig:stat_purchase}
\end{figure}

\begin{figure}[h]
  \centering
  \includegraphics[width=0.4\linewidth]{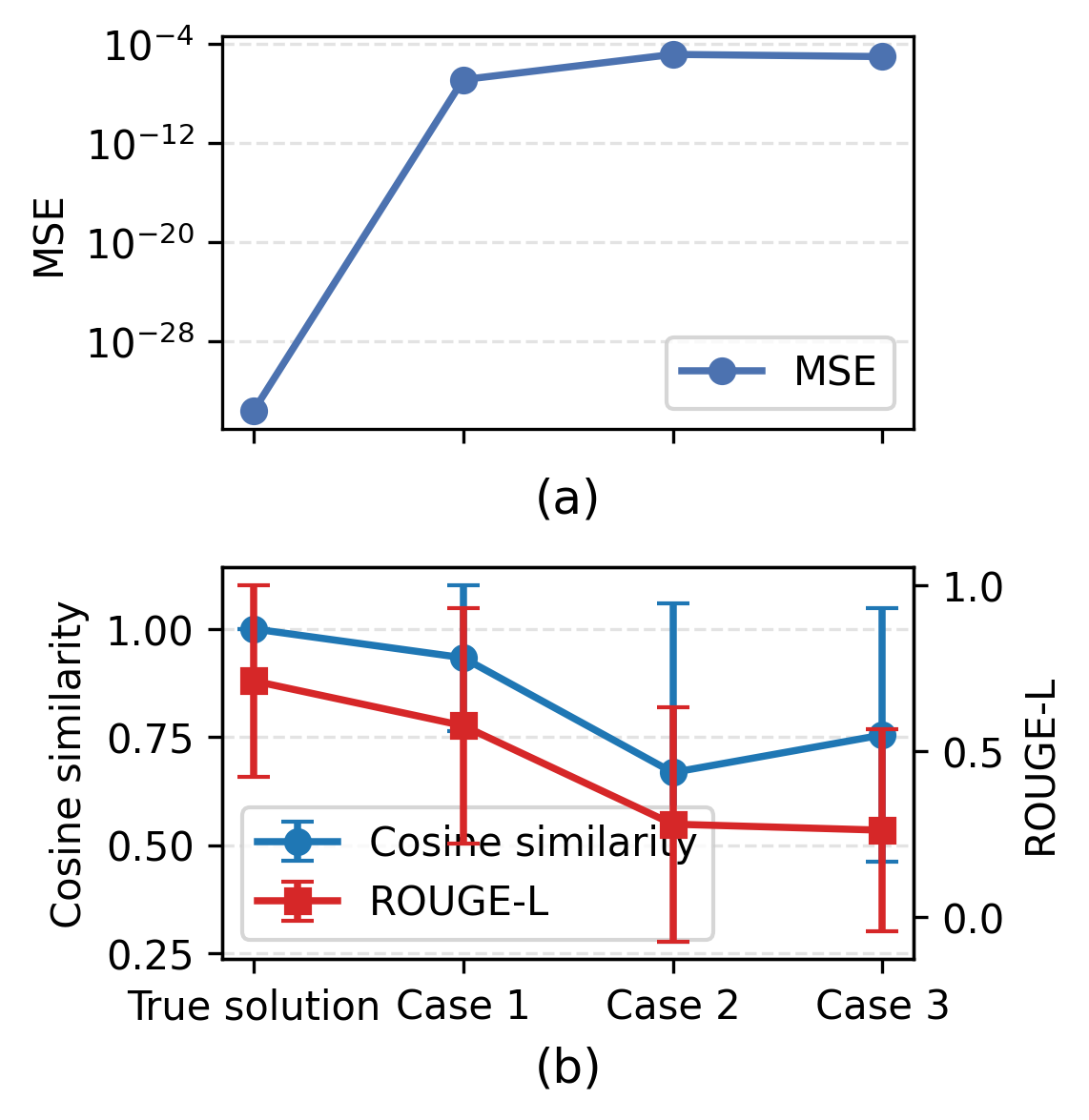}
  \caption{Statistical reconstruction performance on the Sentiment140 dataset. (a) MSE between the HSSP-reconstructed gradients and the GT gradients. (b) Distribution of cosine similarity and Rouge-L scores for training data reconstructed from the sampled candidate gradients.}
  \Description{}
  \label{fig:stat_text}
\end{figure}

\begin{figure}[h]
  \centering
  \includegraphics[width=0.6\linewidth]{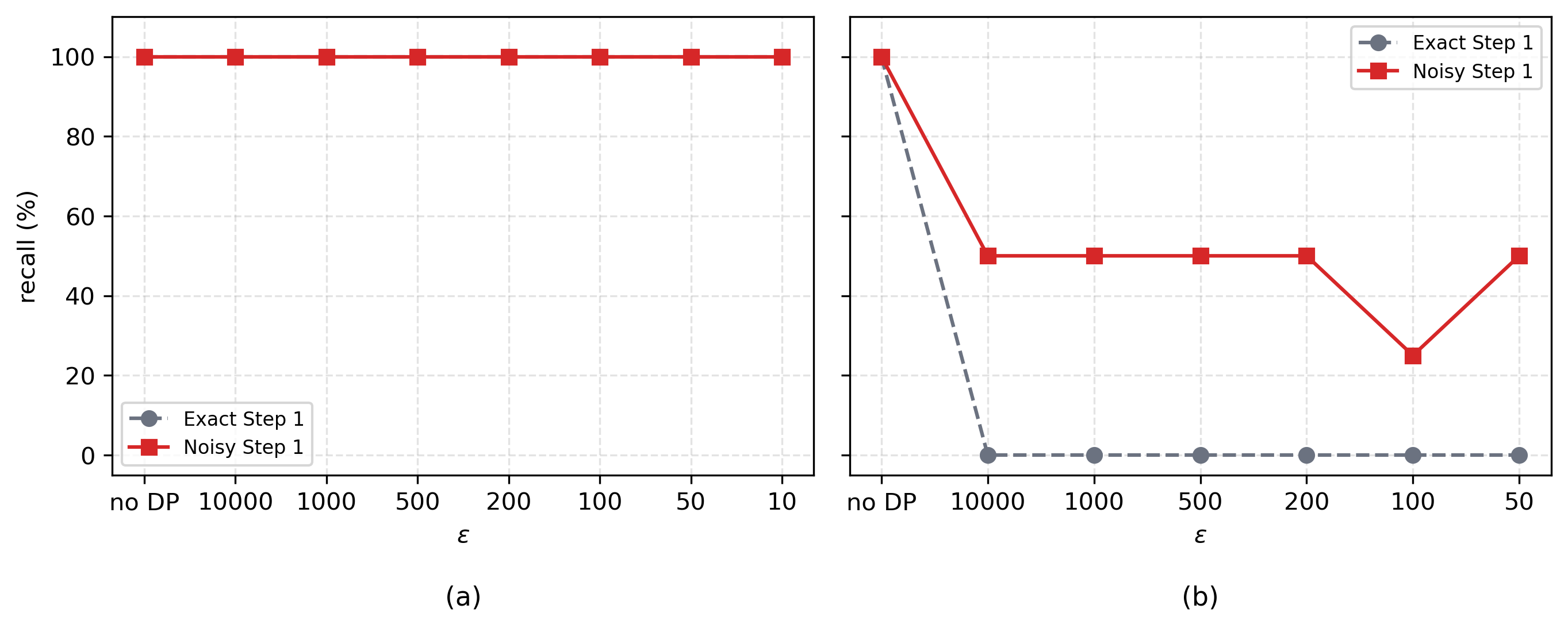}
  \caption{\red{Recall by exact Step~1 and noisy Step~1. (a) LDP and (b) CDP.}}
  \Description{}
  \label{fig:aggregation_dp_noisy}
\end{figure}

\section{Complete Results of Input Reconstruction}\label{app.complete}



\begin{figure*}[!ht]
  \centering
  \includegraphics[width=0.8\linewidth]{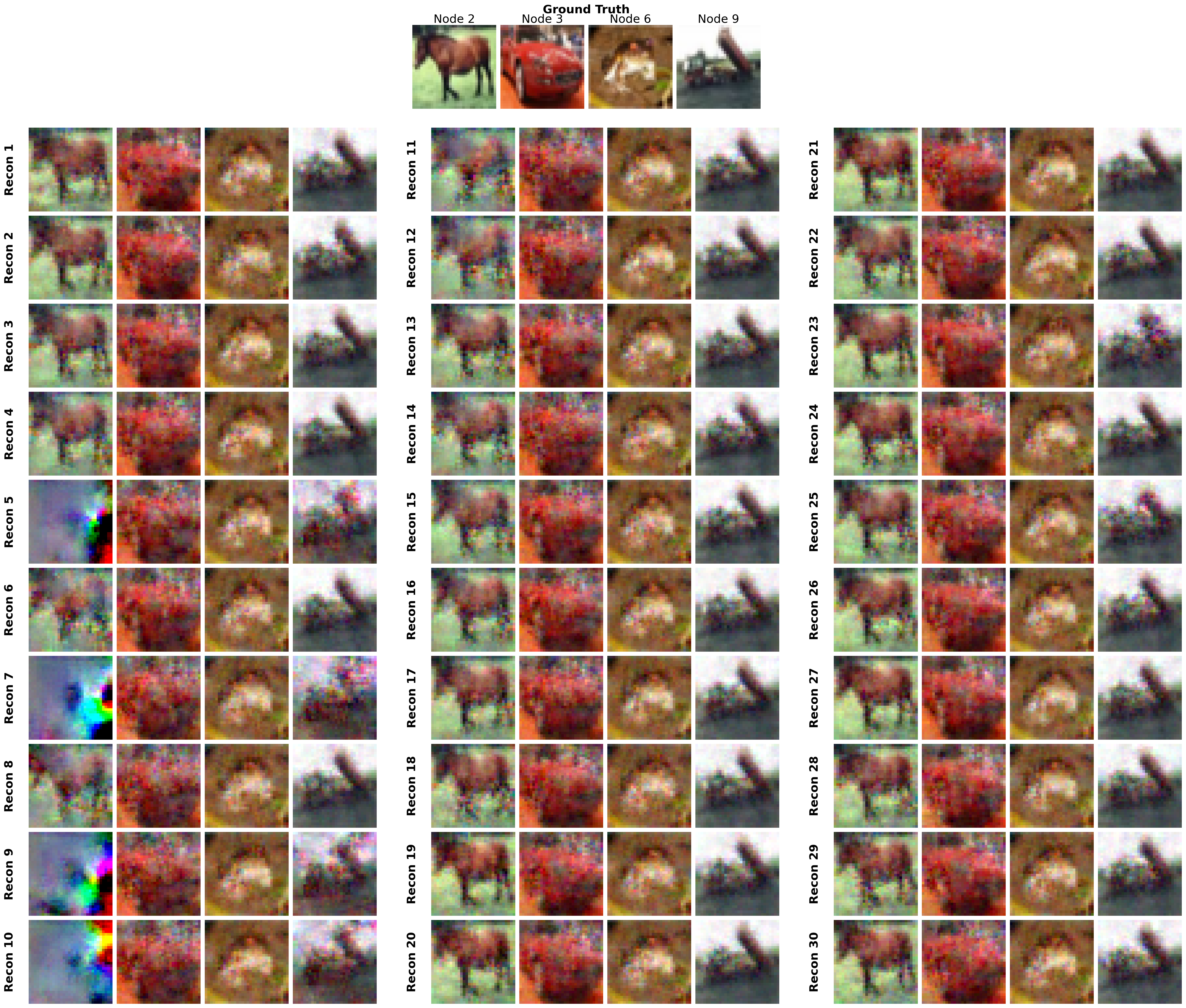}
  \caption{Reconstructed images from 30 candidate weight matrices in Case 1. The 24th entry corresponds to the true solution.}
  \Description{}
  \label{fig:img_complete_1}
\end{figure*}

\begin{figure*}[!ht]
  \centering
  \includegraphics[width=0.8\linewidth]{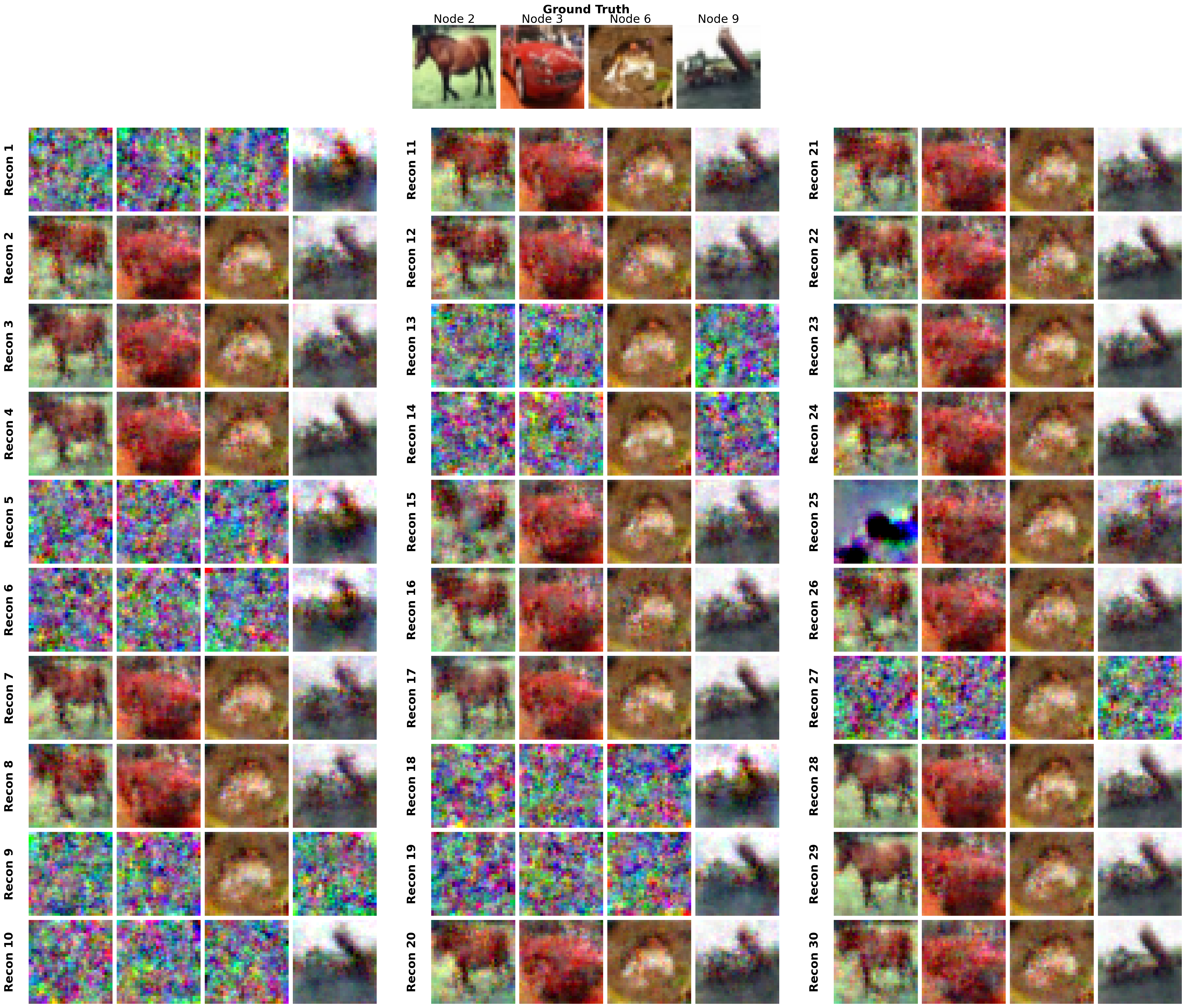}
  \caption{Reconstructed images from 30 candidate weight matrices in Case 2.}
  \Description{}
  \label{fig:img_complete_2}
\end{figure*}

\begin{figure*}[!ht]
  \centering
  \includegraphics[width=0.8\linewidth]{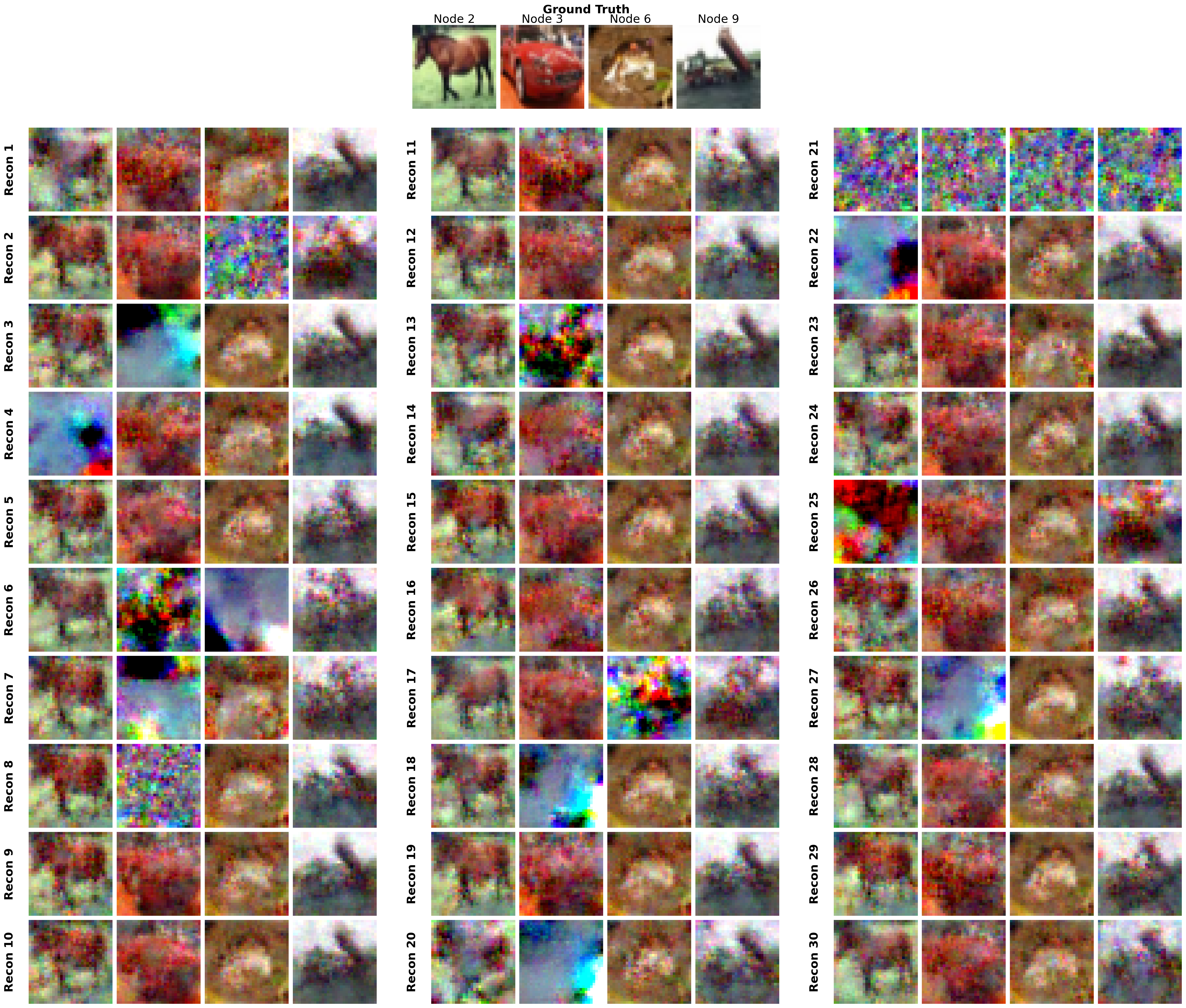}
  \caption{Reconstructed images from 30 candidate weight matrices in Case 3.}
  \Description{}
  \label{fig:img_complete_3}
\end{figure*}

\begin{figure*}[!ht]
  \centering
  \includegraphics[width=0.8\linewidth]{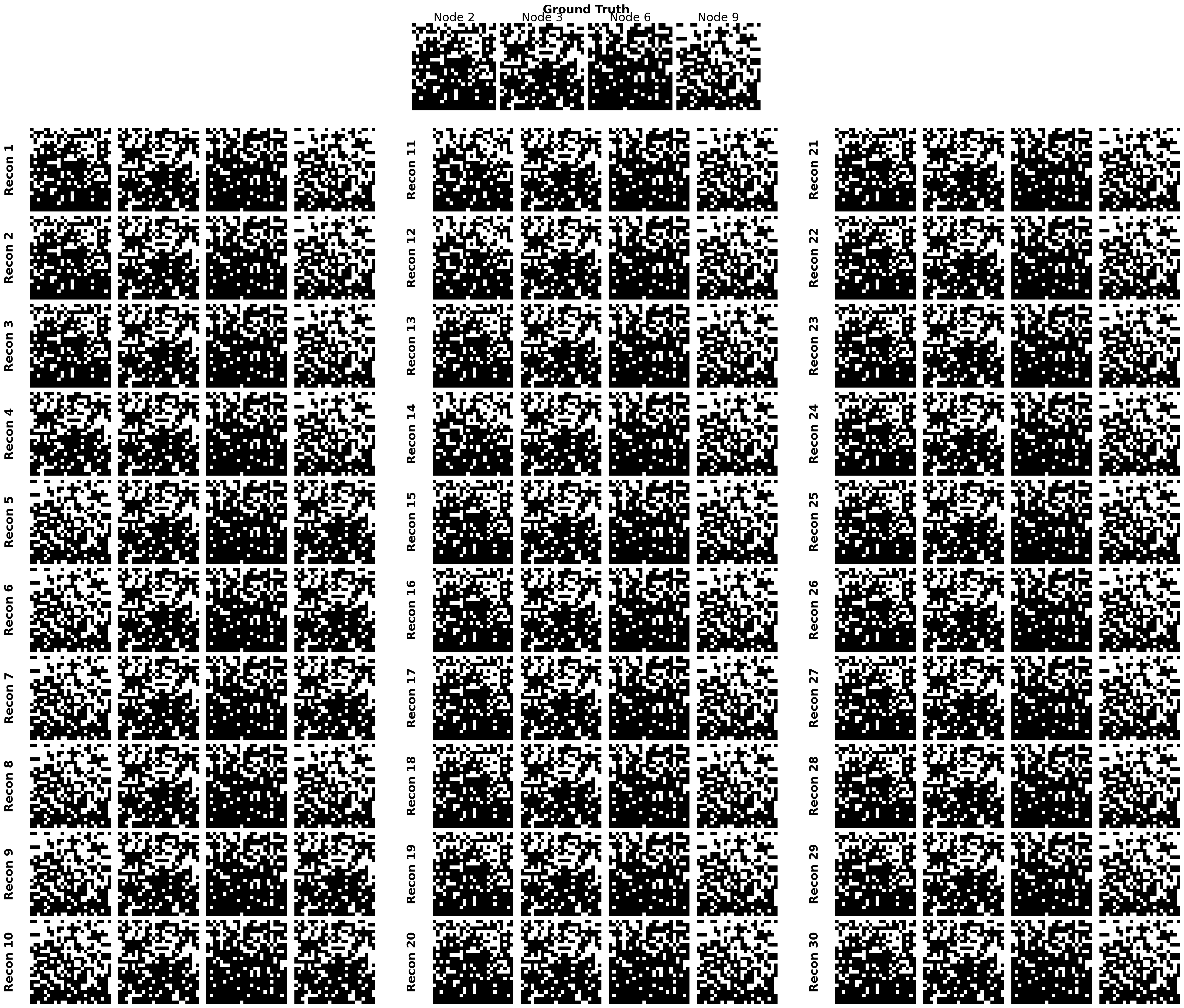}
  \caption{Reconstructed tabular data from 30 candidate weight matrices in Case 1. The 24th entry corresponds to the true solution.}
  \Description{}
  \label{fig:tabular_complete_1}
\end{figure*}

\begin{figure*}[!ht]
  \centering
  \includegraphics[width=0.8\linewidth]{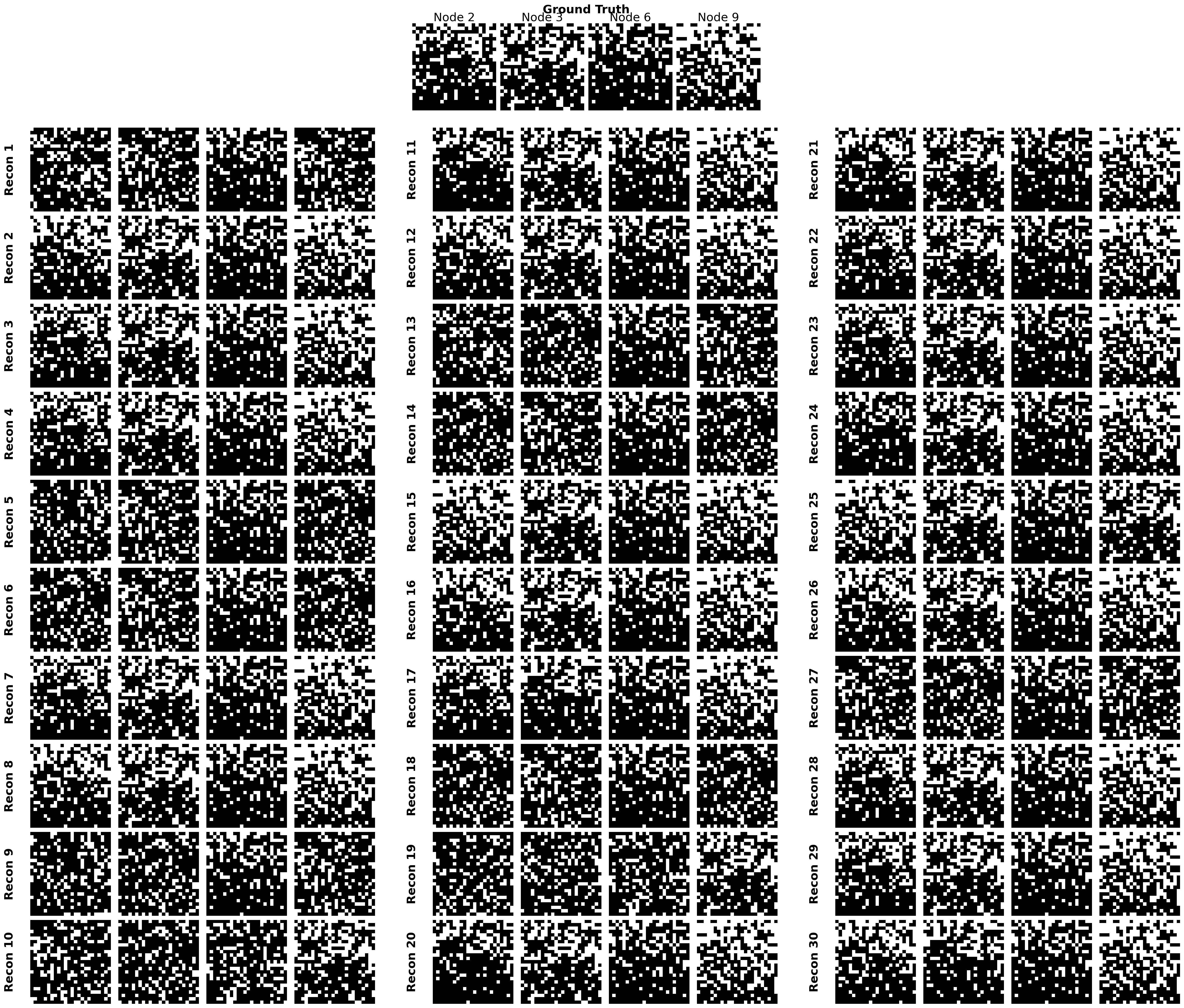}
  \caption{Reconstructed tabular data from 30 candidate weight matrices in Case 2.}
  \Description{}
  \label{fig:tabular_complete_2}
\end{figure*}

\begin{figure*}[!ht]
  \centering
  \includegraphics[width=0.8\linewidth]{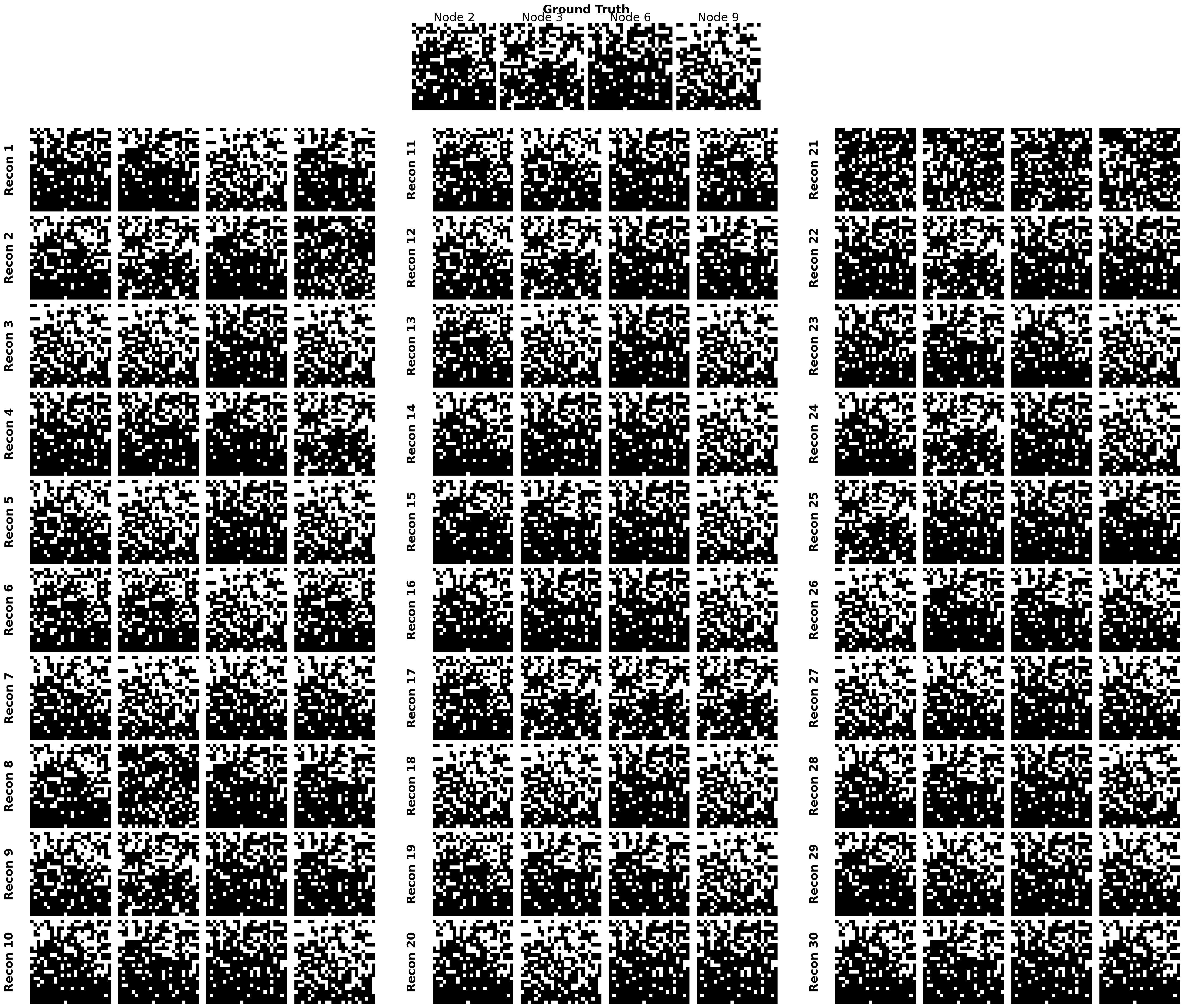}
  \caption{Reconstructed tabular data from 30 candidate weight matrices in Case 3.}
  \Description{}
  \label{fig:tabular_complete_3}
\end{figure*}

\onecolumn

\begin{longtable}{cccp{8cm}cc}
\hline
 & Node & Cosine Similarity & Recovered Text & Rouge-L & BERTScore \\
\hline
\multirow{4}{*}{Ground Truth} & 2 & - & Crunch berries!  im tired. Who wants to do something tomorrow? & - & - \\
 & 3 & - & @LAmale really? They're a band...  good music! & - & - \\
 & 6 & - & Where are all the banana slugs?  & - & - \\
 & 9 & - & @Jack\_thm aww why not?! Heck they do the job! I can't find those anywhere nemore either!!  ur ipod ones hurt!!! & - & - \\
\hline
\multirow{4}{*}{1} & 2 & 0.983 & Crunch berries!                                                       Who can do     Tired? & 0.500 & -0.270 \\
 & 3 & 0.889 & What band are they on?  Lamalame is a really good music band.                                                       & 0.400 & 0.201 \\
 & 6 & 1.000 & Where are all the banana slugs?                                                            & 1.000 & 1.000 \\
 & 9 & 1.000 & mmmm                                                        Why can't anyone find them   They hurt & 0.323 & -0.437 \\
\hline
\multirow{4}{*}{2} & 2 & 0.985 & Crunch berries!                                                       Who wants to do something tomorrow? & 0.889 & -0.152 \\
 & 3 & 0.981 & Really, they're a good band!                                                          Music? & 0.800 & -0.540 \\
 & 6 & 1.000 & Where are all the banana slugs?                                                            & 1.000 & 1.000 \\
 & 9 & 1.000 & mmmm                                                        Why can't anyone find them   They hurt & 0.323 & -0.437 \\
\hline
\multirow{4}{*}{3} & 2 & 0.879 & Who wants to do a cranberry crunch tomorrow?                                                          & 0.556 & 0.280 \\
 & 3 & 1.000 & @Lameah... They're a really good band!                                                        music & 0.625 & -0.475 \\
 & 6 & 1.000 & Where are all the banana slugs?                                                            & 1.000 & 1.000 \\
 & 9 & 1.000 & mmmm                                                          Why can't you find them  & 0.276 & -0.498 \\
\hline
\multirow{4}{*}{4} & 2 & -0.096 & They're a really good band, and they're a good music. They're a really good band, and they're a good music. They're a really good band, and they're a good music. They're a really good band, and they're a good music. They're a really good band, and they're a good music. They're a really good band, and they're a really good music. They're a really good band, and they're a really good music & 0.000 & -0.237 \\
 & 3 & 0.991 & @Lamamele  They're a good band, really good music!                                                      & 0.706 & 0.257 \\
 & 6 & 1.000 & Where are all the banana slugs?                                                            & 1.000 & 1.000 \\
 & 9 & 1.000 & mmmm                                                        Why can't anyone find them   They hurt & 0.323 & -0.437 \\
\hline
\multirow{4}{*}{5} & 2 & 0.750 & mmmm                                                        Why can't anyone find them   They hurt & 0.000 & -0.347 \\
 & 3 & 0.956 & @Lamem  They're a good band!                                                         & 0.571 & 0.255 \\
 & 6 & 1.000 & Where are all the banana slugs?                                                            & 1.000 & 1.000 \\
 & 9 & 0.920 & @Lammahey they're a good band!                                                          & 0.071 & -0.031 \\
\hline
\multirow{4}{*}{6} & 2 & 0.750 & mmmm                                                          Why can't you find them  & 0.000 & -0.382 \\
 & 3 & 0.974 & @Lameah They're a good band!                                                          & 0.571 & 0.370 \\
 & 6 & 1.000 & Where are all the banana slugs?                                                     & 1.000 & 1.000 \\
 & 9 & 0.845 & @Lamael They're a good music band, they really are!                                                       & 0.063 & -0.082 \\
\hline
\multirow{4}{*}{7} & 2 & 0.750 & mmmm                                                        Why can't anyone find them   They hurt & 0.000 & -0.347 \\
 & 3 & 0.956 & @Lamem  They're a good band!                                                         & 0.571 & 0.255 \\
 & 6 & 1.000 & Where are all the banana slugs?                                                      & 1.000 & 1.000 \\
 & 9 & 0.919 & @Lammael They're a good music band!                                                        & 0.069 & -0.002 \\
\hline
\multirow{4}{*}{8} & 2 & 0.750 & mmmm                                                        Why can't anyone find them   They hurt & 0.000 & -0.347 \\
 & 3 & 0.956 & @Lamem  They're a good band!                                                         & 0.571 & 0.255 \\
 & 6 & 1.000 & Where are all the banana slugs?                                                            & 1.000 & 1.000 \\
 & 9 & 0.891 & @Lmaem They're a good one!                                                          & 0.143 & -0.001 \\
\hline
\multirow{4}{*}{9} & 2 & 0.750 & mmmm                                                        Why can't anyone find them   They hurt & 0.000 & -0.347 \\
 & 3 & 0.956 & @Lamem  They're a good band!                                                         & 0.571 & 0.255 \\
 & 6 & 1.000 & Where are all the banana slugs?                                                            & 1.000 & 1.000 \\
 & 9 & 0.917 & @Lammahey they're a good music band!                                                         & 0.069 & -0.048 \\
\hline
\multirow{4}{*}{10} & 2 & 0.750 & mmmm                                                         Why can't you find them   & 0.000 & -0.386 \\
 & 3 & 0.956 & @Lamem  They're a good band!                                                         & 0.571 & 0.255 \\
 & 6 & 1.000 & Where are all the banana slugs?                                                            & 1.000 & 1.000 \\
 & 9 & 0.912 & @lmaem  They're a good music band!                                                        & 0.069 & 0.059 \\
\hline
\multirow{4}{*}{11} & 2 & 0.750 & mmmm                                                        Why can't anyone find them   They hurt & 0.000 & -0.347 \\
 & 3 & 0.974 & @Lameah They're a good band!                                                          & 0.571 & 0.370 \\
 & 6 & 1.000 & Where are all the banana slugs?                                                            & 1.000 & 1.000 \\
 & 9 & 0.814 & @Lammathey're actually a good music band!                                                         & 0.000 & -0.115 \\
\hline
\multirow{4}{*}{12} & 2 & 0.298 & They're a good music band, that's why you find them on Lmaem!                                                     & 0.000 & 0.055 \\
 & 3 & 0.974 & @Lameah They're a good band!                                                          & 0.571 & 0.370 \\
 & 6 & 1.000 & Where are all the banana slugs?                                                            & 1.000 & 1.000 \\
 & 9 & 1.000 & mmmm                                                          Why can't you find them   & 0.276 & -0.498 \\
\hline
\multirow{4}{*}{13} & 2 & 0.627 & âAn âAâ âAâ âAâ âAâ âAâ âAâ âAâ âAâ âAâ âAâ âAâ âAâ âAâ âAâ âAâ âBâ âAâ âBâ âAâ âAâ . & 0.000 & -0.577 \\
 & 3 & 0.974 & @Lameah They're a good band!                                                          & 0.571 & 0.370 \\
 & 6 & 1.000 & Where are all the banana slugs?                                                            & 1.000 & 1.000 \\
 & 9 & 1.000 & mmmm                                                        Why can't anyone find them   They hurt & 0.323 & -0.437 \\
\hline
\multirow{4}{*}{14} & 2 & -0.035 & They aren't just LMA's, they're really good LMA's!! You found them on IMDb, so why aren't they a LMA's? Yes, they are a LMA's, they're really good LMA's!! IMDb is a LMA's, they're a LMA's, they're a good LMA's!!                 & 0.000 & -0.082 \\
 & 3 & 0.974 & @Lameah They're a good band!                                                          & 0.571 & 0.370 \\
 & 6 & 1.000 & Where are all the banana slugs?                                                            & 1.000 & 1.000 \\
 & 9 & 1.000 & mmmm                                                        Why can't anyone find them   They hurt & 0.323 & -0.437 \\
\hline
\multirow{4}{*}{15} & 2 & 0.416 & In addition to a variety of foods, there are a number of other foods that can be used to make up a healthy diet. For example, if you eat a lot of fatty foods, you may want to include a variety of vegetables in your diet. For example, if you eat a lot of fatty foods, you may want to include a variety of fruits and vegetables in your diet. For example, if you eat a lot of fatty foods, you may want to include a variety of fruits and vegetables in your diet. & 0.038 & -0.046 \\
 & 3 & 0.991 & @Lamamele  They're a good band, really good music!                                                       & 0.706 & 0.257 \\
 & 6 & 1.000 & Where are all the banana slugs?                                                            & 1.000 & 1.000 \\
 & 9 & 1.000 & mmmm                                                         Why can't you find them   & 0.276 & -0.499 \\
\hline
\multirow{4}{*}{16} & 2 & 0.743 & What is the definition of a cranberry pie? cranberry pie is a cranberry pie that is made from a combination of cranberries, peaches, pears, and other fruits and vegetables. cranberry pie is a cranberry pie that is made from a combination of cranberries, peaches, pears, and other fruits and vegetables. cranberry pie is a cranberry pie that is made from a combination of cranberries, peaches, pears, and other fruits and vegetables. & 0.000 & -0.172 \\
 & 3 & 0.991 & @Lamamele  They're a good band, really good music!                                                       & 0.706 & 0.257 \\
 & 6 & 1.000 & Where are all the banana slugs?                                                            & 1.000 & 1.000 \\
 & 9 & 1.000 & mmmm                                                         Why can't you find them   & 0.276 & -0.499 \\
\hline
\multirow{4}{*}{17} & 2 & 0.934 & Crunchy berries tomorrow. Who wants to do tomorrow?                                                          & 0.667 & 0.432 \\
 & 3 & 0.991 & @Lamamele  They're a good band, really good music!                                                       & 0.706 & 0.257 \\
 & 6 & 1.000 & Where are all the banana slugs?                                                    & 1.000 & 1.000 \\
 & 9 & 1.000 & mmmm                                                        Why can't anyone find them   They hurt & 0.323 & -0.437 \\
\hline
\multirow{4}{*}{18} & 2 & 0.992 & Crunch berries tomorrow!  Tired. Who wants to do something tomorrow? & 0.900 & 0.709 \\
 & 3 & 0.991 & @Lamamele  They're a good band, really good music!                                                       & 0.706 & 0.257 \\
 & 6 & 1.000 & Where are all the banana slugs?                                                            & 1.000 & 1.000 \\
 & 9 & 1.000 & mmmm                                                        Why can't anyone find them   They hurt & 0.323 & -0.437 \\
\hline
\multirow{4}{*}{19} & 2 & 0.976 & Crumble Berry!                                                       Who can eat cranberries tomorrow?    & 0.353 & -0.292 \\
 & 3 & 0.889 & What band are they on?  Lamalame is a really good music band.                                                      & 0.400 & 0.201 \\
 & 6 & 1.000 & Where are all the banana slugs?                                                            & 1.000 & 1.000 \\
 & 9 & 1.000 & mmmm                                                        Why can't anyone find them   They hurt & 0.323 & -0.437 \\
\hline
\multirow{4}{*}{20} & 2 & 0.985 & Ripe berries!                                                      Who wants to do something tomorrow? & 0.778 & -0.161 \\
 & 3 & 0.981 & Really, they're a good band!                                                          Music? & 0.800 & -0.540 \\
 & 6 & 1.000 & Where are all the banana slugs?                                                            & 1.000 & 1.000 \\
 & 9 & 1.000 & mmmm                                                        Why can't anyone find them   They hurt & 0.323 & -0.437 \\
\hline
\multirow{4}{*}{21} & 2 & 0.995 & Crunch berries!     Tired! Who wants to do something tomorrow?                                                     & 0.947 & 0.654 \\
 & 3 & 1.000 & @Lame... They're a really good band!                                                         music & 0.625 & -0.475 \\
 & 6 & 1.000 & Where are all the banana slugs?                                                            & 1.000 & 1.000 \\
 & 9 & 1.000 & mmmm                                                         Why can't you find them   & 0.276 & -0.499 \\
\hline
\multirow{4}{*}{22} & 2 & 0.956 & Crunch berries tomorrow! im tired.  Who wants to do something tomorrow?                                                      & 0.952 & 0.740 \\
 & 3 & 1.000 & @Lame... They're a really good band!                                                         music & 0.625 & -0.475 \\
 & 6 & 1.000 & Where are all the banana slugs?                                                            & 1.000 & 1.000 \\
 & 9 & 1.000 & mmmm                                                        Why can't anyone find them   They hurt & 0.323 & -0.437 \\
\hline
\multirow{4}{*}{23} & 2 & 0.991 & Crunch berries!  im tired.  Who wants to do something tomorrow?                                                      & 1.000 & 0.900 \\
 & 3 & 1.000 & @Lame... They're a really good band!                                                         music & 0.625 & -0.475 \\
 & 6 & 1.000 & Where are all the banana slugs?                                                            & 1.000 & 1.000 \\
 & 9 & 1.000 & mmmm                                                         Why can't you find them   & 0.276 & -0.499 \\
\hline
\multirow{4}{*}{24} & 2 & 1.000 & Crunch berries!   Tired. Who wants to do something tomorrow?                                                       & 0.947 & 0.812 \\
 & 3 & 1.000 & @Lame... They're a really good band!                                                         music & 0.625 & -0.475 \\
 & 6 & 1.000 & Where are all the banana slugs?                                                      & 1.000 & 1.000 \\
 & 9 & 1.000 & mmmm                                                         Why can't you find them   & 0.276 & -0.499 \\
\hline
\multirow{4}{*}{25} & 2 & 0.994 & Crunch berries                                                         Who wants to do something tomorrow?   & 0.889 & -0.218 \\
 & 3 & 0.981 & Really, they're a good band!                                                          Music? & 0.800 & -0.540 \\
 & 6 & 1.000 & Where are all the banana slugs?                                                      & 1.000 & 1.000 \\
 & 9 & 1.000 & mmmm                                                         Why can't you find them   & 0.276 & -0.499 \\
\hline
\multirow{4}{*}{26} & 2 & 0.995 & Crunch berries!                                  Who wants to do something tomorrow? & 0.889 & 0.019 \\
 & 3 & 0.981 & Really, they're a good band!                                                          Music? & 0.800 & -0.540 \\
 & 6 & 1.000 & Where are all the banana slugs?                                                            & 1.000 & 1.000 \\
 & 9 & 1.000 & mmmm                                                        Why can't anyone find them   They hurt & 0.323 & -0.437 \\
\hline
\multirow{4}{*}{27} & 2 & 0.992 & Crunch berries!                                                       Who wants to do something tomorrow?     & 0.889 & -0.152 \\
 & 3 & 0.981 & Really, they're a good band!                                                          Music? & 0.800 & -0.540 \\
 & 6 & 1.000 & Where are all the banana slugs?                                                            & 1.000 & 1.000 \\
 & 9 & 1.000 & mmmm                                                        Why can't anyone find them   They hurt & 0.323 & -0.437 \\
\hline
\multirow{4}{*}{28} & 2 & 0.985 & Crunch berries!                            Who wants to do something tomorrow? & 0.889 & 0.072 \\
 & 3 & 0.889 & What band are they on?  Lamalame is a really good music band.                                                       & 0.400 & 0.201 \\
 & 6 & 1.000 & Where are all the banana slugs?                                                            & 1.000 & 1.000 \\
 & 9 & 1.000 & mmmm                                                         Why can't you find them   & 0.276 & -0.499 \\
\hline
\multirow{4}{*}{29} & 2 & 0.984 & Crunch berries!                                                       Who      Tired? & 0.429 & -0.338 \\
 & 3 & 0.889 & What band are they on?  Lamalame is a really good music band.                                                       & 0.400 & 0.201 \\
 & 6 & 1.000 & Where are all the banana slugs?                                                            & 1.000 & 1.000 \\
 & 9 & 1.000 & mmmm                                                         Why can't you find them   & 0.276 & -0.499 \\
\hline
\multirow{4}{*}{30} & 2 & 0.981 &                                                   Crunch berries! Who wants to do something tomorrow? & 0.889 & 0.706 \\
 & 3 & 0.889 & What band are they on?  Lamalame is a really good music band.                                                       & 0.400 & 0.201 \\
 & 6 & 1.000 & Where are all the banana slugs?                                                            & 1.000 & 1.000 \\
 & 9 & 1.000 & mmmm                                                        Why can't anyone find them   They hurt & 0.323 & -0.437 \\
\hline
\caption{Reconstructed texts from 30 candidate weight matrices in Case 1. The 24th entry corresponds to the true solution.}
\label{tab:sent140_recon_1}
\end{longtable}

\begin{longtable}{cccp{8cm}cc}
\hline
 & Node & Cosine Similarity & Recovered Text & Rouge-L & BERTScore \\
\hline
\multirow{4}{*}{Ground Truth} & 2 & - & Crunch berries!  im tired. Who wants to do something tomorrow? & - & - \\
 & 3 & - & @LAmale really? They're a band...  good music! & - & - \\
 & 6 & - & Where are all the banana slugs?  & - & - \\
 & 9 & - & @Jack\_thm aww why not?! Heck they do the job! I can't find those anywhere nemore either!!  ur ipod ones hurt!!! & - & - \\
\hline
\multirow{4}{*}{1} & 2 & 0.257 & âAs a rule, if a person is a person who is a person who is a person who is a person who is a person who is a person who is a person who is a person who is a person who is a person who is a person who is a person who is a person who is a person who is a person who is a person who is a person who is a person who is a person who is a person who is able to get to know you.ââ & 0.039 & -0.264 \\
 & 3 & 0.084 & In the past, if you had a family member who had a family member who had a family member who had a family member who had a family member who had a family member who had a family member who had a family member who had a family member who had a family member who had a family member who had a family member who had a family member who had a family member who had a family member who had a family member. & 0.022 & -0.369 \\
 & 6 & 0.023 & âAs a result, if youâre able to get a job, youâll be able to get a job, but if youâre unable to get a job, youâll be unable to get a job. âAs a result, if youâre unable to get a job, youâll be unable to get a job. âAs a result, if youâre unable to get a job, you may need to & 0.000 & -0.238 \\
 & 9 & 0.882 &  Where are all the banana slugs?                                                           & 0.071 & -0.157 \\
\hline
\multirow{4}{*}{2} & 2 & 0.913 & Lamma                                                          Who can    Touch them & 0.133 & -0.452 \\
 & 3 & 0.670 & âAlso known as âAlso known as âAlso known as âAlso known as âAlso known as âAlso known as âAlso known as âAlso known as âAlso known as âAlso known as âAlso known as âAlso known as âAlsoâââââ.â & 0.000 & -0.402 \\
 & 6 & 1.000 & Where are all the banana slugs?                                                            & 1.000 & 1.000 \\
 & 9 & 0.744 & I dont think they can hurt your ipod, i found ones that do the job!!                                                       & 0.216 & 0.174 \\
\hline
\multirow{4}{*}{3} & 2 & 0.963 & Crunch berries!                                                        Who can eat these     & 0.375 & -0.351 \\
 & 3 & 0.670 & âAlso known as âAlso known as âAlso known as âAlso known as âAlso known as âAlso known as âAlso known as âAlso known as âAlso known as âAlso known as âAlso known as âAlso known as âAlsoâââââ.â & 0.000 & -0.402 \\
 & 6 & 1.000 & Where are all the banana slugs?                                                      & 1.000 & 1.000 \\
 & 9 & 0.744 & I dont think they can hurt your ipod, i found ones that do the job!!                                                       & 0.216 & 0.174 \\
\hline
\multirow{4}{*}{4} & 2 & 0.960 & Crunch them!                                                        Who can do      & 0.400 & -0.409 \\
 & 3 & 0.670 & âAlso known as âAlso known as âAlso known as âAlso known as âAlso known as âAlso known as âAlso known as âAlso known as âAlso known as âAlso known as âAlso known as âAlso known as âAlsoâââââ.â & 0.000 & -0.402 \\
 & 6 & 0.889 & Where are all the banana slugs in the Banana Slab? Banana slugs in the Banana Slab. Photo courtesy of the Banana Slab Museum. Introduction to Banana Slugs. Introduction to Banana Slugs. Introduction to Banana Slugs. Introduction to Banana Slugs. Introduction to Banana Slugs. Introduction to Banana Slugs. Introduction to Banana Slugs. Introduction to Banana Slugs. Where is all the world's population concentrated?  & 0.174 & 0.350 \\
 & 9 & 1.000 & mmmm                                                        Why can't anyone find them   They hurt & 0.323 & -0.437 \\
\hline
\multirow{4}{*}{5} & 2 & 0.036 & âAlso known as âAlso known as âAlso known as âAlso known as âAlso known as âAlso known as âAlso known as âAlso known as âAlso known as âAlso known as âAlso known as âAlso known as âAlsoâââââ.â & 0.000 & -0.318 \\
 & 3 & 0.070 & The following is a list of the most common types of cancers found in the United States. The most common type of cancer in the United States is melanoma. The most common type of cancer in the United States is melanoma. The most common type of cancer in the United States is melanoma. The most common type of cancer in the United States is melanoma. The most common type of cancer in the United States is melanoma. The most common type of cancer in the United States is melanoma. & 0.021 & -0.307 \\
 & 6 & 0.014 & âAs a general rule, âAs a general rule, âAs a general rule, âAs a general rule, âAs a general rule, âAs a general rule, âAs a general rule, âAs a general rule, âAs a general rule, âAs a general rule, âAs a general rule, âAâ is ââââ. & 0.000 & -0.320 \\
 & 9 & 0.882 & Where are all the banana slugs?                                                            & 0.071 & -0.157 \\
\hline
\multirow{4}{*}{6} & 2 & -0.033 & âAlso known as âAlso known as âAlso known as âAlso known as âAlso known as âAlso known as âAlso known as âAlso known as âAlso known as âAlso known as âAlso known as âAlso known as âAlsoâââââ.â & 0.000 & -0.318 \\
 & 3 & 0.069 & The following is a list of the most common types of cancers found in the United States. The most common type of cancer in the United States is melanoma. The most common type of cancer in the United States is melanoma. The most common type of cancer in the United States is melanoma. The most common type of cancer in the United States is melanoma. The most common type of cancer in the United States is melanoma. The most common type of cancer in the United States is melanoma. & 0.021 & -0.307 \\
 & 6 & 0.012 & In addition to the above, there are a number of other things you can do to improve the quality of your health care. For example, if you have a history of heart disease, you may need to take a course of antibiotics. For example, if you have a history of heart disease, you may need to take a course of antibiotics. For example, if you have a history of heart disease, you may need to take a course of antibiotics. For example, if you have had a history of heart disease, you may need to take a course & 0.038 & -0.243 \\
 & 9 & 0.882 &  Where are all the banana slugs?                                                           & 0.071 & -0.157 \\
\hline
\multirow{4}{*}{7} & 2 & 0.962 & Smack Berry!                                                      Who can        Fruit & 0.267 & -0.410 \\
 & 3 & 0.670 & âAlso known as âAlso known as âAlso known as âAlso known as âAlso known as âAlso known as âAlso known as âAlso known as âAlso known as âAlso known as âAlso known as âAlso known as âAlsoâââââ.â & 0.000 & -0.402 \\
 & 6 & 1.000 & Where are all the banana slugs?                                                     & 1.000 & 1.000 \\
 & 9 & 0.744 & I dont think they can hurt your ipod, i found ones that do the job!!                                                       & 0.216 & 0.174 \\
\hline
\multirow{4}{*}{8} & 2 & 0.915 & Lamma!                                                          Who can ripe   & 0.143 & -0.431 \\
 & 3 & 0.670 & âAlso known as âAlso known as âAlso known as âAlso known as âAlso known as âAlso known as âAlso known as âAlso known as âAlso known as âAlso known as âAlso known as âAlso known as âAlsoâââââ.â & 0.000 & -0.402 \\
 & 6 & 1.000 & Where are all the banana slugs?                                                            & 1.000 & 1.000 \\
 & 9 & 0.744 & I dont think they can hurt your ipod, i find one that works too!!                                                       & 0.222 & 0.123 \\
\hline
\multirow{4}{*}{9} & 2 & 0.044 & âAs a rule of thumb, if youâre in a âno-holds-barredâ state, youâll need to be in a âno-holds-barredâ state. âAs a rule of thumb, if youâre in a âno-holds-barredâ state, youâll need to be in a âno-holds-barredâ state.âââ. & 0.033 & -0.166 \\
 & 3 & 0.057 & âAs a result, a person with a history of a history of a history of a history of a history of a history of a history of a history of a history of a history of a history of a history of a history of a history of a history of a history of a history of a history of a history of a history of a history of a history of a history of a history of a history of a different kind of history of the same name. & 0.020 & -0.361 \\
 & 6 & 0.008 & â â â â â â â â â â â â â â â â â â â â â â â â â â â â â â â â â â â â â â â â â â  & 0.000 & -0.585 \\
 & 9 & 0.882 &  Where are all the banana slugs?                                                           & 0.071 & -0.157 \\
\hline
\multirow{4}{*}{10} & 2 & -0.033 & âAs a result, âAs a result, âas a result, âas a result, âas a result, âas a result, âas a result, âas a result, âas a result, âas a result, âas a result, âas a result, âas a result, â. â. â & 0.000 & -0.275 \\
 & 3 & 0.055 & âAs a result, a person with a history of a history of a history of a history of a history of a history of a history of a history of a history of a history of a history of a history of a history of a history of a history of a history of a history of a history of a history of a history of a history of a history of a history of a history of a history of a different kind of history of the same name. & 0.020 & -0.361 \\
 & 6 & 0.067 & In addition to a number of other things, if you have a history of diabetes, you may have a history of heart disease, a history of heart failure, a history of heart failure, a history of heart failure, a history of heart disease, a history of heart failure, a history of heart disease, a history of heart failure, a history of heart disease, a history of heart disease, a history of heart disease, a history of heart disease, a history of heart disease, and a history of heart attack. & 0.000 & -0.321 \\
 & 9 & 0.361 & The most common reason for a syringe to have a syringe is because of the syringeâs ability to âsnapâ into the syringe. The syringeâs ability to âsnapâ into the syringe is because of the syringeâs ability to âsnapâ into the syringe. The syringeâs ability to âsna & 0.028 & -0.411 \\
\hline
\multirow{4}{*}{11} & 2 & 0.929 & Slumber band!                                                       Who's     All banana & 0.125 & -0.363 \\
 & 3 & 0.670 & âAlso known as âAlso known as âAlso known as âAlso known as âAlso known as âAlso known as âAlso known as âAlso known as âAlso known as âAlso known as âAlso known as âAlso known as âAlsoâââââ.â & 0.000 & -0.402 \\
 & 6 & 1.000 & Where are all the banana slugs?                                                      & 1.000 & 1.000 \\
 & 9 & 0.971 & mmmm                                                        Why can't they find good music?    & 0.267 & -0.426 \\
\hline
\multirow{4}{*}{12} & 2 & 0.916 & Lamma!                                                        Who can      Fruit & 0.143 & -0.399 \\
 & 3 & 0.670 & âAlso known as âAlso known as âAlso known as âAlso known as âAlso known as âAlso known as âAlso known as âAlso known as âAlso known as âAlso known as âAlso known as âAlso known as âAlsoâââââ.â & 0.000 & -0.402 \\
 & 6 & 1.000 & Where are all the banana slugs?                                                            & 1.000 & 1.000 \\
 & 9 & 0.744 & I dont think they can hurt your ipod, i found ones that do the job!!                                                       & 0.216 & 0.174 \\
\hline
\multirow{4}{*}{13} & 2 & 0.365 & âââââââââââââââââââââââââââââââââââââââââââââââââââââââââââ. This is a list of. & 0.000 & -0.574 \\
 & 3 & 0.053 & âAs a result, a person with a history of a history of a history of a history of a history of a history of a history of a history of a history of a history of a history of a history of a history of a history of a history of a history of a history of a history of a history of a history of a history of a history of a history of a history of a history of a different kind of history of the same name. & 0.020 & -0.361 \\
 & 6 & 0.008 & â â â â â â â â â â â â â â â â â â â â â â â â â â â â â â â â â â â â â â â â â â  & 0.000 & -0.585 \\
 & 9 & 0.882 &  Where are all the banana slugs?                                                           & 0.071 & -0.157 \\
\hline
\multirow{4}{*}{14} & 2 & -0.034 & âAlso known as âAlso known as âAlso known as âAlso known as âAlso known as âAlso known as âAlso known as âAlso known as âAlso known as âAlso known as âAlso known as âAlso known as âAlsoâââââ.â & 0.000 & -0.318 \\
 & 3 & 0.063 & âAnother way to look at it is to look at it as if it were a teddy bear. âAnother way to look at it is to look at it as if it were a teddy bear. âAnother way to look at it is to look at it as if it were a teddy bear. âAnother way to look at it is to look at it as if it were a teddy bear. âAnother way to look at it & 0.023 & -0.251 \\
 & 6 & -0.007 & â â â â â â â â â â â â â â â â â â â â â â â â â â â â â â â â â â â â â â â â â â  & 0.000 & -0.585 \\
 & 9 & 0.882 &  Where are all the banana slugs?                                                           & 0.071 & -0.157 \\
\hline
\multirow{4}{*}{15} & 2 & 0.834 &                                                          Who can find them     Pain & 0.133 & 0.018 \\
 & 3 & 0.670 & âAlso known as âAlso known as âAlso known as âAlso known as âAlso known as âAlso known as âAlso known as âAlso known as âAlso known as âAlso known as âAlso known as âAlso known as âAlsoâââââ.â & 0.000 & -0.402 \\
 & 6 & 1.000 & Where are all the banana slugs?                                                      & 1.000 & 1.000 \\
 & 9 & 0.744 & I dont think they can hurt your ipod, i found ones that do the job!!                                                       & 0.216 & 0.174 \\
\hline
\multirow{4}{*}{16} & 2 & 0.932 & Mmm!                                                        Who can ripe     Fruit & 0.133 & -0.386 \\
 & 3 & 0.670 & âAlso known as âAlso known as âAlso known as âAlso known as âAlso known as âAlso known as âAlso known as âAlso known as âAlso known as âAlso known as âAlso known as âAlso known as âAlsoâââââ.â & 0.000 & -0.402 \\
 & 6 & 0.889 & Where are all the banana slugs in the Banana Slab? Banana slugs in the Banana Slab. Photo courtesy of the Banana Slab Museum. Introduction to Banana Slugs. Introduction to Banana Slugs. Introduction to Banana Slugs. Introduction to Banana Slugs. Introduction to Banana Slugs. Introduction to Banana Slugs. Introduction to Banana Slugs. Introduction to Banana Slugs. Where is all the world's population concentrated?  & 0.174 & 0.350 \\
 & 9 & 1.000 & mmmm                                                        Why can't anyone find them   They hurt & 0.323 & -0.437 \\
\hline
\multirow{4}{*}{17} & 2 & 0.946 & Crunch berries!                                                       Who can do      I & 0.500 & -0.325 \\
 & 3 & 0.874 & What band is L.A. Lame?  Good music, really?                                                        & 0.353 & 0.313 \\
 & 6 & 1.000 & Where are all the banana slugs?                                                      & 1.000 & 1.000 \\
 & 9 & 1.000 & mmmm                                                        Why can't anyone find them   They hurt & 0.323 & -0.437 \\
\hline
\multirow{4}{*}{18} & 2 & -0.004 & âAlso known as âAlso known as âAlso known as âAlso known as âAlso known as âAlso known as âAlso known as âAlso known as âAlso known as âAlso known as âAlso known as âAlso known as âAlsoâââââ.â & 0.000 & -0.318 \\
 & 3 & 0.059 & âAs a result, a person with a history of a history of a history of a history of a history of a history of a history of a history of a history of a history of a history of a history of a history of a history of a history of a history of a history of a history of a history of a history of a history of a history of a history of a history of a history of a different kind of history of the same name. & 0.020 & -0.361 \\
 & 6 & -0.002 & â â â â â â â â â â â â â â â â â â â â â â â â â â â â â â â â â â â â â â â â â â  & 0.000 & -0.585 \\
 & 9 & 0.882 &  Where are all the banana slugs?                                                           & 0.071 & -0.157 \\
\hline
\multirow{4}{*}{19} & 2 & 0.106 & âAs a rule of thumb, if youâre able to get a job, youâll be able to get a job in a few weeks.â. âAs a rule of thumb, if youâre able to get a job, youâll be able to get a job in a few weeks.â. âAs a rule of thumb, youâll be able to get a job in a few weeks.â. â & 0.026 & -0.084 \\
 & 3 & -0.331 & The most common reason for a syringe to have a syringe is because of the syringeâs ability to âsnapâ into the syringe. The syringeâs ability to âsnapâ into the syringe is because of the syringeâs ability to âsnapâ into the syringe. The syringeâs ability to âsna & 0.034 & -0.364 \\
 & 6 & 0.054 & âAs a rule of thumb, if you have a ânoâ or ânoâ in your ânoâ or ânoâ in your ânoâ or ânoâ in your ânoâ or ânoâ in your ânoâ or ânoâ in your ânoâ or ânoâ in your ânoâ category & 0.000 & -0.507 \\
 & 9 & 0.060 & â â â â â â â â â â â â â â â â â â â â â â â â â â â â â â â â â â â â â â â â â â  & 0.000 & -0.655 \\
\hline
\multirow{4}{*}{20} & 2 & 0.935 & Slumber Band!                                                       Who's      Fruit & 0.133 & -0.406 \\
 & 3 & 0.670 & âAlso known as âAlso known as âAlso known as âAlso known as âAlso known as âAlso known as âAlso known as âAlso known as âAlso known as âAlso known as âAlso known as âAlso known as âAlsoâââââ.â & 0.000 & -0.402 \\
 & 6 & 1.000 & Where are all the banana slugs?                                                            & 1.000 & 1.000 \\
 & 9 & 0.992 &                                                          Why can't anyone find them   They hurt & 0.333 & 0.054 \\
\hline
\multirow{4}{*}{21} & 2 & 0.929 & Lamma!                                                          Who can ripe bananas? & 0.133 & -0.349 \\
 & 3 & 0.670 & âAlso known as âAlso known as âAlso known as âAlso known as âAlso known as âAlso known as âAlso known as âAlso known as âAlso known as âAlso known as âAlso known as âAlso known as âAlsoâââââ.â & 0.000 & -0.402 \\
 & 6 & 1.000 & Where are all the banana slugs?                                                      & 1.000 & 1.000 \\
 & 9 & 0.959 & mmmm                                                           Why can't you find ones & 0.345 & -0.532 \\
\hline
\multirow{4}{*}{22} & 2 & 0.962 & Lips!                                                       Who can      Crisp berries? & 0.133 & -0.336 \\
 & 3 & 0.670 & âAlso known as âAlso known as âAlso known as âAlso known as âAlso known as âAlso known as âAlso known as âAlso known as âAlso known as âAlso known as âAlso known as âAlso known as âAlsoâââââ.â & 0.000 & -0.402 \\
 & 6 & 0.889 & Where are all the banana slugs in the Banana Slab? Banana slugs in the Banana Slab. Photo courtesy of the Banana Slab Museum. Introduction to Banana Slugs. Introduction to Banana Slugs. Introduction to Banana Slugs. Introduction to Banana Slugs. Introduction to Banana Slugs. Introduction to Banana Slugs. Introduction to Banana Slugs. Introduction to Banana Slugs. Where is all the world's population concentrated?  & 0.174 & 0.350 \\
 & 9 & 1.000 & mmmm                                                        Why can't anyone find them   They hurt & 0.323 & -0.437 \\
\hline
\multirow{4}{*}{23} & 2 & 0.976 & Crunch berries!                                                   Who can do something tomorrow? & 0.706 & -0.205 \\
 & 3 & 0.820 & What's the name of Lebanese band? Really good music.  Lebanese band Lamala is a very good music group, with a diverse lineup of artists, including:                                        . & 0.235 & -0.295 \\
 & 6 & 1.000 & Where are all the banana slugs?                                                            & 1.000 & 1.000 \\
 & 9 & 1.000 & mmmm                                                        Why can't anyone find them   They hurt & 0.323 & -0.437 \\
\hline
\multirow{4}{*}{24} & 2 & 0.887 & Where are all the slime band?                                                          All bananas? & 0.000 & -0.377 \\
 & 3 & 0.670 & âAlso known as âAlso known as âAlso known as âAlso known as âAlso known as âAlso known as âAlso known as âAlso known as âAlso known as âAlso known as âAlso known as âAlso known as âAlsoâââââ.â & 0.000 & -0.402 \\
 & 6 & 0.986 & Where are all the banana slugs?                                                            & 1.000 & 1.000 \\
 & 9 & 1.000 & mmmm                                                        Why can't anyone find them   They hurt & 0.323 & -0.437 \\
\hline
\multirow{4}{*}{25} & 2 & 0.750 & mmmm                                                        Why can't anyone find them   They hurt & 0.000 & -0.347 \\
 & 3 & 0.956 & @Lamem  They're a good band!                                                         & 0.571 & 0.255 \\
 & 6 & 1.000 & Where are all the banana slugs?                                                            & 1.000 & 1.000 \\
 & 9 & 0.919 & @Lamamel They're a good band!                                                          & 0.071 & -0.004 \\
\hline
\multirow{4}{*}{26} & 2 & 0.897 & Lamm!                                                        Who can     They  hurt & 0.133 & -0.405 \\
 & 3 & 0.670 & âAlso known as âAlso known as âAlso known as âAlso known as âAlso known as âAlso known as âAlso known as âAlso known as âAlso known as âAlso known as âAlso known as âAlso known as âAlsoâââââ.â & 0.000 & -0.402 \\
 & 6 & 1.000 & Where are all the banana slugs?                                                      & 1.000 & 1.000 \\
 & 9 & 0.992 &                                                          Why can't anyone find them   They hurt & 0.333 & 0.054 \\
\hline
\multirow{4}{*}{27} & 2 & 0.139 & In the past, if you had a family member who had a family member who had a family member who had a family member who had a family member who had a family member who had a family member who had a family member who had a family member who had a family member who had a family member who had a family member who had a family member who had a family member who had a family member who had a family member. & 0.021 & -0.257 \\
 & 3 & -0.028 & âAs a result, âAs a result, âas a result, âas a result, âas a result, âas a result, âas a result, âas a result, âas a result, âas a result, âas a result, âas a result, âas a result, â. â. â & 0.043 & -0.382 \\
 & 6 & 0.976 &  Where are all the banana slugs?                                                           & 1.000 & 1.000 \\
 & 9 & 0.038 & â â â â â â â â â â â â â â â â â â â â â â â â â â â â â â â â â â â â â â â â â â  & 0.000 & -0.655 \\
\hline
\multirow{4}{*}{28} & 2 & 0.991 & Crunch berries!  im tired.  Who wants to do something tomorrow?                                                      & 1.000 & 0.900 \\
 & 3 & 1.000 & @Lame... They're a really good band!                                                         music & 0.625 & -0.475 \\
 & 6 & 1.000 & Where are all the banana slugs?                                                            & 1.000 & 1.000 \\
 & 9 & 1.000 & mmmm                                                        Why can't anyone find them   They hurt & 0.323 & -0.437 \\
\hline
\multirow{4}{*}{29} & 2 & 0.992 & Crunch berries!  tired tomorrow.  Who wants to do something tomorrow? & 0.900 & 0.668 \\
 & 3 & 0.991 & @Lamamele  They're a good band, really good music!                                                       & 0.706 & 0.257 \\
 & 6 & 1.000 & Where are all the banana slugs?                                                            & 1.000 & 1.000 \\
 & 9 & 1.000 & mmmm                                                         Why can't you find them   & 0.276 & -0.499 \\
\hline
\multirow{4}{*}{30} & 2 & 0.938 & Lamma!                                                        Who can eat ripe berries tomorrow? & 0.235 & -0.261 \\
 & 3 & 0.874 & What band is L.A. Lame?  Good music, really?                                                        & 0.353 & 0.313 \\
 & 6 & 1.000 & Where are all the banana slugs?                                                    & 1.000 & 1.000 \\
 & 9 & 1.000 & mmmm                                                        Why can't anyone find them   They hurt & 0.323 & -0.437 \\
\hline
\caption{Reconstructed texts from 30 candidate weight matrices in Case 2.}
\label{tab:sent140_recon_2}
\end{longtable}

\begin{longtable}{cccp{8cm}cc}
\hline
 & Node & Cosine Similarity & Recovered Text & Rouge-L & BERTScore \\
\hline
\multirow{4}{*}{Ground Truth} & 2 & - & Crunch berries!  im tired. Who wants to do something tomorrow? & - & - \\
 & 3 & - & @LAmale really? They're a band...  good music! & - & - \\
 & 6 & - & Where are all the banana slugs?  & - & - \\
 & 9 & - & @Jack\_thm aww why not?! Heck they do the job! I can't find those anywhere nemore either!!  ur ipod ones hurt!!! & - & - \\
\hline
\multirow{4}{*}{1} & 2 & 0.606 & YES! I can do a pair of ipod touch earpieces                                                      dont & 0.095 & -0.367 \\
 & 3 & 0.901 & Lamma!                                                          Slim bananas are all over the band & 0.125 & -0.744 \\
 & 6 & 0.907 & Lamma                                                         Where are all the good ones?    & 0.615 & -0.642 \\
 & 9 & 0.946 & Who can do the job!                                                              & 0.222 & 0.110 \\
\hline
\multirow{4}{*}{2} & 2 & 0.945 & Ramma                                                           Slap berries  Who & 0.286 & -0.468 \\
 & 3 & 0.989 &  They're a really good band!                                                          African music & 0.625 & -0.518 \\
 & 6 & 0.011 & â â â â â â â â â â â â â â â â â â â â â â â â â â â â â â â â â â â â â â â â â â  & 0.000 & -0.585 \\
 & 9 & 0.889 &                                                     Where are all those banana slugs?        & 0.071 & -0.154 \\
\hline
\multirow{4}{*}{3} & 2 & 0.897 & Mmm  They're awsome!                                                        Who can find  & 0.118 & -0.379 \\
 & 3 & 0.659 & mmmm                                                           Why can't you find ones & 0.000 & -0.713 \\
 & 6 & 1.000 & Where are all the banana slugs?                                                     & 1.000 & 1.000 \\
 & 9 & 0.973 & mmmm  They do the job! Who can find one? i'm nothin'                                                      & 0.412 & 0.102 \\
\hline
\multirow{4}{*}{4} & 2 & 0.869 & Where are all the banana slugs?                                                            & 0.000 & 0.056 \\
 & 3 & 0.670 & âAlso known as âAlso known as âAlso known as âAlso known as âAlso known as âAlso known as âAlso known as âAlso known as âAlso known as âAlso known as âAlso known as âAlso known as âAlsoâââââ.â & 0.000 & -0.402 \\
 & 6 & 0.941 & Amlama                                                 Where are all those slime bananas? They're good!        & 0.500 & -0.449 \\
 & 9 & 0.717 & Where are all the banana slugs?                                                            & 0.071 & -0.157 \\
\hline
\multirow{4}{*}{5} & 2 & 0.916 & Lamm!                                                        Who can do      Fruit & 0.267 & -0.377 \\
 & 3 & 0.386 & Who can find the ipod crans ones??                                                          & 0.000 & -0.101 \\
 & 6 & 0.955 & Where are all the banana slugs?                                                            & 1.000 & 1.000 \\
 & 9 & 0.744 & I dont think they can hurt your ipod, i found ones that do the job!!                                                       & 0.216 & 0.174 \\
\hline
\multirow{4}{*}{6} & 2 & 0.934 & Crunchy                                                        Who can do       & 0.286 & -0.488 \\
 & 3 & 0.787 & Crunch!                                                      Who can do those        & 0.000 & -0.718 \\
 & 6 & 0.759 & mmmm                                                        Why can't anyone find them   They hurt & 0.000 & -0.702 \\
 & 9 & 0.922 & Crunchy berries!                                                       Who can do      & 0.074 & -0.520 \\
\hline
\multirow{4}{*}{7} & 2 & 0.915 & Lamma!                                                          Who can    Touch & 0.143 & -0.400 \\
 & 3 & 0.802 & mmmm  Who can do them?                                                           iPod & 0.000 & -0.629 \\
 & 6 & 0.888 & Lamma                                                        Who can    Push bananas? & 0.182 & -0.730 \\
 & 9 & 0.940 & Mmm!                                                       Who can find        & 0.154 & -0.601 \\
\hline
\multirow{4}{*}{8} & 2 & 0.049 & âAs a result, a person with a history of a history of a history of a history of a history of a history of a history of a history of a history of a history of a history of a history of a history of a history of a history of a history of a history of a history of a history of a history of a history of a history of a history of a history of a history of a different kind of history of the same name. & 0.000 & -0.216 \\
 & 3 & 0.905 & Lamma!                                                        Who can rock music tomorrow?    Fruit & 0.133 & -0.585 \\
 & 6 & 0.946 &                                                     Where are all those banana slugs?        & 0.833 & 0.890 \\
 & 9 & 0.441 & I'm not sure if I'll be able to find a job that will work for me. I'm not sure if I'll be able to find a job that will work for me. I'm not sure if I'll be able to find a job that will work for me. I'm not sure if I'll be able to find a job that will work for me. I'm not sure if I'll be able to find a job that will work for me. I'm not sure if & 0.068 & -0.307 \\
\hline
\multirow{4}{*}{9} & 2 & 0.922 & Lamma!                                                         Who can rock music music?    & 0.125 & -0.359 \\
 & 3 & 0.980 & @Lameem  They're a really good band!                                                       music & 0.625 & -0.496 \\
 & 6 & 0.952 & Where are all those banana slugs?                                                            & 0.833 & 0.890 \\
 & 9 & -0.234 & In the past, a number of people have been diagnosed with cancer, but only a few of them have been diagnosed with cancer. In the past, a number of people have been diagnosed with cancer, but only a few of them have been diagnosed with cancer. In the past, a number of people have been diagnosed with cancer, but only a few have been diagnosed with cancer. In the past, a number of people have been diagnosed with cancer, but only a few have been diagnosed with cancer. & 0.018 & -0.286 \\
\hline
\multirow{4}{*}{10} & 2 & 0.922 & Mmm!                                                        Who can ripe     Fruit & 0.133 & -0.386 \\
 & 3 & 0.874 & What band is L.A. Lame?  Good music, really?                                                        & 0.353 & 0.313 \\
 & 6 & 1.000 & Where are all the banana slugs?                                                      & 1.000 & 1.000 \\
 & 9 & 0.959 & mmmm                                                           Why can't you find ones & 0.345 & -0.532 \\
\hline
\multirow{4}{*}{11} & 2 & 0.984 & Crunch berries                                                   Who wants to do something tomorrow? & 0.889 & -0.178 \\
 & 3 & 0.882 & Lamm!                                                        Who can do      Fruit & 0.000 & -0.673 \\
 & 6 & 1.000 & Where are all the banana slugs?                                                            & 1.000 & 1.000 \\
 & 9 & 0.929 & Crunchy                                                        Who can     Push these  & 0.074 & -0.601 \\
\hline
\multirow{4}{*}{12} & 2 & 0.939 &                                                          Who can do      & 0.308 & 0.014 \\
 & 3 & 0.948 & At L.A.M., are they really a band?  Good music!                                                       & 0.526 & 0.361 \\
 & 6 & 0.933 & Where are they all?                                                              & 0.600 & 0.555 \\
 & 9 & 0.344 & Where is the band Salamande?  All good music, Salamande.  All good music, Salamande.  All good music, Salamande.  All good music, Salamande.  All good music, Salamande.  All good music, Salamande.  All good music, Salamande.  All good music, Salamande.  All good music, Salamande.  All good music, Salamande.  All good music, Salamande.  All good music.   & 0.027 & -0.326 \\
\hline
\multirow{4}{*}{13} & 2 & 0.889 & Mmm                                                         Who can find them     Pain & 0.125 & -0.388 \\
 & 3 & 0.763 & mmmm  it hurts! Who can do them?                                                          & 0.000 & 0.005 \\
 & 6 & 1.000 & Where are all the banana slugs?                                                    & 1.000 & 1.000 \\
 & 9 & 0.965 & mmmm  Who can find one!                                                           Pain & 0.214 & -0.467 \\
\hline
\multirow{4}{*}{14} & 2 & 0.745 & yea                                                          Who can do    & 0.286 & -0.531 \\
 & 3 & 0.573 & Where are all the slugs in Banana Lagoon? All the slugs in Banana Lagoon. Photo courtesy of the University of California, Davis. In the past few decades, the slug population in Banana Lagoon has been steadily growing, with a large number of slugs on the surface and a small number of slugs on the ground. The slug population is dominated by slugs in the family Procyonidae, which means good food. & 0.051 & -0.123 \\
 & 6 & 0.671 & The term slug is used to refer to all slugs in the genus Sluga. The term slug is used to refer to all slugs in the genus Sluga. The term slug is used to refer to all slugs in the genus Sluga. The term slug is used to refer to all slugs in the genus Sluga. The term slug is used to refer to all slugs in the genus Sluga. & 0.079 & -0.122 \\
 & 9 & 0.744 & I dont think they can hurt your ipod, i found ones that do the job!!                                                       & 0.216 & 0.174 \\
\hline
\multirow{4}{*}{15} & 2 & 0.914 & Lame banana band!                                                          Where are all the slugs & 0.000 & -0.355 \\
 & 3 & -0.294 & The most common reason for a syringe to have a syringe is because the syringe is used to hold the syringe in place. The syringe is used to hold the syringe in place while the syringe is used to hold the syringe in place. The syringe is used to hold the syringe in place while the syringe is used to hold the syringe in place while the syringe is used to & 0.025 & -0.353 \\
 & 6 & 0.965 & Where are all the banana slugs?                                                            & 1.000 & 1.000 \\
 & 9 & 0.744 & I dont think they can hurt your ipod, i found ones that do the job!!                                                       & 0.216 & 0.174 \\
\hline
\multirow{4}{*}{16} & 2 & 0.842 & Where is all the banana slugs?                                                            & 0.000 & 0.062 \\
 & 3 & 0.573 & Where are all the slugs in Banana Lagoon? All the slugs in Banana Lagoon. Photo courtesy of the University of California, Davis. In the past few decades, the slug population in Banana Lagoon has been steadily growing, with a large number of slugs on the surface and a small number of slugs on the ground. The slug population is dominated by slugs in the family Procyonidae, which means good food. & 0.051 & -0.123 \\
 & 6 & 0.671 & The term slug is used to refer to all slugs in the genus Sluga. The term slug is used to refer to all slugs in the genus Sluga. The term slug is used to refer to all slugs in the genus Sluga. The term slug is used to refer to all slugs in the genus Sluga. The term slug is used to refer to all slugs in the genus Sluga. & 0.079 & -0.122 \\
 & 9 & 0.744 & I dont think they can hurt your ipod, i found ones that do the job!!                                                       & 0.216 & 0.174 \\
\hline
\multirow{4}{*}{17} & 2 & 0.778 & @Lamem  They're a good band!                                                         & 0.000 & 0.075 \\
 & 3 & 0.969 & @LamaelThey're a really good band!                                                         & 0.429 & 0.278 \\
 & 6 & 0.094 & The following is a list of the most common causes of death in the United States from AIDS: AIDS, AIDS-related illnesses, AIDS-related deaths, and AIDS-related deaths. AIDS-related deaths are caused by AIDS, AIDS-related illnesses, and AIDS-related deaths. AIDS-related deaths are caused by AIDS, AIDS-related illnesses, and AIDS-related deaths. AIDS-related deaths are caused by AIDS, AIDS-related illnesses, and AIDS-related deaths. AIDS-related deaths are caused by AIDS-related & 0.024 & -0.312 \\
 & 9 & 0.915 & @Lamaem They're a good band!                                                          & 0.071 & 0.002 \\
\hline
\multirow{4}{*}{18} & 2 & 0.903 & mmmm  Who can                                                           Who can find them & 0.118 & -0.417 \\
 & 3 & 0.398 & I dont think they can hurt your ipod, i found ones that do the job!!                                                       & 0.087 & 0.016 \\
 & 6 & 1.000 & Where are all the banana slugs?                                                            & 1.000 & 1.000 \\
 & 9 & 0.958 & mmmm  Who can do them! They hurt                                                           & 0.138 & 0.033 \\
\hline
\multirow{4}{*}{19} & 2 & 0.954 & Cherry Bananas.  Tired?  Who wants to do a snack tomorrow?  All the slugs are in the cranberry bushes ?                                            & 0.429 & 0.368 \\
 & 3 & 0.795 & YES they're a really good band, i found lmaemoemoemmmmmmmmmmmmmmmmmmmmmmmmmmmmmmmmmmmmmmmmmmmmmmmmmm music & 0.526 & -0.157 \\
 & 6 & 0.943 & Lamama                                                        Where are all those good banana slugs?  & 0.714 & -0.409 \\
 & 9 & 0.617 & I can't find a syringe for my ipod. I can't find a syringe for my ipod. I can't find a syringe for my ipod. I can't find a syringe for my ipod. I can't find a syringe for my ipod. I can't find a syringe for my ipod. I can't find a syringe for my ipod. Work & 0.116 & -0.141 \\
\hline
\multirow{4}{*}{20} & 2 & 0.606 & yea  i can't find one!                                                         hurt & 0.000 & -0.307 \\
 & 3 & 0.390 & The following is a list of all the species of slugs that are found in the slugs of the slugs of the slugs of the slugs of the slugs of the slugs of the slugs of the slugs of the slugs of the slugs of the slugs of the slugs of the slugs of the slugs of the slugs of the slugs of the slugs of the slugs. & 0.026 & -0.452 \\
 & 6 & 0.671 & The term slug is used to refer to all slugs in the genus Sluga. The term slug is used to refer to all slugs in the genus Sluga. The term slug is used to refer to all slugs in the genus Sluga. The term slug is used to refer to all slugs in the genus Sluga. The term slug is used to refer to all slugs in the genus Sluga. & 0.079 & -0.122 \\
 & 9 & 0.744 & I dont think they can hurt your ipod, i found ones that do the job!!                                                       & 0.216 & 0.174 \\
\hline
\multirow{4}{*}{21} & 2 & -0.024 & âAlso known as âAlso known as âAlso known as âAlso known as âAlso known as âAlso known as âAlso known as âAlso known as âAlso known as âAlso known as âAlso known as âAlso known as âAlsoâââââ.â & 0.000 & -0.318 \\
 & 3 & 0.079 & In the past, if you had a family member who had a family member who had a family member who had a family member who had a family member who had a family member who had a family member who had a family member who had a family member who had a family member who had a family member who had a family member who had a family member who had a family member who had a family member who had a family member. & 0.022 & -0.369 \\
 & 6 & 0.178 & âAs a rule of thumb, if you have a sex license, you must have a sex license. âAs a rule of thumb, you must have a sex license. âAs a rule of thumb, you must have a sex license. âAs a rule of thumb, you must have a sex license. âAs a rule of thumb, you must have a sex license. âAs a rule of thumb, & 0.000 & -0.211 \\
 & 9 & 0.052 & âAssuming that you have a healthy weight and a healthy body mass, you should be able to maintain a healthy weight.â. âAssuming that you have a healthy weight and a healthy body mass, you should be able to maintain a healthy weight.â. âAssuming that you have a healthy body mass, you should be able to maintain a healthy weight.â. âAssuming that you have a healthy body mass, you should be able to maintain  & 0.000 & -0.384 \\
\hline
\multirow{4}{*}{22} & 2 & 0.256 & I can't find a painkiller for my ipod touch. I can't find a painkiller for my ipod touch. I can't find a painkiller for my ipod touch. I can't find a painkiller for my ipod touch. I can't find a painkiller for my ipod touch. I can't find a painkiller for my ipod touch. I can't find a painkiller for my ipod touch. & 0.000 & -0.123 \\
 & 3 & 0.404 & The following is a list of the most common names used in the United States for the term âAtlanticâ: âAtlanticâ: âAtlanticâ: âAtlanticâ: âAtlanticâ: âAtlanticâ: âAtlanticâ: âAtlanticâ: âAtlanticâ: âAtlanticâ: âAtlantic. & 0.054 & -0.368 \\
 & 6 & 0.986 & Where are all the banana slugs?                                                            & 1.000 & 1.000 \\
 & 9 & 0.027 & âAlso known as âAlso known as âAlso known as âAlso known as âAlso known as âAlso known as âAlso known as âAlso known as âAlso known as âAlso known as âAlso known as âAlso known as âAlsoâââââ.â & 0.000 & -0.468 \\
\hline
\multirow{4}{*}{23} & 2 & 0.634 & I can do mine tomorrow!                                                             someone & 0.250 & -0.288 \\
 & 3 & 0.919 & Where's all the bananas?   Good band!                                                         & 0.133 & -0.001 \\
 & 6 & 0.911 & Grubs!                                                        Who can find the ripe banana slugs? & 0.429 & -0.427 \\
 & 9 & 0.969 &                                                          Why can't one find them    & 0.286 & -0.069 \\
\hline
\multirow{4}{*}{24} & 2 & 0.854 &                                                          Who can mmmm music & 0.143 & -0.006 \\
 & 3 & 0.945 & @Lameem they're a really good band! ...                                                        & 0.533 & 0.253 \\
 & 6 & 0.671 & The term slug is used to refer to all slugs in the genus Sluga. The term slug is used to refer to all slugs in the genus Sluga. The term slug is used to refer to all slugs in the genus Sluga. The term slug is used to refer to all slugs in the genus Sluga. The term slug is used to refer to all slugs in the genus Sluga. & 0.079 & -0.122 \\
 & 9 & 0.744 & I dont think they can hurt your ipod, i found ones that do the job!!                                                       & 0.216 & 0.174 \\
\hline
\multirow{4}{*}{25} & 2 & 0.779 & A good band!                                                          Lamalame are really good music! & 0.000 & -0.399 \\
 & 3 & 0.915 & Where are all those slime bananas?                                                            & 0.000 & -0.099 \\
 & 6 & 0.958 & Where are all those banana slugs?                                                            & 0.833 & 0.890 \\
 & 9 & 0.879 & Lamma                                                        Where are all those good banana slugs?   & 0.067 & -0.506 \\
\hline
\multirow{4}{*}{26} & 2 & 0.866 & Mmmm!                                                        Who can find good music    They hurt & 0.111 & -0.293 \\
 & 3 & 0.925 & A  Lamma                                                       Where are all those slugs in the banana band? & 0.211 & -0.629 \\
 & 6 & 0.920 & Lamma                                                        Where are all those good slugs?   & 0.615 & -0.513 \\
 & 9 & 0.937 & Lamma                                                         Oh, why can't you find them   & 0.267 & -0.493 \\
\hline
\multirow{4}{*}{27} & 2 & 0.913 & Lamma                                                          Who can    Touch them & 0.133 & -0.452 \\
 & 3 & 0.760 & mmmm                                                         Why can't you find them   & 0.000 & -0.690 \\
 & 6 & 1.000 & Where are all the banana slugs?                                                            & 1.000 & 1.000 \\
 & 9 & 0.926 & Mmm                                                       Who can find        & 0.154 & -0.663 \\
\hline
\multirow{4}{*}{28} & 2 & 0.876 &                                                          Whoever can do     pain & 0.143 & 0.063 \\
 & 3 & -0.123 & I'm not sure if I'll be able to find a job that will work for me. I'm not sure if I'll be able to find a job that will work for me. I'm not sure if I'll be able to find a job that will work for me. I'm not sure if I'll be able to find a job that will work for me. I'm not sure if I'll be able to find a job that will work for me. I'll be able to & 0.019 & -0.286 \\
 & 6 & 0.821 & Where are all the good-looking slugs in Banana Lagoon? Photo courtesy of the University of California, Los Angeles. All the slugs in Banana Lagoon are in the family Lamnidae, but there are only a small number of slugs in the genus Lamnidae. In recent years, there has been a dramatic resurgence of the slug population in the world's tropical forests. & 0.176 & 0.293 \\
 & 9 & 0.660 & I can't find a ipod ear neem ipod ear neem ipod ear neem ipod ear neem ipod ear neem ipod ear neem ipod ear neem ipod ear neem ipod ear neem ipod ear neem ipod ear neem ipod ear pain nee & 0.156 & -0.329 \\
\hline
\multirow{4}{*}{29} & 2 & 0.916 & Lame bananas!                                                       Where are all the slugs?    & 0.000 & -0.334 \\
 & 3 & 0.906 & Lamma                                                        Who can find good      & 0.154 & -0.752 \\
 & 6 & 1.000 & Where are all the banana slugs?                                                            & 1.000 & 1.000 \\
 & 9 & 0.939 & Lamma                                                          Oh, why can't you find them? & 0.267 & -0.463 \\
\hline
\multirow{4}{*}{30} & 2 & 0.795 & Where is all the banana slug?                                                            & 0.000 & 0.029 \\
 & 3 & 0.882 & Where is the band Lalama?  All good music!                                                         & 0.375 & 0.263 \\
 & 6 & 0.959 & Banana slugs!  Where are they all?                                                         & 0.500 & 0.483 \\
 & 9 & 0.775 & crunch!  Who can do something i dont wanna do tomorrow!                                                        & 0.125 & -0.031 \\
\hline
\caption{Reconstructed texts from 30 candidate weight matrices in Case 3.}
\label{tab:sent140_recon_3}
\end{longtable}

\end{document}